\documentclass[10pt]{article}
\usepackage{amsmath}
\usepackage{amsthm}
\usepackage{amsfonts}
\usepackage{multirow} 
\usepackage{graphicx}
\usepackage{setspace}
\usepackage{amssymb}
\usepackage{subfiles}
\usepackage{titlefoot}
\usepackage{mathtools}
\usepackage{appendix}
\usepackage{mathrsfs}
\usepackage[textsize=footnotesize]{todonotes}
\RequirePackage[colorlinks,citecolor=blue,urlcolor=blue,linkcolor=blue]{hyperref}
\usepackage[top=2.5cm,bottom=2.5cm, outer=2.5cm, inner=2.5cm]{geometry}
\usepackage[normalem]{ulem}

\RequirePackage[numbers]{natbib}
\usepackage{enumitem} %

\usepackage{booktabs}

\usepackage{color}

\usepackage{ulem}

\makeatletter
\def\namedlabel#1#2{\begingroup
#2%
\def\@currentlabel{#2}%
\phantomsection\label{#1}\endgroup
}
\makeatother

\newcommand\email[1]{\rm\href{mailto:#1}{ \nolinkurl{#1}}}

\renewcommand{\theequation}{\arabic{section}.\arabic{equation}}

\newtheorem{theorem}{Theorem}[section]
\newtheorem{definition}[theorem]{Definition}
\newtheorem{lemma}[theorem]{Lemma}
\newtheorem{corollary}[theorem]{Corollary}
\newtheorem{proposition}[theorem]{Proposition}
\newtheorem{remark}[theorem]{Remark}

\newtheorem{example}[theorem]{Example}

\usepackage{dsfont}
\usepackage{etoolbox}
\usepackage{listings}
\usepackage{dirtytalk}
\usepackage{xcolor}
\usepackage{comment}

\newcommand{\mbE}{\mathbb{E}}
\newcommand{\mbR}{\mathbb{R}}
\newcommand{\mbP}{\mathbb{P}}
\newcommand{\mbU}{\mathbb{U}}
\newcommand{\mbN}{\mathbb{N}}

\newcommand{\mcF}{\mathcal{F}}
\newcommand{\mcP}{\mathcal{P}}
\newcommand{\dd}{\mathrm{d}}
\newcommand{\eps}{\varepsilon}
\newcommand{\cW}{\mathcal{W}}
\newcommand{\cI}{\mathcal{I}}
\newcommand{\cJ}{\mathcal{J}}
\newcommand{\cL}{\mathcal{L}}
\newcommand{\wI}{I}
\newcommand{\wJ}{J}

\newcommand{\Pn}{\widetilde{P}^{(n)}}
\newcommand{\Vn}{\widetilde{V}^{(n)}}
\newcommand{\Wn}{W^{(n)}}

\newcommand{\Vncont}{\widetilde{V}^{*,(n)}}
\newcommand{\Vcont}{\widetilde{V}^*}
\newcommand{\Pcont}{\widetilde{P}^*}
\newcommand{\Pncont}{\widetilde{P}^{*,(n)}}

\newcommand{\fBM}{B^H}
\newcommand{\fBMn}{B^{H,(n)}}

\newcommand{\Vnz}{\widetilde{V}^{(n),0}}
\newcommand{\Vno}{\widetilde{V}^{(n),1}}

\newcommand{\Vnoo}{\widetilde{V}^{(n),1,1}}
\newcommand{\Vnot}{\widetilde{V}^{(n),1,2}}

\newcommand{\Vnoi}{\widetilde{V}^{(n),1,i}}

\newtheorem{assumption}[theorem]{Assumption}

\def\benumerate{\begin{enumerate}}\def\eenumerate{\end{enumerate}}
\def\bitemize{\begin{itemize}}\def\eitemize{\end{itemize}}

\def\beqlb{\begin{eqnarray}}\def\eeqlb{\end{eqnarray}}
\def\beqnn{\begin{eqnarray*}}\def\eeqnn{\end{eqnarray*}}
\def\ar{\!\!\!&}

\makeatletter
\def\blfootnote{\gdef\@thefnmark{}\@footnotetext}
\makeatother

\begin{document}

\title{\bf \Large Microstructural Foundation of Rough Log-Normal Volatility Models}

\author{Paul P.~Hager\thanks{Department of Statistics and Operations Research, University of Vienna, Kolingasse 14-16, 1090 Wien;
paul.peter.hager@univie.ac.at} \, Ulrich Horst\thanks{Department of Mathematics and School of Business and Economics, Humboldt University Berlin, Unter den Linden 6, 10099 Berlin; horst@math.hu-berlin.de; Financial support from DFG CRC/TRR 388 “Rough Analysis, Stochastic Dynamics and Related Fields” - Project ID 516748464  - Project B02 is gratefully acknowledged.} \, Thomas Wagenhofer\thanks{Department of Mathematics, Technical University Berlin, Strasse des 17. Juni 136, 10587 Berlin; wagenhof@math.tu-berlin.de. Financial support by IRTG 2544 ”Stochastic Analysis in Interaction” - DFG project-ID 410208580. Additional support from DFG CRC/TRR 388 “Rough Analysis, Stochastic Dynamics and Related Fields”, Project B02 is gratefully acknowledged.} \, Wei Xu\thanks{School of Mathematics and Statistics, Beijing Institute of Technology, No. 5, South Street, Zhongguancun, Haidian District, 100081 Beijing; email: xuwei.math@gmail.com}}

\maketitle
\unmarkedfntext{We thank Antoine Jacquier for insightful comments.}

 \begin{abstract}
   We establish a microstructural foundation of the rough Bergomi model. Specifically, we consider a sequence of order driven financial market models where orders to buy or sell an asset arrive according to a Poisson process and have a long lasting impact on volatility. Using a recently established $C$-tightness result for c\`adl\`ag processes we establish the weak convergence of the price-volatility process to a log-normal rough volatility model. Our weak convergence result is  accompanied by weak error rates that employ a recently established Clark-Ocone formula for Poisson processes and turn our microstructure model into viable alternative to classical simulation schemes. The weak error rates strongly hinge on Poisson arrival dynamics and are novel to the rough microstructure literature.   
 	
 	\medskip
 	\smallskip
 	
 	\noindent \textbf{\textit{MSC2020 subject classifications.}} 
 	 Primary 60G55, 60G22, 60L90;  secondary 60F05, 
 	 
 	  \smallskip
 	 
 \end{abstract}

\section{Introduction and Problem Formulation}

Over the past decade, rough volatility has emerged as one of the most significant advances in quantitative finance, fundamentally reshaping how researchers and practitioners model uncertainty in financial markets. The intellectual roots of this development can be traced to back at least to Gatheral \cite{Gatheral2006}. Although written before the formal rough volatility literature emerged, he documented robust stylized facts of implied volatility surfaces—most notably the steep short-maturity implied volatility skew—that classical Markovian stochastic volatility models struggled to explain without excessive parameter complexity. These persistent empirical discrepancies suggested that the standard assumption of continuous semimartingale volatility dynamics might be fundamentally misspecified.

In their seminal paper \cite{GatheralJaissonRosenbaum2018}, the authors demonstrated that log-volatility behaves statistically like a fractional Brownian motion with Hurst parameter around $0.1$—far below the Brownian benchmark. This implies that volatility paths are significantly rougher than previously assumed. Rather than exhibiting long memory in the classical sense, volatility follows universal scaling laws consistent with a rough fractional structure. This empirical finding overturned decades of modeling intuition and provided a new parsimonious explanation for observed market regularities.

Building on this evidence, Bayer et al.~\cite{BayerFrizGatheral2016} incorporated rough dynamics into option pricing models, introducing the \emph{rough Bergomi model}. They showed that rough fractional stochastic volatility naturally generates the steep short-maturity skew documented in equity markets, resolving a longstanding tension between data and theory. Importantly, this work connected financial modeling with modern stochastic analysis, demonstrating that rough processes could be treated rigorously and simulated effectively despite their non-Markovian nature.

Further advancing tractability and practical applicability, El Euch and Rosenbaum \cite{ElEuchFukasawaRosenbaum2018, ElEuchRosenbaum2019b} extended the classical Heston framework to a rough setting and analyzed a model of the form  
\begin{align*}
     dS_t &= S_t\sqrt{V_t}dW_t , \cr
 \ar\ar \cr
 V_t &= \int_0^t  \frac{(t-s)^{\alpha-1}}{\Gamma(\alpha)} \cdot b\big(\theta - V_s \big) ds
 +\int_0^t \frac{(t-s)^{\alpha-1}}{\Gamma(\alpha)}\cdot \gamma \sqrt{V_s}dB_s.
\end{align*}
 The rough Heston model is described in terms of an affine Volterra process \cite{AbiJaberLarssonPulido2019} and admits a semi-explicit representation of the characteristic function in terms of a fractional Riccati equation, enabling efficient Fourier-based pricing methods.
 Other popular rough volatility models include the quadratic rough Heston model \cite{RosenbaumZhang2021}, the mixed rough Bergomi model \cite{LacombeMuguruzaStone2021} %
 and the rough SABR model \cite{FukasawaGatheral2022}. 

\subsection{Microstructural Foundation of Rough Bergomi Model}

Microstructural foundations of Rough Volatility were first provided by Jaisson and Rosenbaum \cite{JaissonRosenbaum2016}. Modeling order flow through nearly unstable Hawkes processes, they showed that the macroscopic limit of high-frequency trading activity naturally produces rough volatility behavior. 

A Hawkes process is a random point process $\hat N$ that models self-exciting arrivals of random events. In such settings, events arrive at random points in time $\tau_1 < \tau_2 < \tau_3 < \cdots$ according to an  \textit{intensity} process $\{V(t):t\geq 0\}$ that is usually of the form
\begin{equation*}
 V_t \coloneqq \mu(t)+ \sum_{0<\tau_i<t} \phi(t-\tau_i) = \mu(t) + \int_{(0,t)} \phi(t-s) \hat{N}(ds), \quad t \geq 0,
\end{equation*}
  where the \textit{immigration density} $\mu(\cdot)$ captures the arrival of exogenous events and the \textit{kernel} $\phi(\cdot)$ captures the self-exciting impact of past events on the arrivals of future events.

 Hawkes processes with light-tailed kernels have been used in Horst and Xu \cite{HorstXu2022} to establish scaling limits for a class of continuous-time stochastic volatility models with self-exciting jump dynamics. Nearly unstable Hawkes processes with light-tailed kernels were first analyzed by Jaisson and Rosenbaum \cite{JaissonRosenbaum2015}. They proved the weak convergence of the rescaled intensity to a Feller diffusion and the convergence of the rescaled point process to the integrated diffusion; their result was extended to multi-variate processes in \cite{Xu2021}. Under a heavy-tailed condition on the kernel, the same authors \cite{JaissonRosenbaum2016} later considered the weak convergence of the rescaled point process to the integral of a rough fractional diffusion; a corresponding convergence result for the characteristic function of the rescaled Hawkes process has been considered in \cite{ElEuchRosenbaum2019b}. Analogous scaling limits were established in the multivariate case in \cite{ElEuchFukasawaRosenbaum2018,RosenbaumTomas2021}.

Much more refined results were recently established by Horst et al.~in \cite{HorstXuZhang1} for rough Heston-type models and in \cite{HorstXuZhang2} for a class of rough path-dependent volatility models. In both cases, weak convergence results for the volatility processes, rather than the integrated volatility were established. As already argued in \cite{JaissonRosenbaum2016} the main challenge is to prove the C-tightness\footnote{A sequence of processes is C-tight if it is tight and all accumulation points are continuous.} of the sequence of rescaled volatility processes; the increments of the volatility process do not satisfy the standard moment condition that is usually required to establish the C-tightness of a sequence of stochastic processes. To overcome this challenge, Horst et al. \cite{HorstXuZhang1} introduced a novel technique to verify the C-tightness of a sequence c\`adl\`ag processes based on the classical Kolmogorov-Chentsov tightness criterion for continuous processes. This tightness result will also be key to our analysis.

This paper is the first essay to establish a scaling limit for rough log-normal volatility models, in particular for the rough Bergomi model.
Following \cite{BayerFrizGatheral2016}, the model is specified under the physical measure $\mathbb{P}$ by
\begin{equation*}
\frac{\dd S_t}{S_t} = \sqrt{v_t}\,\dd W_t , \qquad
\frac{v_t}{v_0} = e^{2\sigma_v B^{H}_t}, \qquad t\ge0,
\end{equation*}
where $\fBM$ is a standard \textit{fractional Brownian motion} with Hurst parameter $H \in (0,1/2)$, correlated with the Brownian motion $W$, and where $\sigma_v >0$ is the \emph{vol-of-vol} parameter.
More precisely, $\fBM = (B^{H}_t)_{t\in \mathbb{R}}$ is a mean-zero Gaussian process with covariance
\[
\mathbb{E}[\fBM_t \fBM_s]
= \frac{1}{2}\bigl(|t|^{2H} + |s|^{2H} - |t-s|^{2H}\bigr),
\qquad s,t\in\mathbb{R}.
\]
For a Brownian motion $B$ correlated with $W$ via
$
\frac{\dd}{\dd t}\langle B, W \rangle_t = \rho \in [-1,1],
$
$\fBM$ admits the Mandelbrot--van Ness representation
\begin{equation}\label{Def:fBM}
\fBM_t
= \frac{1}{c_H}\Bigg[
\int_0^t (t-s)^{H-\frac12}\,\dd B_s
+ \int_{-\infty}^0\Bigl((t-s)^{H-\frac12}-(-s)^{H-\frac12}\Bigr)\,\dd B_s
\Bigg],
\end{equation}
where $c_H:=(\Gamma(2H+1)\sin(\pi H))^{-\frac12}\Gamma(H+\frac12)$ is a normalization constant.
Note that the commonly presented forward-variance formulation of the model is obtained after a change of measure 
$
\dd B_t = \dd B^{\mathbb{Q}}_t + \lambda_t\,\dd t,
$
in which case the forward variance curve takes the form
\[
\xi_0(t)
= \mathbb{E}^{\mathbb{P}}\!\left[v_t \mid \mathcal{F}_0\right]
\exp\!\left(\frac{2\sigma_v}{c_H}\int_0^t (t-s)^{H-\frac12}\lambda_s\,\dd s\right).
\]

Beyond the well-accepted empirical observation that the realized log-volatility is approximately Gaussian \cite{andersen2001distribution,barndorff2002econometrics, GatheralJaissonRosenbaum2018}, rough log-normal models 
capture empirical smile dynamics, in particular the asymptotics of the \emph{at the money skew} and the \emph{skew-stickiness ratio} much better than the rough Heston model. This was already one of the main motivations in Bergomi’s original work \cite{bergomi2004smile, bergomi2009smile}, where the distinction between Type I and Type II smile dynamics led to the proposal of a lognormal factor model. The same structural argument carries over, mutatis mutandis, to the rough setting and motivates the preference for rough Bergomi-type models over rough Heston models \cite{BayerFrizGatheral2016,bourgey2025refined,friz2025computing,fukasawa2026skew,romer2022emperical}.

Establishing a microstuctural foundation for log-normal models essentially reduces to establishing a microstuctural foundation for fractional Brownian motion. This allows us to consider a much simpler mathematical framework. Unlike the existing literature \cite{ElEuchFukasawaRosenbaum2018,HorstXu2022,HorstXuZhang1, HorstXuZhang2,JaissonRosenbaum2016, dandapani2021quadratic} on microstructural foundations of rough Heston models where the impact of an individual order on volatility is felt indirectly through a series of child orders, we assume that  orders arrive according to a Poisson process $N$ on $\mathbb R$ with unit rate and arrival times $\cdots<\tau_{-1}<\tau_0\leq0<\tau_1<\tau_2<\cdots$ and that each order has a direct, yet slowly decaying impact on future volatility. 

Poisson order arrivals dynamics are commonly used in  (financial) economics to model order flow and have long been studied in the on market microstructure literature. It was Garman \cite{Garman} who inaugurated the theory of market microstructure. He argued that ``market agents can be treated as a statistical ensemble [and that] their market activities [can be] depicted as the stochastic generation of market orders according to a Poisson process''. His model has been extended by many authors to provide a microstructure foundation for the emergence of financial bubbles and crashes \cite{BrunnermeierSannikov,FollmerSchweizer, Lux} and order book dynamics \cite{Cont2013,HorstKreher2018,KreherMilbradt,Rosu} to name just a few applications. 

We assume that order arrivals are governed by a Poisson process $N$ and that the $k$-the order increases/decreases the logarithmic price by a random amount $J_k$ times the volatility. The sequence $\{J_k\}_{k \in \mathbb N}$ is assumed to be i.i.d., centered and without skew. Additionally, the order changes the logarithmic volatility by a random amount $\xi_k$. The sequence $\{\xi_k\}_{k \in \mathbb Z}$ is again assumed to be i.i.d.\, centered and without skew, but price and volatility changes may be correlated. Our key assumption is that the impact of an order on future volatility decays slowly according to some heavy-tailed kernel 
\begin{equation}\label{Eq:Intro_Kernel}
    \phi(t)= (1+t)^{H-\frac{1}{2}}\quad \mbox{with Hurst coefficient} \quad H \in (0,1/2].
\end{equation}
on $[0,\infty)$; by convention, we set $\phi(t)=0$ for $t<0$. 

In other words, in our model total volatility is given by a weighted average of past impacts, whereas the previous literature assumed that volatility triggers additional order flow. More precisely, we assume that the logarithmic price and volatility increments 
\[
    \widetilde{P}_t= \log P_t - \log P_0 \quad \mbox{and} \quad \widetilde{V}_t= \log V_t - \log V_0 
\]    
of the price-volatility process $(P,V)$ are given by (the dynamics turns out to be well defined)
\begin{equation} \label{Eq:dynamics}
\begin{split}
\widetilde{P}_t&=  \sum_{k=1}^{N_t} J_k e^{\widetilde{V}_{\tau_k-}},\quad t\geq 0,\cr 
\widetilde{V}_t &= 
\lim_{M \rightarrow \infty} \biggl(\sum_{k=-M}^{N_t} \xi_{k} \phi(t-\tau_k) -   \sum_{k=-M}^{N_0} \xi_{k}\phi(-\tau_k)\biggr)
\\
&=\sum_{k=1}^{N_t} \xi_{k} \phi(t-\tau_k) 
+ \sum_{k=-\infty}^{0} \xi_{k}\big( \phi(t-\tau_k) - \phi(-\tau_k)\big),\quad t\geq 0.
\end{split}
\end{equation}

The decomposition of the volatility process $V$ as $V_t = V_0 e^{\log \frac{V_t}{V_0}}$ is key to our approach. It allows us to additively decompose the volatility process into an impact 
\begin{equation}\label{Eq:V0_Intro}
    V^0_t := \sum_{k=-\infty}^{0} \xi_{k} \,\big( \phi(t-\tau_k) - \phi(-\tau_k)\big),\quad t\geq 0
\end{equation}
of past trades  that arrived before time zero and an impact 
$$
    V^1_t := \sum_{k=1}^{N_t} \xi_{k} \, \phi(t-\tau_k),\quad t\geq 0
$$ 
of trades that arrived after time zero. After rescaling, the process the sum of these processes converges to a fractional Brownian motion in Mandelbrot-van Ness form.

Specifically, our goal is to prove the weak convergence in the Skorohod topology of the family of rescaled price-volatility processes 
\[
    \widetilde{P}^n_t := \frac{\widetilde P_{nt}}{\sqrt{n}} \quad \mbox{and} \quad \widetilde{V}^n_t := \frac{\widetilde V_{nt}}{n^H},\quad t\geq 0,
\]
with kernels $\phi_n$ that asymptotically coincide with the one given in Equation \eqref{Eq:Intro_Kernel}\footnote{We shall see that the specific choice of the kernels $\phi_n$ that optimize weak error rates.} to the rough Bergomi model of the form 
\begin{equation*}
\Pcont_{t}=  \sigma_p \int_0^t e^{\widetilde{V}^*_s}dW_s
\quad \mbox{and}\quad
\Vcont_t= \sigma_v \fBM_t,\quad t\geq 0.
\end{equation*}
Here, $\sigma_p$ and $\sigma_v$ are the standard deviations of $J_k$ and $\xi_k$ respectively,
$W$ is a standard Brownian motion, and $\fBM$ is a fractional Brownian motion of the form \eqref{Def:fBM}. 
The correlation of the Brownian motions $W$ and $B$ results from the correlation of $J_k$ and $\xi_k$.\

The key challenge in proving the weak convergence is again the $C$-tightness of the rescaled volatility processes; proving the convergence of the price process is standard once the weak convergence of the volatility process has been shown. The $C$-tightness of the rescaled process $V^{0}$ away from the initial time is standard; the challenge is to include the initial time for which - heuristically - we prove the uniform convergence in probability to zero $[0,\epsilon]$ for small enough $\epsilon >0$. The $C$-tightness of the rescaled process $V^{1}$ is proven using the $C$-tightness result for c\`adl\`ag processes established in \cite{HorstXuZhang1}. Once the $C$-tightness has been shown, proving the convergence of the finite-dimensional distribution to fractional Brownian motion and hence the convergence of the price-volatility process is not difficult.

\subsection{Weak Convergence Rates}

Having established our weak convergence result, we are interested in evaluating (option) payoffs $\Phi$ in the prelimit and the limiting model. Our Poisson arrival dynamics allow us to consider \emph{weak errors} of the form 
\begin{align*}
    \biggl|\mbE[ \Phi(\Pn)]-\mbE[ \Phi(\Pcont)]\biggr|.
\end{align*}
We are unaware how to establish weak error rates for non-Poisson arrivals such as Hawkes-type dynamics and hence how to extend our results from rough Bergomi to rough Heston models.

Given these errors decay at a reasonable rate, then the microstructure model is a viable alternative to classical simulation schemes, such as the Euler or Milstein scheme, with the additional advantage of retaining a clear economic interpretation.

The non-Markovian nature of the dynamics makes numerical implementations of the rough Bergomi very challenging. Moreover, the fractional Brownian motion $\fBM$ is at most $H-$H\"older continuous, meaning that typical left- or mid-point approximations achieve, at best, an $L^2$-convergence rate of order $H$. Previous works have indeed observed a \emph{strong error} of order $H$ for Euler-Maruyama approximations \cite{Neuenkirch_Shalaiko_2016}, and of order $2H$ for Milstein-schemes \cite{Li_Huang_Hu_2021,Nualart_Saikia_2022,Richard_Tan_Yang_2021}, already including extensions to more general stochastic Volterra equations.

Even in classical Markovian settings, weak errors typically decay at twice the rate of strong errors. For the rough Stein-Stein model, Bayer et al.~\cite{BAYER_HALL_TEMPONE_2022}  use Markovian approximations of the fractional Brownian motion allowing them to use classical techniques for weak error analysis developed by Talay and Tubaro \cite{Talay_Tubaro_1990}. For the Markovian approximation, they establish a weak error rate of order $H+1/2$, for which they also can show that this does not depend on the approximation level, allowing them to pass to the limit. 

A different strategy was employed by Gassiat \cite{Gassiat_2023} and later, in a more general setting, by Friz et al.~\cite{Friz_Salkeld_Wagenhofer_2025}.  
Exploiting the fact that the Clark-Ocone formula yields an explicit martingale representation of the (rough Bergomi) volatility, these authors derived a moment formula for the (logarithmic) stock price.  This representation in turn allows them to prove a weak convergence rate of order $3H+1/2$. 

The difference between a strong convergence rate of order $H$ and a weak convergence rate of order $3H + 1/2$ is substantial for practical applications. For $H \approx 0.1$, to halve the strong error, the number of approximation steps must increase by a factor of $2^{1/H}=2^{10}\approx 1000$. In stark contrast, halving the weak error requires only $2^{1/({3H+1/2})}=2^{1/{0.8}}\approx 2.4$ times as many time steps.

Bonesini et al.~\cite{bonesini2023rough} also obtained a weak rate of $3H+1/2$, furthermore extending the result to a broader class of test functions. A key technique they require is a Feynman-Kac formula for Volterra processes, introduced in
\cite{Viens_Zhang_2019,Wang_Yong_Zhang_2022}, leading to path-dependent PDEs. To the best of our knowledge, no extension to jump-driven settings has been established for weak convergence within the rough setting.

Classical (weak) convergence rates for the approximation of Gaussian processes by compound Poisson processes rely on Stein's method \cite{Barbour_1990,Bourguin_Peccati_2016}. In general, convergence of order $1/2$ is obtained, in accordance with the Berry–Esseen theorem. Since compound Poisson processes typically exhibit non-vanishing skewness, Edgeworth expansions do not improve this rate. 

In our setting, we impose a no-skew condition. Under such symmetry previous works established weak convergence of order $1$, see, for example Dobler and Peccati~\cite{Dobler_Peccati_2018}. These results, however, are derived in a memoryless,  Markovian context and do not account for the convolution kernels that generate the non-Markovian structure in our model. Stein's method, for example, relies on the availability of a suitable generator, which is neither readily accessible for the limiting price process, nor for its microstructure prelimit counterpart.

Moreover, neither the limiting (log-)volatility, nor the microstructure analogue are semimartingales, further complicating the analysis. 
We therefore adopt the strategy of \cite{Friz_Salkeld_Wagenhofer_2025}, which is well-suited for our framework, as the required mathematical tools, most notably a suitable Clark-Ocone formula, remain available in the Poisson setting \cite{Zhang_2009}. 
This approach is feasible in our Poisson setting because our microstructural model allows for tractable computations of functional differences with respect to point variations (the Poisson analogue of Malliavin derivatives). Similar computations in a Hawkes model would be much harder, due to its self-exciting dynamics.  As a by-product of our method, we also obtain an (approximate) moment formula for the microstructure model, that could be of further use for statistical analysis from limit order book data.

Our weak error rates strongly hinge on the choice of the (asymptotically identical)  kernels $\{\phi_n\}$ for the rescaled models. The error rate for the benchmark choice $\phi_n = \phi$ where $\phi$ is given by Equation \eqref{Eq:Intro_Kernel} turns out to be $2H$. If, instead, we choose $\phi_n(t) \approx (n^{-\alpha} +t)^{H-1/2}$ for an optimized choice of  $\alpha$ plus some corrections for small times, then the weak error rate is of order $\tfrac13+\frac{4H}{3-6H}$. 

To establish this convergence order, we first apply the (approximate) moment representation for the microstructure- and rough Bergomi model. This enables us to compare the integrands of this representation term by term and to  derive an upper bound that depends solely on a functional of the kernels. We then construct a specific kernel that allows us to obtain the desired rate. Moreover, we show, that this functional decays at most of the order $\tfrac13+\frac{4H}{3-6H}$, which demonstrates that our rate is optimal in this sense.

\section{Assumptions and Main Results}

In this section, we introduce a benchmark asset price model for which we drive a scaling limit in a later section. 
We assume throughout that all random variables and stochastic processes are defined on a common probability space $(\Omega, \mathcal{F} , \mathbb{P})$ endowed with a filtration $\{\mathcal{F}_t : t \geq 0\}$ that satisfies the usual conditions. The convergence concept for stochastic processes we use will be weak convergence in the space $\mathbf{C}([0,T];\mathbb{R}^d)$ of all $\mathbb{R}^d$-valued continuous functions on $[0,T]$ endowed with the uniform topology or in the space $\mathbb{D}([0,T];\mathbb{R}^d)$ of all $\mathbb{R}^d$-valued c\`{a}dl\`{a}g functions on $[0,T]$ endowed with the Skorokhod topology; see \cite{Billingsley1999,JacodShiryaev2003}.

 \subsection{The Benchmark Model}
 
 We consider an order-driven model where asset prices are driven by incoming orders to buy or sell the asset. Order to buy and sell an asset arrive according to Poisson process $N$ with unit rate on $\mathbb R$. The order arrivals times are described by an increasing sequence random times $\{\tau_k  \}_{k\in \mathbb{Z}}$. The impact of orders on prices and volatilities is described by independent sequence of i.i.d.~random variables $\{J_m\}_{m \in \mathbb N}$ and $\{\xi_k\}_{k \in \mathbb Z}$, respectively. In terms of these sequences we define the random point measure
 \begin{equation*}
 N(ds,du,dv) := \sum_{k=-\infty}^\infty {\mathds{1}}_{\{\tau_k \in ds, J_k \in du, \xi_k \in dv\}}
 \end{equation*}
 on $\mathbb R^3$ with intensity measure $ds\,\mathcal{P}(du,dv)$ for a probability measure $\mathcal{P}(du,dv)$ on $\mathbb{R}^2$.  We assume that the logarithmic price/volatility process satisfies the dynamics \eqref{Eq:dynamics} which can be conveniently rewritten in integral form as
\begin{align*}
\widetilde{P}_t&= \int_0^t \int_{\mathbb{R}^2}   e^{\widetilde{V}_{s-}}  u\, N(ds,du,dv),\cr
\widetilde{V}_t &=  \int_0^t \int_{\mathbb{R}^2}  \phi(t-s)  v \,N(ds,du,dv) 
\\
&\qquad + \int_{-\infty}^0\int_{\mathbb{R}^2}  \big( \phi(t-s)-\phi(-s)\big)  v\, N(ds,du,dv) , 
\end{align*}
Under the below assumptions on the intensity measure and the kernel the volatility dynamics will indeed be well defined; cf. Lemma \ref{Lem:VMoments}.

We assume throughout that the intensity measure satisfies the following standard assumptions. 

\begin{assumption}\label{Ass:Compensator}
The following conditions on the intensity measure are satisfied.
\begin{enumerate}[label=(\roman*)]
\item Centered: $\int_{\mbR^2} u \mcP(du,dv)=\int_{\mbR^2} v \mcP(du,dv)=0$,
\item Known second moments: \ $\int_{\mbR^2} u^2 \mcP(du,dv)=\sigma_p^2$, $\int_{\mbR^2} v^2 \mcP(du,dv)=\sigma_v^2$,
\item No skew: $\int_{\mbR^2} u^3 \mcP(du,dv)=\int_{\mbR^2} v^3 \mcP(du,dv)=0$,
\item %
(Sub-)Gaussian tails:
For some $\alpha >0$ it holds that $\int_{\mbR^2}e^{\alpha u^2+\alpha v^2} \mcP(du,dv)<\infty.$
\item No mixed skew: $\int_{\mbR^2} uv^2 \mcP(du,dv)=\int_{\mbR^2} u^2v \mcP(du,dv)=0$.
\item Correlation structure: $\sigma_p\sigma_v\rho = \int_{\mbR^2} uv \,\mcP(du,dv)$.
\end{enumerate}
\end{assumption}

\begin{example}
If $\mcP$ is a two-dimensional centered Gaussian distribution, then all the above items are satisfied. In this case, $\int_{\mbR^2}uv\mcP(du,dv)$ equals the covariance of the price and volatility increments.    
\end{example}

Our goal is to establish a weak convergence result for the sequence of rescaled processes 
\begin{align*}
    \Pn_{t} & := \int_0^{nt} \int_{\mathbb{R}^2}   e^{ \widetilde{V}^n_{s-}} \cdot \frac{u}{\sqrt{n}}\, N(ds,du,dv) \\ 
\Vn_{t} & := \int_0^{nt} \int_{\mathbb{R}^2}  \phi(nt-s) \cdot \frac{v}{n^H}\, N(ds,du,dv) \\
&\qquad + \int_{-\infty}^0\int_{\mathbb{R}^2}  \big( \phi(nt-s)-\phi(-s)\big) \cdot  \frac{v}{n^H}\, N(ds,du,dv)
\end{align*}
Since $$n^{-H}\phi(nt-s)=n^{-H}(1+nt-s)^{H-1/2}=n^{-1/2}\Bigl(\frac1n+t-\frac sn\Bigr)^{H-1/2}\eqqcolon n^{-1/2}\phi_n\Bigl(t-\frac sn\Bigr)$$ we change the intensity of the random measure and consider the following general dynamics.

\begin{definition}\label{Def:Pn_Vn}
Let $N(dt,du,dv)$ be a Poisson random measure with compensator $dt \cdot \mcP(du,dv)$ satisfying Assumption \ref{Ass:Compensator}. Let $\{\phi_n\}_{n \in \mathbb{N}}$ be a family of square integrable functions, and $T>0$ be a finite time horizon. For $t \in [0,T]$ we define the rescaled logarithmic price-volatility processes
\begin{align}
    \Pn_{t}&=  \int_0^{t} \int_{\mathbb{R}^2}   e^{ \Vn_{s-}} \cdot \frac{u}{\sqrt{n}}\, N(n\cdot ds,du,dv), \label{Def:Price_Process}
    \\
    \Vn_{t}&=  \int_0^{t} \int_{\mathbb{R}^2}  \phi_n(t-s) \frac{v}{\sqrt{n}}\, N(n\cdot ds,du,dv)  \label{Def:Vol_Process}
    \\
    &\qquad + \int_{-\infty}^0\int_{\mathbb{R}^2}  \big( \phi_n(t-s)-\phi_n(-s)\big) \cdot  \frac{v}{\sqrt{n}}\, N(n\cdot ds,du,dv). \nonumber
\end{align}
\end{definition}
\begin{remark}\label{Rem:CompRep}
    By part (i) in Assumption \ref{Ass:Kernel} we can substitute the driving Poisson random measures in \eqref{Def:Price_Process} and \eqref{Def:Vol_Process} by its compensated analogue $\tilde N$.
\end{remark}

We now state our assumptions on the kernels $\phi_n$. The benchmark model that we have in mind is
\begin{equation}\label{Eq:Benchmark}
    \phi_n(t) = (n^{-1}+t)^{H-1/2}.
\end{equation}
However, as we will see, much better weak convergence rates can be obtained for slightly different kernels. %

\begin{assumption}\label{Ass:Kernel}
Let $\phi_\infty(t)=t^{H-1/2}$, $(t\ge 0)$  and $\{\phi_n\}_{n\in \mbN}$ be a family of kernels that is almost everywhere differentiable with derivative  $\phi^\prime_n(t) $.
We assume that there exists a constants $C>0$ and $\theta>2$ independent of $n\in \mbN$ such that the following hold:
\begin{enumerate}[label=(\roman*)]
\item Positivity and majorant:
$   
    0\le \phi_n(t) \le C\left(\frac{1}{n}+ t\right)^{H-\frac12}$, $t\ge0.\label{Cond:Bound1}
$
\item Monotonicity and bounded derivative:
$ 
0 \le -\phi_n^{\prime}(t) \le C \left(\frac{1}{n}+ t\right)^{H-\frac32}$, $t\ge 0$, $n\in \mbN$.\label{Cond:Bound2}
\item Uniform continuity:
$
 \lim\limits_{n\to \infty}\sup\limits_{h \leq n^{-\theta}} \int\limits_0^T |\phi'_n(h+s) - \phi'_n(s)|ds = 0.
   \label{Cond:Bound3}
$
\item $L2$-convergence:
$
\lim\limits_{n\to \infty}\int\limits_0^T |\phi_n(T-s) -\phi_\infty(T-s)|^2ds = 0 \label{Cond:Bound4}.
$
\end{enumerate}
\end{assumption}

\begin{example}
The benchmark kernels in Equation \eqref{Eq:Benchmark} satisfy Assumption $\ref{Ass:Kernel}$. Indeed, Conditions \ref{Cond:Bound1} and \ref{Cond:Bound2} from Assumption \ref{Ass:Kernel} follow easily. Condition \ref{Cond:Bound3} hold because for any $h\le n^{-\theta}, \theta >2$,
\begin{align*}
    & \int_0^T \Bigl|\Bigl(\frac1n+h+s\Bigr)^{H-3/2}-\Bigl(\frac1n+s\Bigr)^{H-3/2}\Bigr|ds \\ = &\int_0^h \Bigl(\frac1n+s\Bigr)^{H-3/2}\,ds + \int_{T-h}^T \Bigl(\frac1n+h+s\Bigr)^{H-3/2}\,ds
    \\
    \le & ~ 2h n^{3/2-H} \\ \le & ~ 2 n^{3/2-\theta-H}\xrightarrow{n\rightarrow \infty}0.
\end{align*}
Since $\{\phi_n(t)\}_{n\in\mathbb{N}}$ is an increasing sequence that converges pointwise to $t^{H-1/2}$, Condition \ref{Cond:Bound4} also holds by dominated convergence. 
\end{example}

\subsection{Main Results}

We are now ready to state the main results of this paper. The first states that the sequence of rescaled price-volatility processes converges weakly to a log-normal rough volatility model. The proof is given in Section \ref{Sec:SL}.

\begin{theorem}\label{Thm:WeakConv}
The rescaled logarithmic price-volatility process $\{(\Pn,\Vn)\}_{n \in \mbN}$ converges weakly to the process $(\widetilde{P}^*,\widetilde{V}^*)$ in $\mathbb{D}([0,T];\mathbb{R}^2)$ as $n\to\infty$ where
\begin{equation*}
\Pcont_{t}=  \int_0^t \sigma_p  e^{\widetilde{V}^*_s}dW_s
\quad \mbox{and}\quad
\Vcont_t= \sigma_vc_H \fBM_t,\quad t\geq 0,
\end{equation*}
with $\sigma_p$ and $\sigma_v$ from Assumption \ref{Ass:Compensator}. Here  $W$ is a standard Brownian motion and
$\fBM$ is a fractional Brownian motion with Hurst index $H$ that admits the representation \eqref{Def:fBM}
in terms of a standard Brownian motion $B$. The Brownian motions $W$ and $B$ are correlated with parameter $\rho$ from Assumption \ref{Ass:Compensator}
\end{theorem}

The key step in proving the above theorem is to establish the $C$-tightness of the family of rescaled volatility processes. This will be achieved by using the $C$-tightness established \cite{HorstXuZhang1}. After establishing a series of auxiliary results in Section \ref{Sec:TechLemmas}, we decompose the volatility processes into a sum of two processes that represent, respectively, the impact of orders that arrived before, respectively, after time zero on future volatility and show their $C$-tightness Sections \ref{Sec:Tight_V0} and \ref{Sec:Tight_V1} respectively.
In Section \ref{Sec:FinDimConv} we prove the weak convergence of the finite dimensional distributions. Section \ref{Sec:PrfWeakConv} summarizes the results and concludes the proof of Theorem \ref{Thm:WeakConv}.

Having shown the weak convergence of the price-volatility processes the next theorem establishes weak convergence rates. The proof is given in Section \ref{Sec:Weak_Conv}. 

\begin{theorem}
\label{Thm:Main_Thm}
There exists a family $\{\widehat\phi_n\}_{n \in \mathbb{N}}$ such that for
 the continuous-time (log-price) model $X$ and for the discrete model $\Pn$ as well as for any $N \in \mathbb{N}$ and $t\in [0,T]$ there is a constant $C$ only depending on $N,T,H$ such that
\begin{align*}
    \biggl|\mbE\Bigl[ \bigl(\Pn_t\bigr)^N\Bigr]-\mbE\Bigl[ \bigl(
    \Pcont_t \bigr)^N\Bigr]\biggr|
    \le C \begin{cases}n^{-\frac13-\frac{4H}{3-6H}} & \text{if } H \in (0,1/4),\\
        n^{-1} \log(n) & \text{if } H =1/4,\\
n^{-1} &\text{if } H \in (1/4,1/2).
\end{cases} 
\end{align*}
\end{theorem}

\begin{remark}\label{Rem:CounterEx}
We emphasize that the weak error rate of Theorem \ref{Thm:Main_Thm} requires a specific choice of kernels. For instance, the kernels $\phi_n(t)=(n^{-\alpha}+t)^{H-1/2}$ satisfy Assumption \ref{Ass:Kernel}. However, redoing the computations of Section \ref{Sec:Kernel_Estimates} for this family yields an upper bound of order $n^{-2H\alpha}$, hence a much slower convergence rate compared to Theorem \ref{Thm:Main_Thm}.    
\end{remark}

The key challenge in proving the above theorem is to obtain a suitable moment formula for the limiting price process and an approximate one for the prelimit models. The moment formulas are derived in Section \ref{Sec:Mom_Formula}; they rely on a Clark-Ocone representation for Poisson point processes originally established in \cite{Zhang_2009} and recapped in Section \ref{Sec:Clark_Ocone} as well as a series of a priori estimates given in Section \ref{Sec:PV_estimates}. The moment formulas yield an estimate of the form
\begin{align}\label{Eq:Mom_Estimate}
        \biggl|\mbE\Bigl[ \bigl(\Pn_t\bigr)^N\Bigr]-\mbE\Bigl[ \bigl(\Pcont_t\bigr)^N\Bigr]\biggr|\lesssim  \star(n)+\diamond(n)+\square(n)+\triangle(n),
\end{align}
where $\star(n)$ depends on the integrated 4th moments of the approximating kernels, $\diamond(n)$ and $\square(n)$ depend on the integrated differences of their covariance functions and $\triangle(n)$ denotes the $L^1$-difference of $\phi_n$ and $\phi_\infty$. 

In Section \ref{Sec:Kernel_Estimates} we construct a \emph{specific} family of kernels, such that for $H<1/6$
\begin{align}\label{Eq:IntroBd}
      \star(n)+\diamond(n)+\square(n)+\triangle\lesssim  n^{-\frac13-\frac{4H}{3-6H}}.
\end{align}
Our moment estimate in Equation \ref{Eq:Mom_Estimate} is not sharp (although we do not expect that it can be improved much) but the estimate  \eqref{Eq:IntroBd} is. Specifically, we prove that 
for any sequence of kernels $\{\phi_n\}_{n \in \mbN}$ satisfying Assumption \ref{Ass:Kernel} it holds that 
    \begin{align*}
       \star(n)+ \diamond(n) \gtrsim n^{-\frac13-\frac{4H}{3-6H}}.
    \end{align*}

\subsection{Numerical Simulations}

To illustrate our model, we simulated the log-volatility and log-price processes given in Definition \ref{Def:Pn_Vn}. We assumed a Hurst parameter of $H=0.15$ and a   correlation of $\rho=-0.7$. 
Figure \ref{Fig:samplepathG} displays a typical sample path of the log-dynamics for $n=10$ with $\xi_k \sim \mathcal{N}(0,0.25)$%
and $J_k\sim \mathcal{N}(0,1)$. For ease of computation, we assumed no volatility impact before time zero; see Equation \eqref{Eq:V0_Intro} and considered the kernel $\phi(t)=(1/n+t)^{H-1/2}$. Thus, our limiting volatility corresponds to a Riemann-Liouville fBM.  

\begin{figure}[!ht]
    \centering
\includegraphics[width=0.48\textwidth]{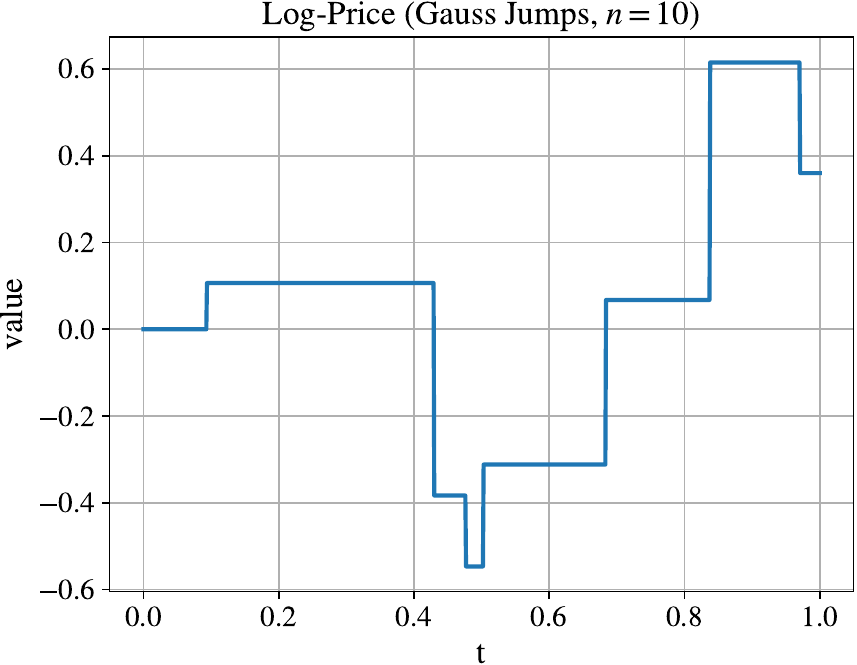  }
\includegraphics[width=0.48\textwidth]{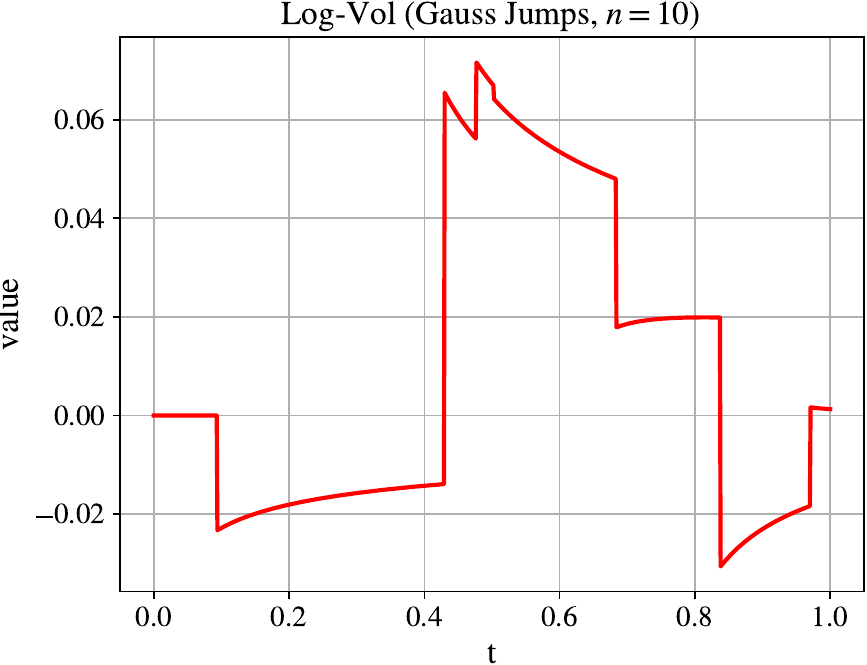}
    \caption{Plots of negatively-correlated log-price and Log-volatility of a sample path with Gaussian jumps for the volatility, kernel $\phi_n(t)=(1/n+t)^{H-1/2}$ and scaling  parameter $n={10}$.}
    \label{Fig:samplepathG} 
\end{figure} 

Figure \ref{Fig:samplepathB} displays sample paths for the same Hurst and correlation coefficient, but this time we used Bernoulli-jumps by applying the sign-function to the Gaussian jumps from the previous Figure and multiplying the result by $0.5$.  We used a jump-intensity of $20$, leading to jumps of the log-volatility process by $0.5\cdot \bigl(\frac1{20}\bigr)^{-0.35}\frac1{\sqrt{20}} \approx 0.32$. By construnction the jumps of the log-price depend on the level of the volatlity.

\begin{figure}[!ht]
    \centering
\includegraphics[width=0.48\textwidth]{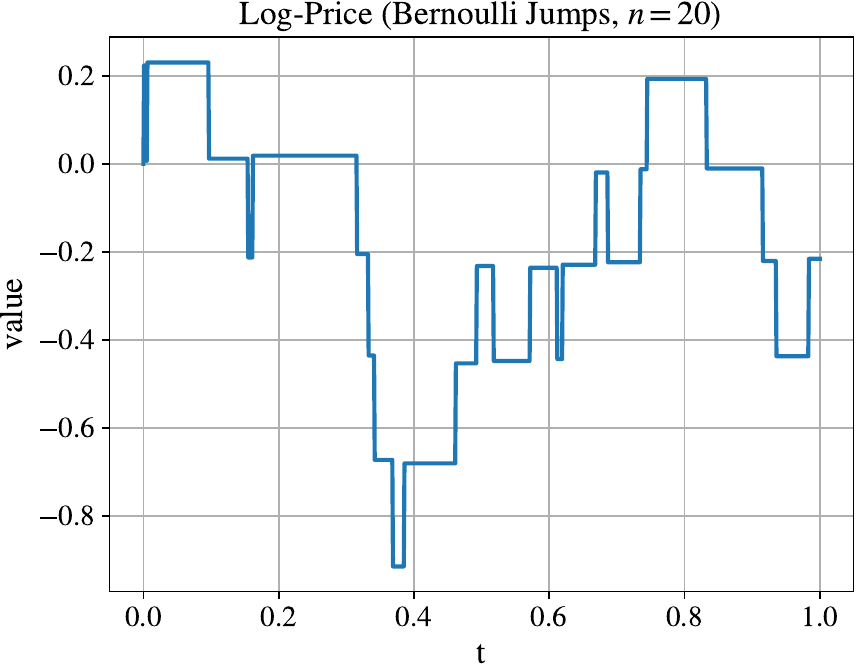}
\includegraphics[width=0.48\textwidth]{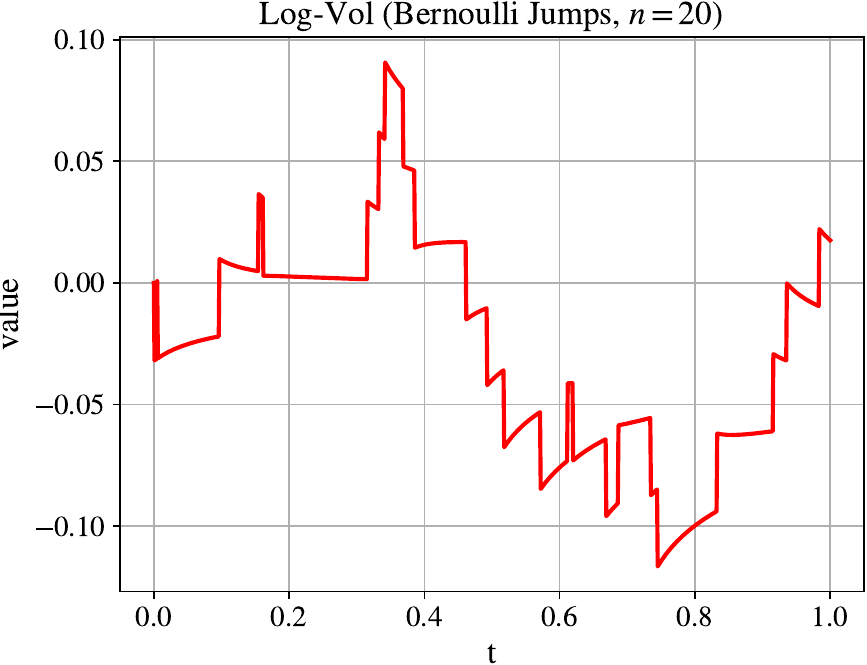}
     \caption{Plots of negatively-correlated log-price and Log-volatility of a sample path with Bernoulli jumps for the volatility, kernel $\phi_n(t)=(1/n+t)^{H-1/2}$ and scaling  parameter $n={20}$.}
    \label{Fig:samplepathB}
\end{figure}

In Figure \ref{Fig:ballplot} we characterize a neighborhood in Skorokhod space for log-price in the Bernoulli model from Figure \ref{Fig:samplepathB}. The path in the first picture is exactly the log-vol of Figure \ref{Fig:samplepathB}. For the second picture, an intensity scaling parameter $n=1000$ was used. As expected, the plot for the larger intensity over large time regions  is visually close to a continuous function. 

\begin{figure}[!ht]
    \centering
\includegraphics[width=0.48\textwidth]{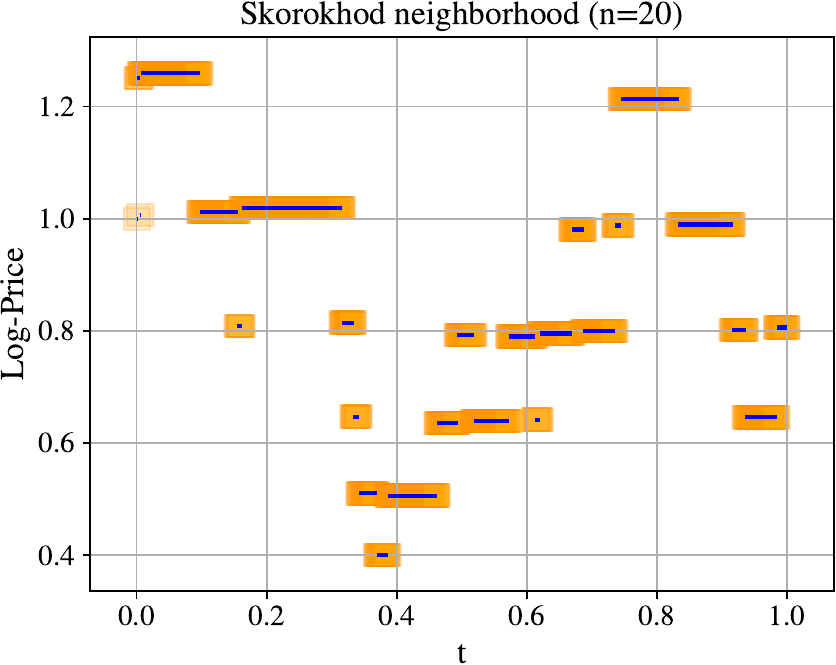}
\includegraphics[width=0.48\textwidth]{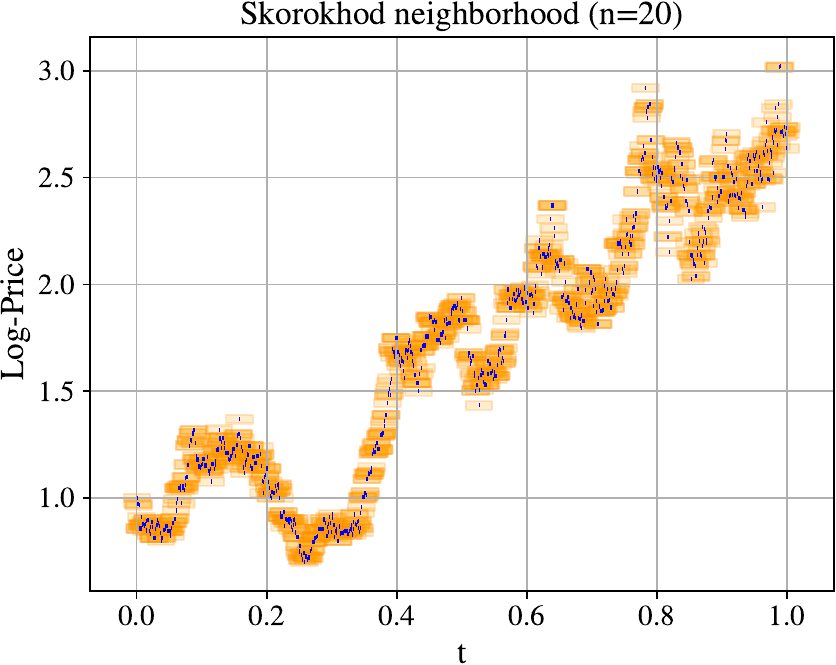}
    \caption{Depiction of Skorokhod neighborhoods for two prelimit  sample price processes with Gaussian price- and Bernoulli volatility-jumps for $n={200}$, respectively $n=1000$.}
    \label{Fig:ballplot}
\end{figure}

Figure \ref{Fig:MomError} illustrates the weak error rate for the fourth moment. As above, we assumed no volatility impact before time zero and considered a Hurst parameter of $H=0.15$ and $H=0.3$. The correlation was assumed to be $\rho=-1$ and $\xi_k \sim \mathcal{N}(0,(0.05)^2)$ and $J_k \sim \mathcal{N}(0,1)$.  
For a range of parameters $n$, we simulated $M = 2\cdot 10^{7}$ independent copies of $\Pn$ and estimated the expectation of the degree-four Hermite polynomial, $H_4(x)=x^4-6x^2+3$, via $$\mbE[H_4(\Pn_1)] \approx \frac{1}{M}\sum_{j=1}^M H_4(\Pn_1(\omega_j)).$$
We then compared this quantity with $\mbE[H_4(\Pcont_1)]$, where this benchmark value was computed by simulating $\Pcont$ with the Euler scheme using $5000$ time steps and $3.2\cdot 10^{7}$ Monte Carlo samples, so as to keep the Monte Carlo error of the reference value as small as possible.

The blue curve (\emph{Poisson simulation} in Figure \ref{Fig:MomError}) shows the Monte Carlo estimate of the weak error for the degree-four Hermite polynomial, while the dashed black line (\emph{Rate $\frac13+\frac{4H}{3-6H}$}) represents the proven theoretical convergence rate from Theorem~\ref{Thm:Main_Thm}, with a fitted constant prefactor. We show shaded confidence bands around the theoretical prediction. These bands represent approximate $68\%$, $95\%$, and $99.7\%$ Monte Carlo uncertainty regions centered at the theoretical line. A simulated weak error lying within these bands, or below them, for large enough $n$ is therefore consistent with the theoretical result, since the remaining discrepancy can be attributed to sampling noise. If the simulated weak error lies strictly within the bands, this moreover suggests that the proven convergence rate is in fact \emph{sharp}. Since the Poisson scheme is simulated with $M=2\cdot10^7$ samples, while the benchmark for the limiting model is computed from an Euler discretization with $3.2\cdot10^7$ Monte Carlo samples, the overall Monte Carlo error is sufficiently small to make this comparison meaningful. For both values of $H$, the simulated weak error remains within the displayed confidence bands, thereby supporting the theoretical result and suggesting that the proven convergence rate is in fact sharp.

\begin{figure}[!ht]
    \centering
\includegraphics[width=0.48\textwidth]{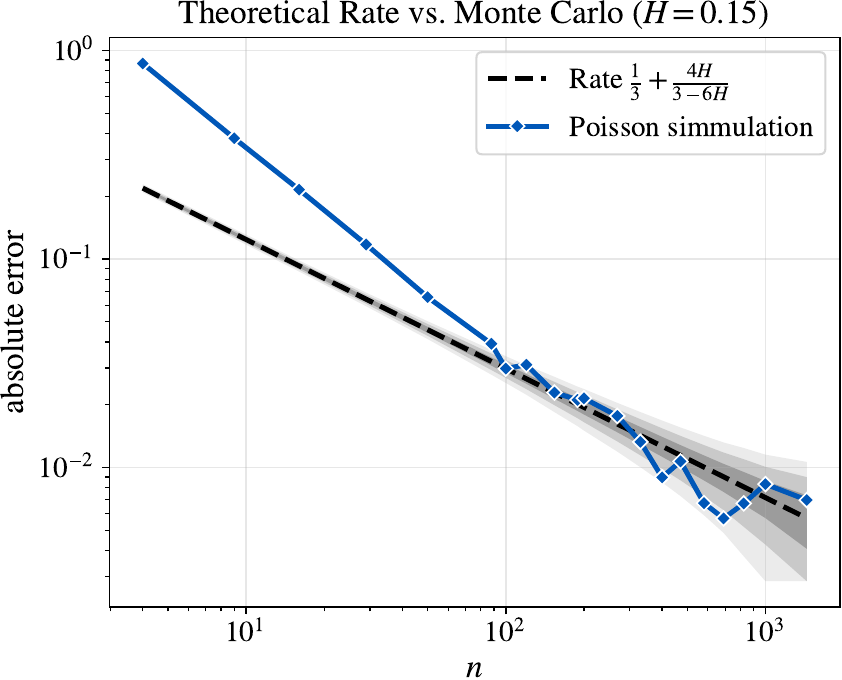}
\includegraphics[width=0.48\textwidth]{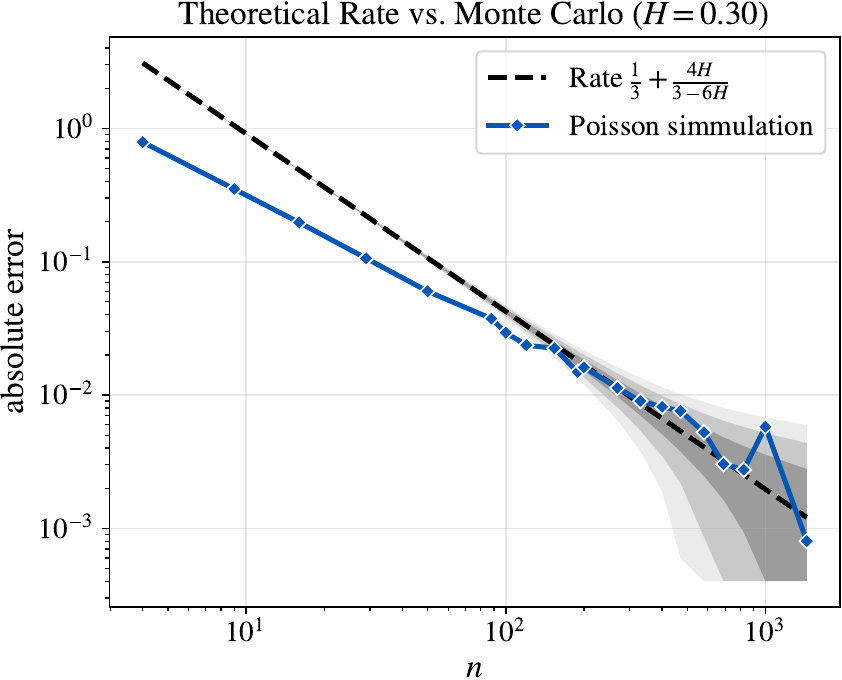}
    \caption{Weak error for the degree-four Hermite polynomial in the Poisson prelimit price model with the optimized kernel sequence $\{\widehat{\phi}_n\}_{n\in\mathbb{N}}$, shown for $H=0.15$ (left) and $H=0.30$ (right). The blue curve displays the Monte Carlo estimate of the absolute weak error relative to a benchmark computed with the Euler scheme on a time grid of size $5000$, while the dashed black line represents the proven theoretical convergence rate with a fitted constant prefactor. The shaded regions indicate approximate $68\%$, $95\%$, and $99.7\%$ Monte Carlo uncertainty bands around this theoretical prediction. Both panels are shown on log-log scales.}
    \label{Fig:MomError}
\end{figure}

\section{Proof of Weak Convergence}
\label{Sec:SL}
\setcounter{equation}{0}

In this section we prove our weak convergence result for rough log-normal volatility models. By a change of variables, we first rewrite the rescaled log-price process $\Pn$ as 
\begin{align}
\Pn_t&\coloneqq  \int_0^{t}    e^{ \Vn_{s-}} \, d\Wn_s
\quad\mbox{with}\quad
\Wn_t\coloneqq \int_0^{nt}  \int_{\mathbb{R}^2} \frac{u}{\sqrt{n}}\, N(ds,du,dv),\quad t\geq 0. \label{Eq:Wn}
\end{align}

Standard arguments yield the weak convergence of the martingales  $\{\Wn \}_{n \in \mbN}$ to Brownian motion. They key challenge is the weak convergence of the log-volatility processes $\{ \Vn\}_{n \in \mbN}$ that, with Remark \ref{Rem:CompRep} in mind, can be decomposed into the sum of the following two processes:
\begin{align*}
\Vnz_t&\coloneqq  \int_{-\infty}^0\int_{\mathbb{R}^2} \Bigl[ \phi_n(t-s)-\phi_n(-s)\Bigr]     \frac{v}{\sqrt{n}}\, \tilde N(n\cdot ds,du,dv),\\
\Vno_t&\coloneqq  \int_0^{t} \int_{\mathbb{R}^2} \phi_n(t-s)  \frac{v}{\sqrt{n}}\, \tilde N(n\cdot ds,du,dv). 
\end{align*}

By  \cite[Corollary~3.33, p.353]{JacodShiryaev2003}, it suffices to prove the $C$-tightness of the two processes separately, which is achieved in Sections \ref{Sec:Tight_V0} and \ref{Sec:Tight_V1}, respectively. Subsequently, we identify their weak limits.

\subsection{A-priori Estimates}\label{Sec:TechLemmas}

We start with two a priori estimates that will be used repeatedly in what follows. The first recalls a BDG-type result for stochastic integrals w.r.t. compensated Poisson random measures. 

\begin{lemma}[\cite{Kunita_2004}, Theorem 2.11] \label{Lem:BDG}
Let $\tilde N$ be a compensated Poisson random measure with intensity measure~$\mathcal{P}$. Let $g$ be predictable and square integrable w.r.t.\ $\lambda \otimes \mathcal{P}$ and let  
\begin{equation*}
X_t
:= x
+ \int_0^t b(r)\,dr
+ \int_0^t \int_{\mathcal{\mathbb{R}^2}} g(s,u,v)\,\tilde{N}(ds\,du,dv),
\end{equation*}
Then, for any $p \ge 2$, there exists a positive constant $C_p$ such that
\begin{align*}
\mathbb{E}\!\left[ \sup_{0 \le s \le t} |X_s|^p \right]
\le\;
C_p \Bigg\{ &
|x|^p
+ \mathbb{E}\!\left[ \left( \int_0^t |b(r)|\,dr \right)^p \right]
+ \mathbb{E}\!\left[ \left( \int_0^t \int_{\mbR^2} |g(s,u,v)|^2\,\mcP(du,dv)\,ds \right)^{p/2} \right]
\\
&\quad
+ \mathbb{E}\!\left[ \int_0^t \int_{{\mbR^2}} |g(s,u,v)|^p\,\mcP(du,dv)\,ds \right]
\Bigg\}.
\end{align*}
\end{lemma}
The second lemma establishes exponential integrability of the log-volatility processes and polynomial integrability of the log-price process. 
\begin{lemma}\label{Lem:VMoments}
For any $k\in \mathbb{N}$,  
\begin{align*}
    \sup_{t\in [0,T]}\sup_{n \in \mathbb{N}}\mbE\bigl[ \exp\bigl( 
    k \Vn_t
    \bigr)\bigr]<\infty\quad 
\text{ and } \quad
    \sup_{t\in [0,T]}\sup_{n \in \mathbb{N}}\mbE\bigl[ \bigl| \Pn_t\bigr|^k \bigr]<\infty.
\end{align*}

\end{lemma}
\begin{proof}
Let $n\in \mathbb{N}$ and $t \in [0,T]$ and let 
$$
    C:=\sup_{n  \in \mbN} (\|\phi_n n^{-1/2}\|_\infty+\|\phi_n/\phi_\infty\|_\infty).
$$

Calculating the moment generating function of integrals w.r.t.\ the Poisson Random Measure $N$ yields
\begin{align*}
    \mbE\bigl[ \exp\bigl( 
    k \Vn_t
    \bigr)\bigr]&=\exp\Bigl( \int_0^t \int_{\mbR^2} n\Bigl(\exp\bigl( \phi_n(t-s)\frac{kv}{\sqrt{n}}\bigr)-1\Bigr) \mcP(du,dv)ds\Bigr),
\end{align*}
in the sense that the left hand side is finite if and only if the right hand side is.
A standard Taylor estimate, using the assumption of centered first moments, then yields
\begin{align*}
    & \mbE\bigl[ \exp\bigl( 
    k \Vn_t 
    \bigr)\bigr] \\
    &\le \exp\biggl( \int_0^t \int_{\mbR^2} \frac12\phi_n^2(t-s)k^2v^2 + \frac16n^{-1/2}\phi_n^3(t-s)|kv|^3\exp\Bigl( \phi_n(t-s)\frac{|kv|}{\sqrt{n}}\Bigr) \mcP(du,dv)ds\biggr)
    \\
    &\le \exp\biggl( \int_0^T \int_{\mbR^2} \frac12C^2\phi^2(T-s)k^2v^2+ \frac16C^3\phi^2(T-s)|kv|^3\exp\bigl(Ck|v|\bigr) \mcP(du,dv)ds\biggr),
\end{align*}
which is finite, due to the square integrability of $\phi$ and Assumption \ref{Ass:Compensator}. As the last term depends neither on $n$ nor on $t$, the claim follows.

For the log-price process we first note that
\begin{align*}
    \int_0^t \int_{\mbR^2} u^k e^{k\Vn_{s-}} \mcP(du,dv)\,ds\le T \sup_{s \in [0,t]}e^{k\Vn_{s-}} \int_{\mbR^2} u^k  \mcP(du,dv).
\end{align*}
Taking expectations and using Doob's inequality, finiteness follows from exponential moments for $\Vn$ and the sub-Gaussian tails in Assumption \ref{Ass:Compensator}.
\end{proof}

\subsection[C-tightness of Vn]{$C$-tightness of $\{\Vnz\}$}\label{Sec:Tight_V0}

In this section, we prove the $C$-tightness of the sequence $\{  \Vnz \}_{n \in \mbN}$. We start with the following general tightness criterion for sequences of processes processes 
$$\{ X^{(n)} \}_{n\in \mbN} \subset \mathbb{D}([0,T];\mathbb{R}).$$

\begin{lemma}\label{Lem:TightCriterion}
If the sequence of initial states $\{X^{(n)}_0\}_{n \in \mbN}$ is tight in $\mathbb{R}$, then the sequence $\{X^{(n)} \}_{n \in \mbN}$ is tight [resp. $C$-tight] in $\mathbb{D}([0,T];\mathbb{R})$ if the following two  conditions hold:
\begin{enumerate}   
\item[(1)] For any $\epsilon>0$, the sequence $\{X^{(n)} \}_{n  \in \mbN }$  is tight [resp. $C$-tight] in $\mathbb{D}\big([\epsilon,T];\mathbb{R}\big)$.

\item[(2)] $ \sup\limits_{t\in[0,\epsilon]}\big| X^{(n)}_t-X^{(n)}_0 \big|\overset{\mathbb{P}}\to 0$ as $n\to\infty$ and then $\epsilon\to 0+$.
\end{enumerate} 

\end{lemma}
\begin{proof} 
In view of Condition (1) a standard diagonal argument yields a process $X^*$ whose restriction to $(0,T]$ belongs to $\mathbb{D}((0,T];\mathbb{R})$ and  (along a subsequence) for any rational $\epsilon > 0$, satisfies 
$$
    X^{(n)}(0) \overset{d} \to X^*_0, \quad 
    X^{(n)}\to X^* ~~ \mbox{weakly in} ~~ \mathbb{D}\big([\epsilon,T];\mathbb{R}\big)
$$
By the Skorokhod representation theorem, we may assume that we actually have a.s.~convergence and by Condition (2)
\begin{align*}    
\limsup_{\delta\to 0+}\big|X^*_\delta-X^*_0 \big|= \limsup_{\delta\to 0+}\lim_{n\to\infty} \big|X^{(n)}_\delta - X^{(n)}_0\big| \overset{\mathbb{P}}=0
\end{align*}
from which we deduce that 
$$
    X^* \in \mathbb{D}([0,T];\mathbb{R}).
$$

To obtain the weak convergence $X^{(n)}\to X^*$ in $\mathbb{D}\big([0,T];\mathbb{R}\big)$, it suffices to prove that for any bounded and uniformly continuous function $F:\mathbb{D}\big([0,T];\mathbb{R}\big)\rightarrow \mathbb{R}$ the following limit holds 
\begin{align*}%
\lim_{n \to \infty} \Big| \mathbb{E}\big[  F(X^{(n)})\big] -	\mathbb{E}\big[  F(X^*)\big] \Big| = 0.
\end{align*}
For each rational number $\epsilon \in (0,1)$ and  $n\in \mbN$, we can write $X^{(n)} = X^{(n),\eps} + Y^{(n),\eps}$ and 
$X^* = X^{*,\eps} + Y^{*,\eps}$ with 
\begin{align*}
X^{(n),\eps}_t \coloneqq X^{(n)}_{t\vee \epsilon},\quad Y^{(n),\eps}_t :=  X^{(n)}_{t \wedge \epsilon} - X^{(n)}_\eps, \quad
X^{*,\eps}_t \coloneqq X^*_{t\vee \epsilon},\qquad\  Y^{*,\eps}_t:=  X^*_{t \wedge \epsilon} - X^*_\eps.
\end{align*}
Then, 
\begin{equation*}
\begin{split}
& \Big| \mathbb{E}\big[  F(X^{(n)})\big] -	\mathbb{E}\big[  F(X^*)\big] \Big| \\
&\leq 
\Big| \mathbb{E}\big[ F(X^{(n),\eps}+ Y^{(n),\epsilon} )-F(X^{(n),\eps})\big] \Big| + \Big| \mathbb{E}\big[  F(X^{(n),\eps} )\big] -	\mathbb{E}\big[  F(X^{*,\eps}  )\big] \Big|
\\
&\qquad + \Big| \mathbb{E}\big[  F(X^{*,\eps} )- F(X^{*,\eps} + Y^{*,\epsilon})\big] \Big|.
\end{split}
\end{equation*}

By Condition~(1) $X^{(n),\eps}\to X^{*,\eps}$ weakly in $\mathbb{D}\big([0,T];\mathbb{R}\big)$ and hence  
\begin{equation*}
\lim_{n\to\infty}  \Big| \mathbb{E}\big[  F(X^{(n),\eps} )\big] -	\mathbb{E}\big[  F(X^{*,\eps}  )\big] \Big| =0. 
\end{equation*}

To prove the convergence to zero of the remaining two terms we denote by $d$ the Skorokhod metric. Since $F$ is bounded and uniformly continuous, for any $\varepsilon\in (0,1)$, there exists a constant $\delta>0$ s.t. 
$$\big| F(x+y)-F(x) \big| \leq \varepsilon \cdot \mathds{1}_{\{d(y,0) \leq \delta \}} + 2\|F\|_\infty \cdot \mathds{1}_{\{d(y,0) > \delta \}}.
$$
As a result, 
\begin{align*}
\Big| \mathbb{E}\big[ F(X^{(n),\eps}+ Y^{(n),\epsilon} )-F(X^{(n),\eps})\big] \Big|  \leq \varepsilon  +2 \|f\|_\infty  \mathbb{P}\bigl(d(Y^{(n),\epsilon},0) \geq \delta \bigr)
\end{align*}
and 
\begin{align*}
\Big| \mathbb{E}\big[  F(X^{*,\eps} )- F(X^{*,\eps} + Y^{*,\epsilon)}\big] \Big| \leq \varepsilon  
+2 \|f\|_\infty  \mathbb{P}\bigl(d(Y^{*,\epsilon},0) \geq \delta \bigr).
\end{align*}

By condition~(2) $\sup\limits_{t\geq 0} | Y^{(n),\eps}_t |\overset{\mathbb{P}}\to 0$ and $\sup\limits_{t\geq 0} | Y^{*,\eps}_t |\overset{\mbP}\to 0$ as $n\to \infty$ and then $\epsilon\to0+$. Hence, as $n\to\infty$, 
$$
    \mathbb{P}\bigl(d(Y^{(n),\eps},0) \geq \delta \bigr) \to 0 \quad \mbox{and} \quad \mathbb{P}\bigl(d(Y^{(n),\eps},0) \geq \delta \bigr) \to 0. 
$$    
\end{proof}

To prove the weak convergence of the sequence $\{  \Vnz \}_{n \in \mbN}$, it remains to to verify the conditions of Lemma~\ref{Lem:TightCriterion}. This is achieved in the next two propositions.

\begin{proposition}
For each $\epsilon >0$ and $p\geq 1$, there exists a constant $C>0$ such that for any $t_2>t_1\geq \epsilon$,
\begin{equation}\label{Eq:DiffMomentBound}
\sup_{n \in \mbN}	\mathbb{E}\Big[ \big| \Vnz_{t_2}- \Vnz_{t_1} \big|^{2p} \Big] \leq C\cdot \big|t_2-t_1\big|^{2p}  .
\end{equation}
In particular, the sequence $\{  \Vnz \}_{n \in \mbN}$ is  $C$-tight in $\mathbb{D}\big([\epsilon,\infty);\mathbb{R}\big) $. 
\end{proposition}

\begin{proof}
The second claim follows from the 
Kolmogorov-Chentsov tightness criterion \cite[Problem 4.11, p.64]{KaratzasShreve1988} derived from Equation \eqref{Eq:DiffMomentBound} and the tightness of $\Vnz_\eps$ that follows from Lemma \ref{Lem:VMoments}.

It remains to prove \eqref{Eq:DiffMomentBound}.  By the definition of $ \Vnz$,   
\begin{align*}
\Vnz_{t_2} - \Vnz_{t_1}
=   \int_{-\infty}^0\int_{\mathbb{R}^2} \bigg[ \phi_n(t_2-s)-\phi_n(t_1-s)  \bigg]    \frac{v}{\sqrt{n}}\, \widetilde{N}(n\cdot ds,du,dv).
\end{align*}
Using the BDG inequality, Lemma \ref{Lem:BDG} and  our moment estimates from Lemma \ref{Lem:VMoments}, we see that there exists a constant $C>0$ depending only on $p$ such that 
\begin{align*}
\mathbb{E}\Big[ \big| \Vnz_{t_2} - \Vnz_{t_1} \big|^{2p} \Big] 
&\leq  C \Big|\int_{-\infty}^0 \bigg| \phi_n(t_2-s)-\phi_n(t_1-s) \bigg|^2 ds    \Big|^p \cr
\\
& \qquad +\frac{C}{n^{p-1}} \int_{-\infty}^0 \bigg|\phi_n(t_2-s)-\phi_n(t_1-s)   \bigg|^{2p}ds  . 
\end{align*}
From the second item of Assumption \ref{Ass:Kernel} we conclude that 
$$
    | \phi_n(t_2-s)-\phi_n(t_1-s)   | \leq C\cdot  (t_1-s )^{H-\frac{3}{2}}\cdot |t_2-t_1| \quad \mbox{for any} \quad s\leq 0 
$$    
and so
\begin{align*}
\mathbb{E}\Big[ \big| \Vnz_{t_2} - \Vnz_{t_1}  \big|^{2p} \Big] 
&\leq  C \cdot |t_2-t_1|^{2p} \cdot \Big|\int_{-\infty}^0  \Big(t_1-s\Big)^{2H-3} ds\Big|^p
\\
&\qquad + \frac{C \cdot |t_2-t_1|^{2p}}{n^{p-1}}\cdot \int_{-\infty}^0  \Big(t_1-s\Big)^{p(2H-3)}  ds 
\\ &\leq  C \cdot |t_2-t_1|^{2p}\eps^{(2H-2)p} +  \frac{C \cdot |t_2-t_1|^{2p}}{n^{p-1}}\eps^{p(2H-3)+1}.
\end{align*}
\end{proof}

\begin{proposition}
We have $ \sup\limits_{t\in[0,\epsilon]}\big| \Vnz_t \big|\overset{\mathbb{P}}\to 0$ as $n\to\infty$ and then $\epsilon\to 0+$. 
\end{proposition}
\begin{proof}
Firstly, by using  the fact that
\begin{align*}
\phi_n(t-s)-\phi_n(-s)=\int_{-s}^{t-s} \phi_n'(r)\,dr, \qquad s\le 0
\end{align*}
and then the stochastic Fubini theorem, we can write $\Vnz $ as
\begin{align*}
\Vnz_t
&= \int_{-\infty}^{0}\int_{\mathbb{R}^2} \bigg[ \int_{-s}^{t-s} \phi_n'(r)dr  \bigg]     \frac{v}{\sqrt{n}}\, \widetilde{N}(n\cdot ds,du,dv)\cr
&= \int_{0}^{\infty}   \int_{-r}^{(t-r)\wedge 0}  \int_{\mathbb{R}^2}    \frac{v}{\sqrt{n}}\, \widetilde{N}(n\cdot ds,du,dv) \cdot   \phi_n'(r) \, dr.
\end{align*}
We split the last triple integral into the following  two terms: 
\begin{align*}
A^{(n)}_{\epsilon} (t)
&\coloneqq \int_{0}^{\epsilon}   \int_{-r}^{(t-r)\wedge 0}  \int_{\mathbb{R}^2}    \frac{v}{\sqrt{n}}\, \widetilde{N}(n\cdot ds,du,dv) \cdot   \phi_n'(r)\, dr,\\
B^{(n)}_{\epsilon} (t)
&\coloneqq \int_{\epsilon}^\infty   \int_{-r}^{(t-r)\wedge 0}  \int_{\mathbb{R}^2}    \frac{v}{\sqrt{n}}\, \widetilde{N}(n\cdot ds,du,dv) \cdot   \phi_n'(r) \, dr.
\end{align*}
It suffices to prove that as $n\to\infty$ and $\epsilon\to 0+$,
\begin{equation*}
\sup_{t\in[0,\epsilon]}| A^{(n)}_{\epsilon} (t)| + \sup_{t\in[0,\epsilon]}| B^{(n)}_{\epsilon} (t)| \overset{\mathbb{P}}\to 0 . 
\end{equation*}
For convenience, we denote the stochastic integral in both $A^{(n)}_{\epsilon}$ and $B^{(n)}_{\epsilon}$ by $ M^{(n)}(r,t)$, i.e., 
\begin{align*}
M^{(n)}(r,t):= \int_{-r}^{(t-r)\wedge 0}  \int_{\mathbb{R}^2} \frac{v}{\sqrt{n}}\, \widetilde{N}(n\cdot ds,du,dv), \qquad t,r\ge 0. 
\end{align*}
 Using BDG inequality, Lemma \ref{Lem:BDG}, and the sub-Gaussion condition of Assumption \ref{Ass:Compensator} for all $p\ge 1$ we have uniformly in $r\geq 0$,
 \begin{equation}
     \label{Eq:Mn_BDG}
\begin{aligned}
\mathbb{E}\Big[\sup_{t\in[0,\epsilon]}\big| M^{(n)}(r,t)\big|^{2p}  \Big] 
&\le  C\cdot  \Big| \int_{-r}^{(\epsilon-r)\wedge 0}  \int_{\mathbb{R}^2}    \frac{|v|^2}{n}\, n\cdot dr \mathcal{P}(du,dv)\Big|^{p}\\
&\qquad + C\cdot   \int_{-r}^{(\epsilon-r)\wedge 0}  \int_{\mathbb{R}^2}    \frac{|v|^{2p}}{n^p}\, n\cdot dr \mathcal{P}(du,dv)
\\
&
\leq C\Big((r\wedge\epsilon )^p+ \frac{r\wedge\epsilon}{n^{p-1}} \Big).
\end{aligned}
\end{equation}

\paragraph{The term $A^{(n)}_{\epsilon}$.}
We let $p>\frac1H$ and introduce a constant  $\theta \in \big(\frac{p(1-2H)+1}{p(3-2H)}, \frac{p-1}{p(3-2H)} \big)$.
By Condition \ref{Cond:Bound2} in Assumption \ref{Ass:Kernel}  we first have   
\begin{align*}
\sup_{t\in[0,\epsilon]} \big| A^{(n)}_{\epsilon} (t) \big|
&\lesssim \int_{0}^{\epsilon}  \sup_{t\in[0,\epsilon]} \big| M^{(n)}(r,t)\big| \cdot   \Big( \frac{1}{n}+r \Big)^{H-\frac{3}{2}} \, dr
\\
&\lesssim \int_{0}^{\epsilon}  \sup_{t\in[0,\epsilon]} \big| M^{(n)}(r,t)\big| \cdot   \Big( \frac{1}{n}+r \Big)^{\theta(H-\frac{3}{2})} \cdot \Big( \frac{1}{n}+r \Big)^{(1-\theta)(H-\frac{3}{2})} \, dr .
\end{align*}
Applying H\"older's inequality to the last integral, then taking expectations on both sides of the preceding inequality and using Fubini's theorem,
\begin{align*}  
&\mathbb{E}\Big[ \sup_{t\in[0,\epsilon]} \big| A^{(n)}_{\epsilon} (t) \big|^{2p} \Big]\\ 
&\leq \int_{0}^{\epsilon} \mathbb{E}\bigg[ \sup_{t\in[0,\epsilon]} \Big| M^{(n)}(r,t)\Big|^{2p}  \bigg]   \Big( \frac{1}{n}+r \Big)^{p\theta(2H-3)}dr
\cdot  \Big|\int_{0}^{\epsilon} \Big( \frac{1}{n}+r \Big)^{(1-\theta)(H-\frac{3}{2})\frac{2p}{2p-1}} \, dr \Big|^{2p-1}.
\end{align*}

Since $(1-\theta)(H-\frac{3}{2})\frac{2p}{2p-1}+1>0$,
the last term on the right hand side of this inequality can be bounded uniformly in $n \in \mbN$ and  $\epsilon\in (0,1)$.  
Plugging \eqref{Eq:Mn_BDG} into the right side and then using the inequality $$(x+y)^{p\theta(2H-3)}\leq x^{p\theta(2H-3)}\wedge y^{p\theta(2H-3)}$$ for any $x,y>0$, we have 
\begin{align*}
\mathbb{E}\Big[ \sup_{t\in[0,\epsilon]} \big| A^{(n)}_{\epsilon} (t) \big|^{2p} \Big]
&\leq  C \int_{0}^{\epsilon} \Big(r^p+ \frac{r}{n^{p-1}} \Big) \cdot   \Big( \frac{1}{n}+r \Big)^{p\theta(2H-3)}dr\\
&\leq  C \int_{0}^{\epsilon}   r^{p\theta(2H-3)+p}dr  +  C \int_{0}^{\epsilon}  \frac{r}{n^{p\theta(2H-3)+p-1}}   \, dr\\
&\leq  C\cdot \Big( \epsilon^{p\theta(2H-3)+p+1} +  \frac{\epsilon^2}{n^{p\theta(2H-3)+p-1}} \Big).
\end{align*}
As $n\to\infty$ and $\epsilon\to 0+$, the r.h.s.\ tends to $0$. 
The last inequality holds because $p\theta(2H-3)+p>1$. 

\paragraph{The term $B^{(n)}_{\epsilon}$.} Similarly as in the previous argument we choose $p>\frac1H$ and take some constant 
$\theta\in \big(\frac{H}{3-2H},\frac{1-H}{3-2H}\big) \subset [0,1]$
Using Condition \ref{Cond:Bound2} of Assumption \ref{Ass:Kernel},  
\begin{align*}
\sup_{t\in[0,\epsilon]} \big| B^{(n)}_{\epsilon} (t) \big|
&\lesssim \int_{\epsilon}^\infty  \sup_{t\in[0,\epsilon]} \big| M^{(n)}(r,t)\big|    \Big( \frac{1}{n}+r \Big)^{H-\frac{3}{2}} \, dr\\
&\lesssim \int_{0}^{\epsilon}  \sup_{t\in[0,\epsilon]} \big| M^{(n)}(r,t)\big|    \Big( \frac{1}{n}+r \Big)^{\theta(H-\frac{3}{2})}  \Big( \frac{1}{n}+r \Big)^{(1-\theta)(H-\frac{3}{2})} \, dr.
\end{align*}
Taking expectation and using H\"older inequality, this yields
\begin{align*}
\mathbb{E}\Big[  \sup_{t\in[0,\epsilon]} \big| B^{(n)}_{\epsilon} (t)\big|^{2p} \Big] 
&\lesssim \int_{\epsilon}^{\infty}   \mathbb{E}\Big[\sup_{t\in[0,\epsilon]} \big| M^{(n)}(r,t)\big|^{2p} \Big]   \Big( \frac{1}{n}+r \Big)^{p\theta(2H-3)} dr \\
& \qquad \cdot \bigg|\int_{\epsilon}^{\infty} \Big( \frac{1}{n}+r \Big)^{(1-\theta)(H-\frac{3}{2})\frac{2p}{2p-1}}   \, dr \bigg|^{2p-1}.
\end{align*}

Since $(H-\frac{3}{2})(1-\theta)\frac{2p}{2p-1}+1<0 $, there is a constant $C>0$ that is independent of $n$ and $\epsilon$ such that
\begin{equation*}
\bigg|\int_{\epsilon}^{\infty} \Big( \frac{1}{n}+r \Big)^{(1-\theta)(H-\frac{3}{2})\frac{2p}{2p-1}}   \, dr \bigg|^{2p-1}
\leq C  \Big( \frac{1}{n}+\epsilon \Big)^{(1-\theta)p(2H-3)+2p-1} . 
\end{equation*}

By \eqref{Eq:Mn_BDG} and the fact that $\theta p(2H-3)+1< 0$ we have
\begin{align*}
\int_{\epsilon}^{\infty}   \mathbb{E}\Big[\sup_{t\in[0,\epsilon]} \big| M^{(n)}(r,t)\big|^{2p} \Big]   \Big( \frac{1}{n}+r \Big)^{p\theta(2H-3)} dr 
&\leq 
C \Big( \epsilon^p + \frac{\epsilon}{n^{p-1}} \Big)\int_{\epsilon}^{\infty}     \Big( \frac{1}{n}+r \Big)^{p\theta(2H-3)} dr
\\
&\leq C \Big( \epsilon^p + \frac{\epsilon}{n^{p-1}} \Big)   \Big( \frac{1}{n}+ \epsilon \Big)^{p\theta(2H-3)+1}.
\end{align*}
Combining the previous estimates and then using $(x+y)^{p(2H-1)}\leq x^{p(2H-1)}\wedge y^{p(2H-1)}$ for any $x,y>0$, again, 
\begin{align*}
\mathbb{E}\Big[  \sup_{t\in[0,\epsilon]} \big| B^{(n)}_{\epsilon} (t)\big|^{2p} \Big] 
\leq  C  \Big( \epsilon^p + \frac{\epsilon}{n^{p-1}} \Big) \cdot \Big( \frac{1}{n}+\epsilon \Big)^{ p(2H-1)} 
\leq   C  \Big( \epsilon^{2Hp} + \frac{\epsilon}{n^{2Hp-1}} \Big).
\end{align*}
The last expression vanishes as $n\to\infty$ and $\epsilon\to 0+$ whenever $p\geq \frac{1}{2H}$.
\end{proof}

\subsection[C-tightness of Vn]{$C$-tightness of $\{  \Vno \}$}\label{Sec:Tight_V1}

In this section, we prove the $C$-tightness of the sequence $\{  \Vno \}_{n \in \mbN}$ by using the following criterion given in \cite{HorstXuZhang1}. For a real number $x>0$ we denote by $\lfloor x\rfloor$ the integer part of $x$, i.e.\ $\lfloor x\rfloor=\sup\{n \in \mathbb{Z}: n \le x\}$.

\begin{lemma}[\cite{HorstXuZhang1}, Lemma 3.5] \label{Lem:tightness_condition}
A sequence of processes $\{ X^{(n)} \}_{n \in \mbN} \subset \mathbb{D}([0,T];\mathbb{R})$ with $\{X^{(n)}_0\}_{n \in \mbN}$ being tight in $\mathbb{R}$,
is $C$-tight in $\mathbb{D}([0,T];\mathbb{R})$ if 
for any $T\geq 0$ and some constant $\theta>2$, the following two  conditions hold.
\begin{enumerate}
\item[(1)] $\displaystyle\sup_{k=0,1,\cdots,\lfloor Tn^\theta\rfloor} \sup_{h\in[0,1/n^{\theta}]} \big|X^{(n)}_{k/n^\theta+h}-  X^{(n)}_{k/n^\theta}\big|  \overset{\mathbb{P}}\to 0$  as $n\to\infty$.

\item[(2)] There exist some constants $C >0 $, $p\geq 1$, $m\in\{1,2, ...\}$ and pairs $\{ (a_i,b_i) \}_{i=1,\cdots ,m}$ satisfying 
\begin{equation*}
a_i\geq 0,\quad b_i>0,\quad \min_{1\leq i\leq m} \big\{ b_i + a_i/\theta \big\}>1 , 
\end{equation*}
such that for all $n \in \mbN$ and $h \in (0,1)$,
\begin{equation*} %
\sup_{t\in [0,T]}	\mathbb{E}\Big[\big|  X^{(n)}_{t+h}-X^{(n)}_t\big|^{p}\Big]\leq 
C\cdot \sum_{i=1}^m \frac{h^{b_i}}{n^{a_i}}. 
\end{equation*} 
\end{enumerate} 
\end{lemma}

In the next two propositions, we verify that the sequence $\{  \Vno \}_{n \in \mbN}$ satisfies the two conditions of Lemma~\ref{Lem:tightness_condition}. Before we do so we need the following remark.

\begin{remark}\label{Rem:KerL2_Diff}
Conditions \ref{Cond:Bound1} and \ref{Cond:Bound2} of Assumption \ref{Ass:Kernel} already imply that there is some $C>0$ such that for all $h \in [0,1]$
\begin{align*}
\int_0^T \bigl( \phi_n(s)-\phi_n(s+h)\bigr)^2ds\le C h^{2H}.
\end{align*}
Indeed, under these conditions direct calculation shows that
\begin{align*}
\int_0^T \bigl( \phi_n(s)-\phi_n(s+h)\bigr)^2ds&\le 2\int_0^h \phi_n(s)^2ds+\int_{h}^T \biggl( \int_s^{s+h} \phi_n'(r)\,dr\biggr)^2ds
\\
&\le C \int_0^h \Bigl(\frac1n+s\Bigr)^{2H-1}\,ds+C\int_h^T h^2 \Bigl(\frac1n+s\Bigr)^{2H-3}\,ds
\\
&\le C h^{2H}.
\end{align*}
\end{remark}

\begin{lemma}
For any $p\geq \frac{1}{H}$, there exists a constant $C>0$ such that for any $n\in \mbN$ and $t,h \ge 0$ such that $t+h \in [0,T]$,
\begin{equation*}%
\mathbb{E}\Big[ \big| \Vno_{t+h}- \Vno_t \big|^{2p}  \Big]
\leq C\cdot \Big( h^{2Hp} + \frac{h}{n^{2Hp-1}}   +\frac{h^{2H}}{n^{ 2H(p-1) }}  \Big). 
\end{equation*}
Moreover,	Condition~(2) in Lemma~\ref{Lem:tightness_condition} holds whenever $2<\theta<\frac{2H(p-1)}{1-2H}$. 
\end{lemma}
\begin{proof}
By  definition of $\Vno$ we can write $\Vno_{t+h} - \Vno_t= I_1^{(n)}(t,h) +I_2^{(n)}(t,h)$ with
\begin{align*}
I_1^{(n)}(t,h)&\coloneqq \int_{t}^{t+h} \int_{\mathbb{R}^2} \phi_n(t+h-s)\frac{v}{\sqrt{n}}\, \widetilde{N}(n\cdot ds,du,dv),\cr
I_2^{(n)}(t,h)&\coloneqq  \int_0^{t} \int_{\mathbb{R}^2} \Big( \phi_n(t+h-s)-\phi_n(t-s)\Big)\cdot \frac{v}{\sqrt{n}}\, \widetilde{N}(n\cdot ds,du,dv).
\end{align*}
 We use the BDG inequality from Lemma \ref{Lem:BDG}, our moment estimates from Lemma \ref{Lem:VMoments}, change variables and use  Condition \ref{Cond:Bound2} in Assumption \ref{Ass:Kernel} to obtain that
\begin{align*}
\mathbb{E}\Big[ \big|I_1^{(n)}(t,h)\big|^{2p}  \Big] 
&\leq C \Big| \int_{t}^{t+h}  \phi_n^2(t+h-s)\, ds \cdot  \int_{\mathbb{R}^2}|v|^2 \mathcal{P}(du,dv) \Big|^{p}  \\
&\qquad+   C  \int_{t}^{t+h}  \phi_n^{2p}(t+h-s)ds \cdot \int_{\mathbb{R}^2}  \frac{|v|^{2p}}{n^{p-1}}\,\mathcal{P}(du,dv) \cr
&\leq  C \Big| \int_0^{h}  s^{2H-1}\, ds  \Big|^{p}   +   \frac{C}{n^{p-1}}\cdot  \int_0^{h} \Big( \frac{1}{n}+s \Big)^{p(2H-1)}ds \cr
&\leq C \Big( h^{2Hp} + \frac{h}{n^{2Hp-1}} \Big).  
\end{align*}
A similar calculation shows that
\begin{align*}
\mathbb{E}\Big[ \big|I_2^{(n)}(t,h)\big|^{2p}  \Big] 
&\leq  C \Big|\int_0^{t} \Big(\phi_n(t+h-s)-\phi_n(t-s)\Big)^2ds  \cdot\int_{\mathbb{R}^2}   |v|^2 \, \mathcal{P}(du,dv)\Big|^p
\\
&\qquad + C\int_0^{t}  \Big(\phi_n(t+h-s)-\phi_n(t-s)\Big)^{2p} ds\cdot \int_{\mathbb{R}^2}\frac{|v|^{2p}}{n^{p-1}}\,  \mathcal{P}(du,dv)
\\
&\leq  C \Big|\int_0^{t} \Big(\phi_n(t+h-s)-\phi_n(t-s)\Big)^2ds   \Bigr|^p
\\
&\qquad + \frac{C}{n^{p-1}}\int_0^{t}  \Big(\phi_n(t+h-s)-\phi_n(t-s)\Big)^{2p} ds.
\end{align*}
By Condition \ref{Cond:Bound1} of Assumption \ref{Ass:Kernel} it follows that 
\begin{align*}
\int_0^{t}\Big(\phi_n(t+h-s)-\phi_n(t-s)\Big)^{2p} ds\le C n^{(1/2-H)(2p-2)}\int_0^{t}  \Big(\phi_n(t+h-s)-\phi_n(t-s)\Big)^{2} ds.
\end{align*} 
Using that $n^{1-p}\cdot n^{(1/2-H)(2p-2)}=n^{-2H(p-1)}$ together with Remark \ref{Rem:KerL2_Diff} it follows that
\begin{align*}
\mathbb{E}\Big[ \big|I_2^{(n)}(t,h)\big|^{2p}  \Big]  \le C h^{2Hp}+C\frac{h^{2H}}{n^{2H(p-1)}}.
\end{align*}
\end{proof}

\begin{lemma}
For each $\theta>2$ and $t\in [0,T]$, we have as $n\to\infty$,
\begin{equation}\label{Eq:Vn_conv}
\sup_{k=0,\cdots,\lfloor n^\theta t \rfloor} \sup_{h\leq n^{-\theta}} \big|  \Vno_{h+k/n^\theta} -\Vno_{k/n^\theta}\big|\overset{\mathbb{P}}\to 0.
\end{equation} 
\end{lemma}
\begin{proof}
The  equality  
$\phi_n(t-s)=\phi_n(0)-\int_s^t \phi_n'(t-r)\,dr
$ along with the stochastic Fubini theorem allows us to rewrite $\Vno$ as 
\begin{equation*}
    \Vno_t=  \Vnoo_t-{\Vnot_t} ,\quad t\geq 0,
\end{equation*}
 with the decmoposition
\begin{align}
\Vnoo_t&\coloneqq   \int_0^{t} \int_{\mathbb{R}^2}   \frac{v\phi_n(0)}{\sqrt{n}}\, \widetilde{N}(n\cdot ds,du,dv),\cr 
\Vnot_t&\coloneqq \int_0^{t} \phi_n'(t-r) \int_0^r\int_{\mathbb{R}^2}   \frac{v}{\sqrt{n}}\, \widetilde{N}(n\cdot ds,du,dv)\, dr. \label{Eq:V12}
\end{align}
The triangle inequality tells us that the limit \eqref{Eq:Vn_conv} holds if for $i =1,2$
\begin{equation*}%
\sup_{k=0,\cdots,\lfloor n^\theta t \rfloor} \sup_{h\leq n^{-\theta}} \big|  \Vnoi_{h+k/n^\theta} -\Vnoi_{k/n^\theta}\big|\overset{\mathbb{P}}\to 0.
\end{equation*}

\paragraph{Convergence for $\Vnoo$.}
For each $\eta>0$ and some constant $p>1$ to be specified later, by Chebyshev's inequality, we have 
\begin{align*}
\mathbb{P} \Big(\sup_{k=0,\cdots,\lfloor n^\theta t \rfloor} \sup_{h\leq n^{-\theta}} &\big|  \Vnoo_{h+k/n^\theta} -\Vnoo_{k/n^\theta}\big| \geq \eta \Big)
\\
&\le \sum_{k=0}^{\lfloor n^\theta t \rfloor}\mathbb{P} \Big( \sup_{h\leq n^{-\theta}} \big|  \Vnoo_{h+k/n^\theta} -\Vnoo_{k/n^\theta}\big| \geq \eta \Big)
\\
&\le \sum_{k=0}^{\lfloor n^\theta t \rfloor} \frac{1}{\eta^{2p}} \mathbb{E} \Big[ \sup_{h\leq n^{-\theta}} \big|  \Vnoo_{h+k/n^\theta} -\Vnoo_{k/n^\theta}\big|^{2p} \Big].
\end{align*}
Notice that  $\Vnoo$ is a martingale, hence by BDG inequality,  Lemma \ref{Lem:BDG},   we have 
\begin{align*}
\mathbb{E} \Big[ \sup_{h\leq n^{-\theta}} \big|  \Vnoo_{h+k/n^\theta} -\Vnoo_{k/n^\theta}\big|^{2p} \Big] 
& \le C \Biggl|\int_{kn^{-\theta}}^{h+kn^{-\theta}} \int_{\mathbb{R}}\Bigl(\frac{\phi_n(0)}{\sqrt{n}}\Bigr)^2 v^2 n\,\mathcal{P}(du,dv)\,ds \Biggr|^p
\\
&  +C \int_{kn^{-\theta}}^{h+kn^{-\theta}} \int_{\mathbb{R}}\Bigl(\frac{\phi_n(0)}{\sqrt{n}}\Bigr)^{2p} v^{2p} n\,\mathcal{P}(du,dv)\,ds
\\
&  \le C \Bigl(\frac{\phi_n(0)}{\sqrt{n}}\Bigr)^{2p}\Big( n^{p(1-\theta)}+n^{1-\theta} \Big). 
\end{align*}
By Condition \ref{Cond:Bound1} of Assumption \ref{Ass:Kernel} we have $\phi_n(0)\le C n^{1/2-H}$
and hence 
\begin{align*}
\mathbb{P} \Big(\sup_{k=0,\cdots,\lfloor n^\theta t \rfloor} \sup_{h\leq n^{-\theta}} \big|  \Vnoo_{h+k/n^\theta} -\Vnoo_{k/n^\theta}\big| \geq \eta \Big)
&\leq C  n^\theta  \frac1{\eta^{2p}} n^{-2Hp}  \Big( n^{p(1-\theta)}+n^{1-\theta} \Big) 
\\
&\leq \frac{C}{\eta^{2p}} \cdot\Big( n^{p(1-\theta)+\theta-2Hp}+n^{1-2Hp}  \Big) ,
\end{align*}
which goes to $0$ as $n\to\infty$ whenever $p>\frac{1}{2H} \vee \frac{\theta}{2H+\theta-1}=\frac{1}{2H}$. 

\paragraph{Convergence for $\Vnot$.}
We denote the inner stochastic integral in \eqref{Eq:V12} by $M^{(n)}_r$, i.e.,
\begin{align*}
M^{(n)}_r:= \int_0^r\int_{\mathbb{R}^2}   \frac{v}{\sqrt{n}}\, \widetilde{N}(n\cdot ds,du,dv),
\end{align*}
which is a martingale with $ \mathbb{E} [ \sup\limits_{r\in[0,T]}  |M^{(n)}_r  |^2  ]
\leq C \int_{\mathbb{R}^2} |v|^2 \mathcal{P}(du,dv)<\infty$. We can then decompose the difference $ \Vnot_{k/n^\theta+h} -\Vnot_{k/n^\theta}$ as
\begin{align*}
\Vnot_{k/n^\theta+h} -\Vnot_{k/n^\theta}&=\int_{k/n^\theta}^{ k/n^\theta+h}
\phi_n'(kn^{-\theta}+h-r)
M^{(n)}_r\, dr
\\
&\qquad+ \int_0^{ k/n^\theta} \Big( \phi_n'(kn^{-\theta}+h-r)-\phi_n'(kn^{-\theta}-r)\Bigr)  M^{(n)}_r\, dr.
\end{align*}
Consequently, by Condition \ref{Cond:Bound2} of Assumption \ref{Ass:Kernel} and the monotonicity of the power function it follows that
\begin{align*}
\sup_{k\le \lfloor n^\theta T\rfloor} \sup_{h\leq n^{-\theta}} \big|  \Vnot_{k/n^\theta +h} -&\Vnot_{k/n^\theta}\big|
\le  C \sup_{t\in[0,T]} \big| M^{(n)}_r \big|\cdot  \sup_{h\leq n^{-\theta}}\int_{0}^{ h} \Big( \frac{1}{n}+r\Big)^{H-\frac{3}{2}} \, dr
\\
& + C \sup_{t\in[0,T]} \big| M^{(n)}_r \big| \biggl| \int_0^{ k/n^\theta} \Big( \phi_n'(kn^{-\theta}+h-r)-\phi_n'(kn^{-\theta}-r)\Bigr) \, dr\biggr|.
\end{align*}
Since $\sup_{n \in \mbN}\mathbb{E} [ \sup_{r\in[0,T]}  |M^{(n)}(r)  |^2  ]<\infty$, it suffices to prove that the last two integrals converge to $0$ as $n\to\infty$.
The first integral is bounded by $n^{3/2-H-\theta}$, which converges to $0$ whenever  $\theta>\frac{3}{2}-H$. The second integral converges to $0$ by Condition \ref{Cond:Bound3} of Assumption \ref{Ass:Kernel}.
\end{proof}

\subsection{Convergence of Finite Dimensional Distributions}
\label{Sec:FinDimConv}

The following lemma establishes the convergence of finite dimensional distributions. For this, we recall the processes $\Wn$ from Equation \eqref{Eq:Wn} as well as the fractional Brownian motion $\fBM$ and the Brownian motion $W$ from Theorem \ref{Thm:WeakConv}.

\begin{lemma}\label{Lem:FDIM_Conv}

For any $m \in \mbN$, $0\le t_1<\dots<t_m \le T$ and $z_1,\dots,z_m,\tilde z_1,\dots,\tilde z_m\in \mathbb{R}$ we have
\begin{align*}
    \mbE\Bigl[\exp\Bigl( \sum_{k=1}^m \mathtt{i}\bigl( z_k \Wn_{t_k}+\tilde z_k \Vn_{t_k} \big)   \Bigr) \Bigr] \xrightarrow{n\rightarrow \infty}   \mbE\Bigl[\exp\Bigl( \sum_{k=1}^m \mathtt{i}\bigl( z_k \sigma_p W_{t_k}+\tilde z_k\sigma_v c_H \fBM_{t_k} \big)   \Bigr) \Bigr]
\end{align*}
\end{lemma}

\begin{proof}
Using the convention $t_0=0$ we can rewrite
\begin{align*}
    \sum_{k=1}^m \Bigl(z_k \Wn_{t_k}+\tilde z_k \Vn_{t_k}\Bigr) = ~ & \int_{-\infty}^0\sum_{k=1}^m \tilde z_k\bigl(\phi_n(t_k-s)-\phi_n(-s)\bigr)\frac v{\sqrt{n}} \tilde N(n\cdot ds,du,dv)
    \\
    &+\sum_{j=0}^{m-1} \int_{t_{j}}^{t_{j+1}}\sum_{k=j+1}^{m}\int_{\mbR^2} z_k \frac{u}{\sqrt{n}}+\tilde z_k \phi_n(t_k-s) \frac{v}{\sqrt{n}}\tilde N(n\cdot ds,du,dv).
\end{align*}
Taking exponentials and using the independence of increments, we observe that
\begin{align*}
     \mbE\Bigl[\exp\Bigl(& \sum_{k=1}^m \mathtt{i}\bigl( z_k \Wn_{t_k}+\tilde z_k \Vn_{t_k}    \Bigr) \Bigr] 
     \\
     &=\mbE\biggl[ \exp\Bigl(\mathtt{i}\int_{-\infty}^0\sum_{k=1}^m \tilde z_k\bigl(\phi_n(t_k-s)-\phi_n(-s)\bigr)\frac v{\sqrt{n}} \tilde N(n\cdot ds,du,dv)
    \Bigr)
    \biggr]
    \\
    &\quad \cdot 
    \prod_{j=0}^{m-1}\mbE\biggl[ \exp\Bigl(\mathtt{i}
        \int_{t_{j}}^{t_{j+1}}\sum_{k=j+1}^{m}\int_{\mbR^2}z_k \frac{u}{\sqrt{n}}+ \tilde z_k \phi_n(t_k-s) \frac{v}{\sqrt{n}}\tilde N(n\cdot ds,du,dv)
    \Bigr)
    \biggr].
\end{align*}

By Assumption \ref{Cond:Bound1} we know that $\phi_n \lesssim \sqrt{n}$. A Taylor expansion yields existence of constants $C_2,C_3>0$ such that
\begin{align*}
    \prod_{j=0}^{m-1}\mbE\biggl[& \exp\Bigl(\mathtt{i}
        \int_{t_{j}}^{t_{j+1}}\sum_{k=j+1}^{m}\int_{\mbR
        ^2} z_k \frac{u}{\sqrt{n}}+ \tilde z_k \phi_n(t_k-s) \frac{v}{\sqrt{n}}\tilde N(n\cdot ds,du,dv)
    \Bigr)
    \biggr]
    \\
    &=\exp\biggl(
    \sum_{j=0}^{m-1}  \int_{t_{j}}^{t_{j+1}} \int_{\mbR^2} n\Bigl(\exp\Bigl(\mathtt{i}\sum_{k=j+1}^{m} z_k \frac{u}{\sqrt{n}}+ \tilde z_k \phi_n(t_k-s) \frac{v}{\sqrt{n}}\Bigr)-1\Bigr) \,\mcP(du,dv)ds
    \biggr)
    \\
    &=\exp\biggl(
    \sum_{j=0}^{m-1} \int_{t_{j}}^{t_{j+1}} \int_{\mbR^2} 
     \sqrt{n}\Bigl( \mathtt{i}\sum_{k=j+1}^{m}z_k u +\tilde z_k \phi_n(t_k-s) v\Bigr) 
    -\frac{1}2\Bigl(\sum_{k=j+1}^{m}z_k u + \tilde z_k \phi_n(t_k-s)v \Bigr)^2
     \\
     &\qquad \qquad\qquad \qquad \qquad +\mathcal{O}\Bigl\{\frac1{\sqrt{n}} e^{C_2 u+C_3 v}\Bigl|\sum_{k=j+1}^{m}z_k u + \tilde z_k \phi_n(t_k-s) v\Bigr|^3\Bigr\}
     \,\mcP(du,dv)ds
    \biggr).
\end{align*}

By Assumption \ref{Ass:Compensator} the integral of the first term in the exponent vanishes. In addition, the third term vanishes as $n \rightarrow \infty$, because by Assumption \ref{Ass:Kernel} 
$$
    \int_0^t |\phi_n(s)|^3\,ds\le C \int_0^t (n^{-1}+s)^{3H-3/2}\,ds \lesssim n^{1/2-3H}
$$    
Furthermore, by Assumption \ref{Ass:Compensator}, the middle term, after expanding the square, satisfies
\begin{align*}
   \int_{t_j}^{t_{j+1}} \int_{\mbR^2} 
      &\frac{1}2\Bigl(\sum_{k=j+1}^{m}z_k \Bigr)^2 u^2
      +
      \Bigl(\sum_{k=j+1}^{m}z_k \Bigr) \Bigl(\sum_{k=j+1}^{m}\tilde z_k \phi_n(t_k-s) \Bigr)uv
      +\frac{1}2\Bigl(\sum_{k=j+1}^{m}\tilde z_k \phi_n(t_k-s) \Bigr)^2 v^2
     \,\mcP(du,dv)ds
     \\
     &= \int_{t_j}^{t_{j+1}} 
 \frac{1}2\Bigl(\sum_{k=j+1}^{m}z_k \Bigr)^2 \sigma_p^2+
      \Bigl(\sum_{k=j+1}^{m}z_k \Bigr) \Bigl(\sum_{k=j+1}^{m}\tilde z_k \phi_n(t_k-s) \Bigr) \rho\sigma_p\sigma_v+
      \frac{1}2\Bigl(\sum_{k=j+1}^{m}\tilde z_k \phi_n(t_k-s) \Bigr)^2\sigma_v^2 \,ds
\end{align*}

A similar calculation yields
\begin{align*}
   \lim_{n\rightarrow\infty} \mbE\biggl[ \exp\Bigl( \mathtt{i}\int_{-\infty}^0\sum_{k=1}^m &\tilde z_k\bigl(\phi_n(t_k-s)-\phi_n(-s)\bigr)\frac v{\sqrt{n}} \tilde N(n\cdot ds,du,dv)
    \Bigr)
    \biggr]
    \\
    &=\exp\biggl( - \frac12\int_{-\infty}^0 \Bigl( \sum_{k=1}^m \tilde{z_k} \bigl( \phi_n(t_k-s)-\phi_n(-s)\bigr)\Bigr)^2\sigma_v^2\,ds\biggl).
\end{align*}
Combining these terms yields the characteristic function of $ \sum\limits_{k=1}^m\bigl( z_k \sigma_p W_{t_k}+\tilde z_k\sigma_v c_H \fBM_{t_k}    \Bigr)$.
\end{proof}

\subsection{Proof of Theorem \ref{Thm:WeakConv}}
\label{Sec:PrfWeakConv}

The martingales  $\{\Wn \}_{n \in \mbN}$ given in \eqref{Eq:Wn} satisfy the Lindeberg condition and therefore converge weakly to $\sigma_p\cdot W$ in $\mathbb{D}([0,T];\mathbb{R})$; see Theorem~5.4 in \cite{JacodShiryaev2003}. 
Moreover, the (predictable) covariation process of $\Wn$ satisfies  $\langle \Wn,\Wn \rangle_t=\sigma_p^2 t $ and therefore  $\{\Wn \}_{n \in \mbN}$ is uniformly tight; see Proposition~6.13 in \cite[p.379]{JacodShiryaev2003}. 

By Lemma \ref{Lem:FDIM_Conv} there are Brownian motions $B$ and $W$ with correlation $\rho$ such that $\Wn \rightarrow \sigma_pW$ and $\Vn \rightarrow \sigma_v c_H \fBM$ weakly in $\mathbb{D}([0,T];\mathbb{R}^2)$ as $n \rightarrow \infty$. Corollary~3.33 in \cite[p.353]{JacodShiryaev2003} implies their joint convergence, i.e. 
$$
    (\Vn,\Wn)\to (\sigma_vc_H \fBM ,\sigma_p W) \quad \mbox{weakly in} \quad \mathbb{D}([0,T];\mathbb{R}^2).
$$

Using Theorem~6.22(b) in \cite[p.383]{JacodShiryaev2003},
we have $\Pn\to \widetilde{P}^*$ with 
\begin{equation*}    
\widetilde{P}^*= \int_0^t \sigma_p e^{-V^*(s)}dW(s),\quad t\geq 0.
\end{equation*}
Another application of Corollary~3.33 in \cite[p.353]{JacodShiryaev2003} then implies that $\bigl( \Vn,\Pn\bigr) \rightarrow \bigl( \Vcont,\Pcont\bigr)$ weakly in $\mathbb{D}([0,T];\mathbb{R}^2)$.

\section{Weak Convergence Rate}\label{Sec:Weak_Conv}

In this section, we prove the weak convergence rate stated in Theorem \ref{Thm:Main_Thm}. To this end, we use a triangular argument: we compare moments of the prelimit model $\Pn$ and the limit model $\Pcont$ with those of an \emph{approximate rough Bergomi model} $\Pncont$ that lies between the two and is defined as follows:
\begin{definition}\label{Def:Pncont_Vncont}
    Let $B$ and $W$ be two Brownian motions with constant correlation $\rho$. For $n \in \mbN$ and $\sigma_p,\sigma_v$ from Assumption \ref{Ass:Kernel} define the approximate fractional Brownian motions
    \begin{align*}
        \fBMn=\frac1{c_H}\biggl(\int_0^t \phi_n(t-s)\,dB_s+\int_{-\infty}^0 \phi_n(t-s)-\phi_n(-s)\,dB_s\biggr),
    \end{align*}
    as well as the associated rough-Bergomi  (log-) price processes
    \begin{align*}
        \Pncont_t&=\sigma_p\int_0^t e^{\Vncont_s}\,dW_s, \qquad \Vncont_t=c_H \sigma_v \fBMn.
    \end{align*}
\end{definition}
The first step is therefore to establish moment formulas for all three models. For the prelimit model, this requires a Clark--Ocone representation for Poisson point processes, as established in \cite{Zhang_2009}, which we recall in Section~\ref{Sec:Clark_Ocone}. After providing a series of {a priori} moment estimates in Section~\ref{Sec:PV_estimates}, we prove an (approximate) moment recursion for $\Pn$ in Lemma~\ref{Lem:Mom_Recursion}. With this representation at hand, we can apply techniques similar to those in \cite{Friz_Salkeld_Wagenhofer_2025}, collected in Appendix~\ref{Sec:AppMoments}, and establish moment formulas for all three models in Theorem~\ref{Thm:MomRep}.
Equipped with the moment formula we perform a term-by-term comparison in Section \ref{Sec:MomBounds}, proving that 
\begin{align}
    \biggl|\mbE\Bigl[ \bigl(\Pn_t\bigr)^N\Bigr]-\mbE\Bigl[ \bigl(\Pncont_t\bigr)^N\Bigr]\biggr|
    &\lesssim \star(n), \nonumber\\\label{Eq:moment_est}    \Bigl|\mbE\Bigl[ \bigl(\Pn\bigr)^N\Bigr]-\Bigl[\bigl(\Pcont\bigr)^N\Bigr]\Bigr|&\lesssim \diamond(n)+\square(n)+\diamond(n)+\star(n).
\end{align}
where, for a sequence $\{\phi_n\}_{n \in \mbN}$ of kernels satisfying Assumption \ref{Ass:Kernel}, we set 
\begin{equation}\label{Eq:Rate_Notation} 
\begin{split}
    \star(n) &\coloneqq \frac1n \int_0^T  \phi_n(T-s)^4 \,ds +\frac1n \int_{-\infty}^0 \Bigl(\phi_n(T-s) -\phi_n(-s)\Bigr)^4\,ds  ,
    \\
    \diamond (n) &\coloneqq \sup_{t\in [0,T]} \int_0^t \phi_\infty(t-s)\Bigr|C_n(s,t)-C_\infty(s,t)\biggr|\,ds,
    \\
    \square(n)&\coloneqq \sup_{t\in [0,T]}\int_0^t \phi_n(t-s) \Bigl| 
    C_n(s,s)-C_\infty(s,s)\Bigr| \,ds,
    \\
    \triangle(n)&\coloneqq \|\phi_n-\phi_\infty\|_{L^1([0,T])},
\end{split}
\end{equation}
and $C_n$ denotes for $n \in\mbN\cup \{\infty\}$ the \emph{covariance} function
\begin{equation} \label{Eq:Cn} 
\begin{split}
    C_n(t,s) & ~ =\int_0^{s\wedge t} \phi_n(t-r)\phi_n(s-r)\,dr 
    \\ & \quad + \int_{-\infty}^0 \Bigl(\phi_n(t-r) -\phi_n(-r)\Bigr) \Bigl(\phi_n(s-r) -\phi_n(-r)\Bigr)\,dr.
\end{split}
\end{equation}

While these estimates are valid for any sequence $\{\phi_n\}_{n \ge 0}$ of kernels satisfying Assumption~\ref{Ass:Kernel}, they do not yield the claimed rate, as already discussed in Remark \ref{Rem:CounterEx}. Therefore, we provide in Section~\ref{Sec:Kernel_Estimates} a family of kernels that attains the optimized rate $\tfrac13+\tfrac{4H}{3-6H}$ from Theorem~\ref{Thm:Main_Thm} by showing that
\begin{align}\label{Eq:rhs_est}
      \star(n)+\diamond(n)+\square(n)+\triangle\lesssim n^{-\frac13-\frac{4H}{3-6H}}.
\end{align}
Indeed, we also provide a lower bound in Lemma \ref{Lem:LowerBd}. As a result, although the estimate in \eqref{Eq:moment_est} may not be sharp, the rate in Equation \eqref{Eq:rhs_est} cannot be improved by choosing better approximating kernels.

For convenience, we also introduce the following terminology.%
\begin{definition}\label{Def:Exp_Form}
We say that $f:\mathbb{R}^n\rightarrow \mathbb{R}$ is of \emph{exponential form} if
\begin{align*}
    f(x_1,\dots,x_n)=\exp(c_1x_1+\dots+c_nx_n)
\end{align*}
for some constants $c_1,\dots,c_n\in \mathbb{R}$. Furthermore, we denote by $\Delta_m^\circ$ the open simplex
\begin{equation}
    \label{Eq:Simplex}
    \Delta_m^\circ = \Big\{ \mathbf{t}=(t_1, ..., t_m) \in [0,T]^{\times m}: 0<t_m<t_{m-1}<\dots<t_1<T \}. 
\end{equation}
\end{definition}

For non-negative sequences $\{a_n\}_{n\ \in \mbN}$, $\{b_n\}_{n\ \in \mbN}$ and $\{c_n\}_{n\ \in \mbN}$ we say $a_n \lesssim b_n$, if there is a constant $C>0$ such that $a_n \le C b_n$ for all $n \in \mbN$. Similarly, we say $a_n = b_n + \mathcal{O}(c_n)$ if $|a_n-b_n| \lesssim c_n$.

\subsection{Clark-Ocone Formula}\label{Sec:Clark_Ocone}

In what follows, we recall a Clark-Ocone formula for Poisson random measures that will be key to our subsequent analysis. 
Following the terminology of \cite{Zhang_2009}, we let $(\Omega,(\mcF_t)_{t \in [0,T]},\mbP)$ be the canonical space of a Poisson random measure $N$ on $[0,T]\times \mbU$ with $\mbU$ being a Polish space. Let $\mu_\omega$ be the canonical measure defined by 
\[
    \mu_\omega(A)=\omega(A), \quad A \in \mathcal{B}([0,T]\times \mbU) 
\]    
with compensator $\nu$. Denote by $\tilde \mu = \mu - \nu$ the compensated Poisson random measure and define the mapping $\eps^+ : (\Omega,\mbU,[0,T])\rightarrow \Omega$ via
\begin{align*}
    \Bigl(\eps^+_{(u,t)}(\omega)\Bigr)(A)=\omega\bigl(A\cap \{(u,t)\}^c\bigr) + {1}_{A}(u,t).
\end{align*}

This map essentially adds a jump at $(u,t)$. For a functional $F$ on $\Omega$, the difference operator $D_{(u,t)}$\footnote{This operator corresponds to the Malliavin derivative on Wiener space.} is then defined by
\begin{align*}
\bigl(D_{(u,t)}F\bigr)(\omega) = \bigl(F \circ \eps^+_{(u,t)}\bigr)(\omega)-F(\omega).
\end{align*}

With a slight abuse of notation we will identify $\mbE\bigl[ D_{(u,t)} F\big| \mcF_{t-}\bigr]$ with its predictable projection. This predictable projection exists, due to \cite[Lemma 3.3]{Zhang_2009}.

\begin{lemma}
Let $F$ be such that $\int_0^T \int_{\mbU} | F|+|F\circ \eps_{(u,t)}^+|\, \nu(du,dt)<\infty$. Then
\begin{align*}
    F=\mbE\bigl[F\bigr]+\int_0^T \int_{\mbU} \mbE\Bigl[ D_{(u,t)} F\Big| \mcF_{t-}\Bigr] \tilde \mu(du,dt).
\end{align*}
\end{lemma}
\begin{proof}
By Theorem 3.4 in \cite{Zhang_2009} this formula holds for bounded $F$. 
The result follows by approximating $F$ by a bounded sequence $\{F^n\}_{n \in \mathbb{N}}$ and then using dominated convergence.
\end{proof}

Applying the previous result to our volatility process 
$$
    \Vn_t=\int_{\mbR^2}\int_0^t \phi_n(t-s) \frac{u}{\sqrt{n}} N_n(ds,du,dv) 
$$    
with Poisson random measure $N_n(dt,du,dv)=N(n\cdot dt,du,dv)$ and compensated measure $\tilde N_n$ we obtain the following corollary.

\begin{corollary}\label{Cor:ClarkOconeExpVol}
Let $f:\mathbb{R}^k\rightarrow \mathbb{R}$ be of exponential form in the sense of Definition \ref{Def:Exp_Form}, for some constants $c_1,\dots,c_k$. Let $T \ge t_1 \ge \dots \ge t_k \ge 0$. Then
\begin{align*}
     f\bigl( \Vn_{t_1},\dots,\Vn_{t_k}\bigr)  = & ~ \exp\biggl( 
    \int_0^T \int_{\mathbb{R}^2} n\Bigl(e^{\sum_{j=1}^k \frac{c_j v}{\sqrt{n}}\phi_n(t_j-s)}
    -1 \Bigr) \mcP(du,dv) dt
    \\
    &\qquad \qquad + \int_{-\infty}^0 \int_{\mathbb{R}^2} n\Bigl(e^{\sum_{j=1}^k \frac{c_j v}{\sqrt{n}}\bigl(\phi_n(t_j-s)-\phi_n(-s)\bigr)}
    -1 \Bigr) \mcP(du,dv) dt
    \biggr)
    \\
    &+\int_0^T  \int_{\mbR^2} 
    \mbE\Bigl[
     f(\Vn_{t_1},\dots,\Vn_{t_k}\bigr)\Bigl(e^{\sum_{j=1}^k \frac{c_j v}{\sqrt{n}}\phi_n(t_j-s)} 
    -1
    \Bigr)
    \big|\mcF_{s-}\Bigr]
    \tilde N_n(ds,du,dv)
    \\
    &+\int_{-\infty}^0  \int_{\mbR^2} 
    \mbE\Bigl[
     f(\Vn_{t_1},\dots,\Vn_{t_k}\bigr)\Bigl(e^{\sum_{j=1}^k \frac{c_j v}{\sqrt{n}}\bigl(\phi_n(t_j-s)-\phi_n(-s)\bigr)} 
    -1
    \Bigr)
    \big|\mcF_{s-}\Bigr]
    \tilde N_n(ds,du,dv). 
\end{align*}
\end{corollary}

\subsection{Moment Formula and Error Bounds}\label{Sec:Mom_Formula}

We are going to derive moment formulas for $\Pn$, $\Pncont$ and $\Pcont$. The following lemma is key; its proof is given in the next subsection. 

\begin{lemma}
    \label{Lem:Mom_Recursion}
     Let $\mathbf{t}=(t_1, ..., t_{m}) \in \Delta_{m}^\circ$ and $t\in [0,t_{1}]$. Let $F: \mathbb{R}^{m}\times \Delta_{m}^\circ \to \mbR$ be of the form  $$F(x_1,\dots,x_m,t_1,\dots,t_m)=f_1(x_1,\dots,x_m)f_2(t_1,\dots,t_m),$$ with $f_1$ being of exponential form and let $\Vn_{\mathbf{t}}=\bigl( \Vn_{t_1},\dots,\Vn_{t_m}\bigr)$. Then
    \begin{align*}
        \mbE\Bigl[ (\Pn_t)^k F\bigl(\Vn_{\mathbf{t}},\mathbf{t}\bigr) \Bigr]
        &=
        \rho\sigma_p\sigma_v k \int_0^t \mbE\Bigl[ (\Pn_s)^{k-1} \sum_{j=1}^m \partial_{x_j} F\bigl(\Vn_{\mathbf{t}},\mathbf{t}\bigr)e^{\Vn_s} \phi_n(t_j-s) \Bigr] \,ds
        \\
        &\quad+ \sigma_p^2\frac{k(k-1)}2
        \int_0^t \mbE\Bigl[ (\Pn_s)^{k-2} F \bigl(\Vn_{\mathbf{t}},\mathbf{t}\bigr)e^{\Vn_s}\Bigr] \,ds
        \\
        &\quad +\mathcal{O}\bigl(\star(n)\bigr).
    \end{align*}
\end{lemma}

In view of the previous lemma the moment formulas can be established using methods and techniques from \cite{Friz_Salkeld_Wagenhofer_2025}; details including the proof of Theorem \ref{Thm:MomRep} are given in Appendix \ref{Sec:AppMoments}.

\begin{theorem}
    \label{Thm:MomRep}
    Let $N\in \mbN$, $\mathcal{W}$ be as in Definition \ref{Def:Wordset} and $\mathcal{L}^w$ as well as $\alpha$, $C_w$ and $\psi_{\mathbf{l}}$ from Definition \ref{Def:Afcts}, with $\psi_{\mathbf{l}}$ of exponential form as in Definition \ref{Def:Exp_Form}. Then the price process \ref{Def:Pn_Vn} satisfies
    \begin{equation}\label{Eq:MomRep}
        \begin{aligned}
        \mbE\Bigl[ (\Pn_T)^N \Bigr]
            &=
        \sum_{\substack{w\in \cW \\ \ell(w)=N}} 
        \sum_{\mathbf{l}\in \cL^w}  C_w \hspace{-5pt}\int\limits_{\mathbf{t}\in \Delta^\circ_{|w|}}
        \mbE\bigl[
        \psi_{\mathbf{l}}\big( \Vn_{t_1},\dots,\Vn_{t_{|w|}}\big) 
        \bigr] 
        \prod_{i=2}^{|w|} \phi_n(t_{\alpha_{\mathbf{l}}(i)},t_i) 
        d\mathbf{t}
       +\mathcal{O}\bigl(\star(n)\bigr).
        \end{aligned}
    \end{equation}
    The approximate continuous log-price model from Definition \ref{Def:Pncont_Vncont} satisfies
    \begin{equation}
        \label{Eq:Thm_MomRep_BMn}
        \begin{aligned}
            &\mbE\Bigl[ (\Pncont_T)^N \Bigr]=\sum_{\substack{w\in \cW \\ \ell(w)=N}} 
        \sum_{\mathbf{l}\in \cL^w} C_w\hspace{-5pt}\int\limits_{\mathbf{t}\in \Delta^\circ_{|w|}}
        \mbE\bigl[
        \psi_{\mathbf{l}}\big(  \sigma_vc_H\fBMn_{t_1},\dots, \sigma_vc_H\fBMn_{t_{|w|}}\big) 
        \bigr] 
        \prod_{i=2}^{|w|} \phi_n(t_{\alpha_{\mathbf{l}}(i)},t_i) 
        d\mathbf{t},
        \end{aligned}
    \end{equation}
    and the continuous log-price model given in Theorem \ref{Thm:WeakConv} satisfies
    \begin{equation}
        \label{Eq:Thm_MomRep_BM}
        \begin{aligned}
            &\mbE\Bigl[ (\Pcont_T)^N \Bigr]
            =\sum_{\substack{w\in \cW \\ \ell(w)=N}} 
        \sum_{\mathbf{l}\in \cL^w} C_w\hspace{-5pt}\int\limits_{\mathbf{t}\in \Delta^\circ_{|w|}}
        \mbE\bigl[
        \psi_{\mathbf{l}}\big(  \sigma_vc_H\fBM_{t_1},\dots, \sigma_vc_H\fBM_{t_{|w|}}\big) 
        \bigr] 
        \prod_{i=2}^{|w|} \phi_\infty(t_{\alpha_{\mathbf{l}}(i)},t_i) 
        d\mathbf{t}
        \end{aligned}
    \end{equation}
\end{theorem}

\subsubsection{Proof of Lemma \ref{Lem:Mom_Recursion}}\label{Sec:PV_estimates}

The proof of Lemma \ref{Lem:Mom_Recursion} - and ultimately the proof of Theorem \ref{Thm:Main_Thm} - requires a series of a-priori estimates that we provide in what follows. 
We start with an It\^{o}-type formula for moments of $\Pn$ when tested with martingales. 
We recall that $\tilde N$ denotes the compensated Poisson random measure associated to $N$.

\begin{lemma}\label{Lem:Pn_Product}
Let $\Pn_t$ be defined as in Equation \eqref{Def:Price_Process}. Let $g:[0,T]\times\mbR\times \Omega \rightarrow \mbR$ be a predictable process such that $\int_0^T \int_{\mbR} \mbE[ |g_s(v)|^p] \,\mcP(\mbR,dv)\,ds<\infty$ for all $p >1$ and consider the martingale 
\[
    Y_t=\int_{-\infty}^t \int_{\mbR^2} g_s(v)\frac1{\sqrt{n}} \tilde N(n\cdot ds,du,dv).
\]
Then it holds for all $k \ge 2$ and $t \in [0,T]$ that
\begin{align*}
    \mbE\Bigl[ \bigl(\Pn_t\bigr)^k Y_t]&=\int_0^t \sigma_p^2\frac{k(k-1)}2\mbE\Bigl[\bigl(\Pn_s\bigr)^{k-2} e^{2 \Vn_{s^-}} Y_t\Bigr] ds
    \\
    &\qquad +\int_0^t \int_{\mbR^2} k \,\mbE\Bigl[\bigl(\Pn_{s^-}\bigr)^{k-1} e^{ \Vn_{s^-}}g_{s}(v)\Bigr] u \mcP(du,dv)ds
    \\
    &\qquad+\frac{1}{\sqrt{n}}\int_0^t \int_{\mbR^2} \frac{k(k-1)}2 \mbE\Bigl[\bigl(\Pn_{s^-}\bigr)^{k-2} e^{2 \Vn_{s^-}}g_{s}(v)\Bigr] u^2 \mcP(du,dv)ds
    + \mathcal{O}(n^{-1}).
\end{align*}
Furthermore, for all $k \ge 2$ and $t \in [0,T]$, 
\begin{align}\label{Eq:P_Moments}
    \mbE\Bigl[ \bigl(\Pn_t\bigr)^k\Bigr]=\int_0^t \sigma_p^2\frac{k(k-1)}2\mbE\Bigl[\bigl(\Pn_s\bigr)^{k-2} e^{2 \Vn_{s^-}}\Bigr] ds + \mathcal{O}(n^{-1}).
\end{align}
\end{lemma}
\begin{proof}
By Remark \ref{Rem:CompRep} we can represent $\Pn$ as a martingale via %
\begin{align*}
    \Pn_t = \int_0^t \int_{\mbR^2} e^{ \Vn_{s^-}}u \frac1{\sqrt{n}}\,N(n\cdot ds,du,dv)=\int_0^t \int_{\mbR^2} e^{ \Vn_{s^-}}u \frac1{\sqrt{n}} \,\tilde N(n\cdot ds,du,dv).
\end{align*}
From It\^o's formula for jump processes along with the no skew condition of Assumption \ref{Ass:Compensator}, it follows for $k\ge 2$ and $t\in [0,T]$ that
\begin{align}
    \bigl(\Pn_t\bigr)^k &=\int_0^t \int_{\mbR^2} \Bigl(\Pn_{s^-} + \frac{u}{\sqrt{n}}e^{ \Vn_{s^-}}\Bigr)^k-\bigl(\Pn_{s^-}\bigr)^k -  \frac{u}{\sqrt{n}}e^{ \Vn_{s^-}} k \bigl(\Pn_{s^-}\bigr)^{k-1}\,\mcP(du,dv) \,d(n\cdot{}s)
    \notag
    \\
    &\qquad +\int_0^t \int_{\mbR^2} \Bigl(\Pn_{s^-} +  \frac{u}{\sqrt{n}}e^{ \Vn_{s^-}} \Bigr)^k-\bigl(\Pn_{s^-}\bigr)^k \,\tilde N(n\cdot ds,du,dv)
    \notag
    \\
    &=\int_0^t \int_{\mbR^2} \frac{k(k-1)}2 \bigl(\Pn_{s^-}\bigr)^{k-2} e^{2 \Vn_{s^-}} u^2\mathcal{P}(du,dv)\,ds \tag{$P_{I}$}\label{Eq:P_Moments1}
    \\
    & \qquad+\frac1n\sum_{j=4}^k \int_0^t \int_{\mbR^2} \binom{k}{j} {n}^{\frac{4-j}2}\bigl(\Pn_{s^-}\bigr)^{k-j} e^{j \Vn_{s^-}} u^j\mathcal{P}(du,dv)\,ds \tag{$P_{II}$}\label{Eq:P_Moments2}
    \\
    &\qquad +\int_0^t \int_{\mbR^2} k \bigl(\Pn_{s^-}\bigr)^{k-1} e^{ \Vn_{s^-}} u \frac1{\sqrt{n}} \tilde N(n\cdot ds,du,dv) \tag{$P_{III}$}\label{Eq:P_Moments3}
    \\
    & \qquad +\sum_{j=2}^k \int_0^t \int_{\mbR^2} \binom{k}{j} {n}^{\frac{-j}2}\bigl(\Pn_{s^-}\bigr)^{k-j} e^{j \Vn_{s^-}} u^j
    \tilde N(n\cdot ds,du,dv)\tag{$P_{IV}$}\label{Eq:P_Moments4}.
\end{align}
We start proving the second claim. Taking expectations the expressions \eqref{Eq:P_Moments3} and \eqref{Eq:P_Moments4} vanish. It also follows from moment estimates for $\Pn$ that $\mbE\bigl[ \eqref{Eq:P_Moments2}\bigr]\lesssim n^{-1}$. Taking expectations in \eqref{Eq:P_Moments1} and using Assumption \ref{Ass:Compensator} the claim follows. 

To show the first part of this lemma we split the calculations into four parts. For the first part, we use Fubini's theorem to get
\begin{align*}
    \mbE\Bigl[ \eqref{Eq:P_Moments1}\cdot Y_t\Bigr]=\int_0^t \int_{\mbR^2} \frac{k(k-1)}2\mbE\Bigl[\bigl(\Pn_{s-}\bigr)^{k-2} e^{2 \Vn_{s^-}} Y_t\Bigr] u^2 \mcP(du,dv) ds
\end{align*}

For the second part, we also apply Fubini, but this time all expressions are of order $n^{-1}$. Indeed, the integrability conditions on $Y$ together with the moment estimates for $\Pn$ and the exponential moment estimates of $e^{\Vn}$, and Lemma \ref{Lem:VMoments} imply that
\begin{align*}
    \biggl|  \mbE\Bigl[ \eqref{Eq:P_Moments2}\cdot Y_t\Bigr] \biggr| \lesssim 
    \biggl|
    \frac1n\sum_{j=4}^k \int_0^t \int_{\mbR^2} \binom{k}{j} {n}^{\frac{4-j}2}\mbE\Bigl[\bigl(\Pn_{s^-}\bigr)^{k-j} e^{j \Vn_{s^-}}Y_t\Bigr] u^j\mathcal{P}(du,dv)\,ds \biggr|\lesssim n^{-1}.
\end{align*}

For the third term, we apply the product formula for compensated Poisson random measures that see that
\begin{align*}
      \mbE\Bigl[ \eqref{Eq:P_Moments3}\cdot Y_t\Bigr]=\int_0^t \int_{\mbR^2} k \mbE\Bigl[\frac{1}{\sqrt{n}}\bigl(\Pn_{s^-}\bigr)^{k-1} e^{ \Vn_{s^-}}\frac{1}{\sqrt{n}}g_{s^-}(v)\Bigr] u n\,\mcP(du,dv)ds.
\end{align*}
Similarly, %
\begin{align*}
    \mbE\Bigl[ \eqref{Eq:P_Moments4}\cdot Y_t\Bigr]&=
    \sum_{j=2}^k \int_0^t \int_{\mbR^2} \binom{k}{j} {n}^{\frac{-j}2}\mbE\Bigl[\bigl(\Pn_{s^-}\bigr)^{k-j} e^{j \Vn_{s^-}}\frac{1}{\sqrt{n}}g_s(v)\Bigr] u^j
    n \,\mcP(du,dv)ds
\end{align*}
From the moment estimates for $Y, \Pn$ and $e^{\Vn}$ it again follows that
\begin{align*}
    \biggl| \frac1n\sum_{j=3}^k \int_0^t \int_{\mbR^2} \binom{k}{j} {n}^{\frac{3-j}2}\mbE\Bigl[\bigl(\Pn_{s^-}\bigr)^{k-j} e^{j \Vn_{s^-}}g_s(v)\Bigr] u^j
    \,\mcP(du,dv)ds\biggr| \lesssim \frac1n.
\end{align*}
\end{proof}

The next lemma tells us how to deal with expectations of exponentials of $\Vn$ and products thereof.

\begin{lemma}\label{Lem:ClarkOconeEstim}
 Let $c_1,\dots,c_k \in \mathbb{R}$, $0\le t_k \le \dots \le t_1\le T$. Let $\mcP$ satisfy Assumption \ref{Ass:Compensator}. Then,
\begin{align*} 
    &\Biggl|   
    \int_0^T \int_{\mathbb{R}^2} n\Bigl(e^{\sum_{j=1}^k \frac{c_j v}{\sqrt{n}}\phi_n(t_j-s)}
    -1 \Bigr) \mcP(du,dv) ds
   - \frac{\sigma_v^2}{2}\int_0^T \biggl(\sum_{j=1}^k c_j\phi_n(t_j-s)\biggr)^2\,ds
    \Biggr|
    \lesssim \star(n),
\end{align*}
and
\begin{align*}
    &\Biggl|   
    \int_{-\infty}^0 \int_{\mathbb{R}^2} n\Bigl(e^{\sum_{j=1}^k \frac{c_j v}{\sqrt{n}}\bigl(\phi_n(t_j-s)-\phi_n(-s)\bigr)}
    -1 \Bigr) \mcP(du,dv) ds \\
   & \quad \quad - \frac{\sigma_v^2}{2}\int_{-\infty}^0 \biggl(\sum_{j=1}^k c_j\bigl(\phi_n(t_j-s)-\phi_n(-s)\bigr)\Bigr)^2\,ds
    \Biggr|
    \lesssim \star(n).
\end{align*}
In particular, for $f:\mathbb{R}^k\rightarrow \mathbb{R}$  %
of exponential form
\begin{equation*}
    \Biggl| 
         \mbE\Bigl[f(\Vn_{t_1},\dots,\Vn_{t_k})\Bigr] - \mbE\Bigl[f\bigl(\sigma_vc_H\fBMn_{t_1},\dots,\sigma_vc_H\fBMn_{t_k}\bigr)\Bigr]
    \Biggr|
    \lesssim \star(n).
\end{equation*}
\end{lemma}
\begin{proof}
Since $\phi_n(t)=0$ for $t<0$ and $\limsup\limits_{n \in \mbN} \|n^{-1/2}\phi_n(.)\|_\infty<\infty$, 
\begin{align*}
   & \Bigl| e^{\sum_{j=1}^k \frac{c_j v}{\sqrt{n}}\phi_n(t_j-s)}
    -1 -\frac{v}{\sqrt{n}}\sum_{j=1}^k c_j\phi_n(t_j-s)-\frac12 \frac{v^2}{n} \Bigl(\sum_{j=1}^k c_j\phi_n(t_j-s)\Bigr)^2
    \\& \quad 
    -\frac16 \frac{v^3}{n^{3/2}} \Bigl(\sum_{j=1}^k c_j\phi_n(t_j-s)\Bigr)^3\Bigr|
    \\ \lesssim & ~ e^{cv}v^4n^{-2}\Bigl(\sum_{j=1}^k {c_j}\phi_n(t_j-s)\Bigr)^4
\end{align*}
for some constant $c$ depending on $c_1,\dots,c_k$.
Therefore,
\begin{align*}
    &\int_0^T \int_{\mathbb{R}^2} n\Bigl(e^{\sum_{j=1}^k \frac{c_j v}{\sqrt{n}}\phi_n(t_j-s)}
    -1 \Bigr) \mcP(du,dv) ds \\ = &
    \int_0^T \int_{\mathbb{R}^2} \sqrt{n} \sum_{j=1}^k {c_j }\phi_n(t_j-s) v \mcP(du,dv) ds
    \\
    &+\int_0^T \int_{\mathbb{R}^2} \frac{v^2}{2} \Bigl(\sum_{j=1}^k c_j\phi_n(t_j-s)\Bigr)^2 \mcP(du,dv) ds+\int_0^T \int_{\mathbb{R}^2}\frac16 \frac{v^3}{\sqrt{n}} \Bigl(\sum_{j=1}^k c_j\phi_n(t_j-s)\Bigr)^3 \mcP(du,dv) ds
    \\
    &+\mathcal{O}\Bigl(\int_0^T \int_{\mathbb{R}^2}n^{-1} e^{cv}v^4\Bigl(\sum_{j=1}^k {c_j}\phi_n(t_j-s)\Bigr)^4  ds\mcP(du,dv)\Bigr).
\end{align*}
By Assumption \ref{Ass:Compensator} the first and the third integrals on the r.h.s.\ vanish as $n \to \infty$. Using Young's inequality, the sub-Gaussian tails of $\mcP$ as well as the convention $\phi_n(x)=0$ for $x <0$,
\begin{align*}
    \int_0^T \int_{\mathbb{R}^2}n^{-1} e^{cv}v^4\Bigl(\sum_{j=1}^k {c_j}\phi_n(t_j-s)\Bigr)^4  ds\mcP(du,dv) \lesssim \frac1n \int_0^T \phi_n^4(t-s)\,ds = \star(n).
\end{align*}
Hence the first claim follows.  Similar considerations show the second statement.
Recall that 
\begin{align*}
     \mbE\Bigl[f(\Vn_{t_1},\dots,\Vn_{t_k})\bigr]= \exp \biggl(
         \int_0^T \int_{\mathbb{R}^2} n\Bigl(e^{\sum_{j=1}^k \frac{c_j v}{\sqrt{n}}\phi_n(t_j-s)}
    -1 \Bigr) \mcP(du,dv) ds
     \biggr),
\end{align*}
therefore the last statement follows from local Lipschitz continuity of the exponential function, uniform boundedness of the exponent in $n$ and $t \in [0,T]$ and the fact that 
\begin{align*}
    \mbE\Bigl[f\bigl(\sigma_vc_H\fBMn_{t_1},\dots,\sigma_vc_H\fBMn_{t_k}\bigr)\Bigr]&=\exp\biggl( \int_{-\infty}^0 \sigma_v^2\Bigl( \sum_{k=1}^m c_k \bigl( \phi_n(t_k-s)-\phi_n(-s)\bigr)\Bigr)^2\,ds\biggl)
    \\
    &\qquad \cdot \exp\biggl(
    \sum_{j=0}^{m-1} \sigma_v^2\int_{t_{j}}^{t_{j+1}}  \Bigl(\sum_{k=1}^{m-j} c_k \phi_n(t_k-s) \Bigr)^2 v^2 ds
    \biggr).
\end{align*}
\end{proof}

The next lemma will be used to approximate the cross-correlation terms that appear in the proof of Lemma \ref{Lem:Mom_Recursion}.  Recall that $\tilde N$ denotes the compensate version of the Poisson random measure $N$.

\begin{lemma}\label{Lem:fV_Product}
Let $f$ be of exponential form in the sense of Definition \ref{Def:Exp_Form}, and $g:[0,T]\times \Omega \rightarrow\mathbb{R}$ be predictable such that $g\in {L^p([0,1]\times \Omega)}$ for all $p>1$. For $j\in \{1, 2\}$ let 
\[
    X^{(j)}=\frac1{\sqrt{n}^{j}}\int_0^T \int_{\mbU} u^j g_s\,\tilde N_n(ds,du,dv). 
\]    
Then, for $j=1$,
\begin{align*}
    \Bigl| \mbE\bigl[ f(\Vn_{t_1},\dots,\Vn_{t_k})X^{(1)}]-\rho\sigma_p\sigma_v\int_0^T   \mbE\Bigl[g_sf(\Vn_{t_1},\dots,\Vn_{t_k}) \sum_{j=1}^k c_j\phi_n(t_j-s)\Bigr]\,ds
    \Bigr|\lesssim \star(n)
\end{align*}
and for $j=2$,
\begin{align*}
    \Bigl| \mbE\bigl[ f(\Vn_{t_1},\dots,\Vn_{t_k})X^{(2)}]
    \Bigr|\lesssim \star(n)
\end{align*}
\end{lemma}
\begin{proof}
From the product formula and Corollary \ref{Cor:ClarkOconeExpVol} it follows that
\begin{align*}
    &\mbE\bigl[ f(\Vn_{t_1},\dots,\Vn_{t_k})X^{(j)}]
    \\
    &=\int_0^T \int_{\mbR^2}  \mbE\Bigl[u^j \sqrt{n}^{2-j} g_s
    \mbE\Bigl[
     f(\Vn_{t_1},\dots,\Vn_{t_k}\bigr)\Bigl(e^{\sum_{j=1}^k \frac{c_j v}{\sqrt{n}}\phi_n(t_j-s)} 
    -1
    \Bigr)
    \big|\mcF_{s-}\Bigr]\Bigr]\mcP{(du,dv)}dt
    \\
    &=\int_0^T \int_{\mbR^2}  \mbE\Bigl[u^j \sqrt{n}^{1-j} g_s
     f(\Vn_{t_1},\dots,\Vn_{t_k}\bigr)\sqrt{n}\Bigl(e^{\sum_{j=1}^k \frac{c_jv}{\sqrt{n}}\phi_n(t_j-s)} 
    -1
    \Bigr)
    \Bigr]\,\mcP{(du,dv)}dt.
\end{align*}
Here, the second equality follows from the definition of conditional expectation for predictable processes $g$.
Taylor's formula implies that for some $C>0$
\begin{align*}
    \sqrt{n}\Bigl(e^{\sum_{j=1}^k \frac{c_jv}{\sqrt{n}}\phi_n(t_j-s)} 
    -1
    \Bigr)=&~v\sum_{j=1}^k c_j\phi_n(t_j-s)+\frac{v^2}{2\sqrt{n}}\Bigl(\sum_{j=1}^k c_j\phi_n(t_j-s)\Bigr)^2
    \\
    &+\mathcal{O}\Bigl( n^{-1} e^{Ck|v|} \Bigl|\sum_{j=1}^k c_j\phi_n(t_j-s)\Bigr|^3\Bigr).
\end{align*}
Using this estimate, we observe that
\begin{align*}
    & \mbE\bigl[ f(\Vn_{t_1},\dots,\Vn_{t_k})X^{(j)}] \\
    =& ~ \int_0^T \int_{\mbR^2}  \mbE\Bigl[g_sf(\Vn_{t_1},\dots,\Vn_{t_k}) \sum_{j=1}^k c_j \phi_n(t_j-s)\Bigr]\,u^j \sqrt{n}^{1-j}v\mcP{(du,dv)}dt
    \\
    & \quad +\int_0^T \int_{\mbR^2}  \frac1{\sqrt{n}} \mbE\Bigl[g_sf(\Vn_{t_1},\dots,\Vn_{t_k}) \Bigl(\sum_{j=1}^k c_j \phi_n(t_j-s)\Bigr)^2\Bigr]\, v^2\sqrt{n}^{1-j}u^j\mcP{(du,dv)}dt
    \\
    & \quad + \mathcal{O}\Bigl(n^{-1}
    \int_0^T \int_{\mbR^2}  \mbE\Bigl[g_sf(\Vn_{t_1},\dots,\Vn_{t_k}) \Bigl|\sum_{j=1}^k c_j\phi_n(t_j-s)\Bigr|^3\Bigr]e^{k|u|}u^j \sqrt{n}^{1-j}v^3\mcP{(du,dv)}dt
    \Bigr).
\end{align*}

By the uniform boundedness of $\mbE\bigl[f(\Vn_{t_1},\dots,\Vn_{t_k})X^{(j)}]$ and by the sub-Gaussian tail
condition in Assumption \ref{Ass:Kernel}, the last integral is at most with order $\star(n)$.

If $j=1$, then the second term vanishes by the no-skew condition in Assumption \ref{Ass:Compensator} and the statement follows from the fact that $\rho\sigma_u\sigma_v=\int_{\mbR^2}uv\mcP(du,dv)$. 

If $j=2$, then the first term vanishes again by the no-skew condition and the second term remains with a contribution of order $n^{-1}\lesssim \star(n)$.
\end{proof}

\medskip

\begin{proof}[Proof of Lemma \ref{Lem:Mom_Recursion}]
    Using the decomposition of $F$ from Definition \ref{Def:IJ_Operators}, we apply the Clark-Ocone formula given in Corollary \ref{Cor:ClarkOconeExpVol} to see that
    \begin{align*}
        F\bigl(\Vn_{\mathbf{t}},\mathbf{t}\bigr)&=\mbE\Bigl[ F\bigl(\Vn_{\mathbf{t}},\mathbf{t}\bigr) \Bigr] \\
        & \qquad +\int_0^T  \int_{\mbR^2} 
    \underbrace{\mbE\Bigl[
     f_1(\Vn_{\mathbf{t}}\bigr) f_2(\mathbf{t})\Bigl(e^{\sum_{j=1}^k \frac{c_ju}{\sqrt{n}}\phi_n(t_j-s)} 
    -1
    \Bigr)
    \big|\mcF_{s-}\Bigr]}_{\eqqcolon h_{s-}(v)}
    \tilde N_n(ds,du,dv)
    \\ & \qquad +\int_{-\infty}^0 
    \int_{\mbR^2}
    \mbE\Bigl[
     f_1(\Vn_{\mathbf{t}}\bigr) f_2(\mathbf{t})\Bigl(e^{\sum_{j=1}^k \frac{c_j v}{\sqrt{n}}\bigl(\phi_n(t_j-s)-\phi_n(-s)\bigr)} 
    -1
    \Bigr) 
    \big|\mcF_{s-}\Bigr]
     \tilde N_n(ds,du,dv)
     \\
     &\eqqcolon \mbE\Bigl[ F\bigl(\Vn_{\mathbf{t}},\mathbf{t}\bigr) \Bigr] + Y_T. 
    \end{align*}
    According to Equation \eqref{Eq:P_Moments} in Lemma \ref{Lem:Pn_Product}, 
    \begin{align}
        \mbE\Bigl[ (\Pn_t)^k\mbE\bigl[ F\bigl(\Vn_{\mathbf{t}},\mathbf{t}\bigr) \bigr]  \Bigr]=\int_0^t \sigma_p^2 \frac{k(k-1)}2\mbE\Bigl[\bigl(\Pn_s\bigr)^{k-2} e^{2 \Vn_{s^-}}\mbE\bigl[ F\bigl(\Vn_{\mathbf{t}},\mathbf{t}\bigr) \bigr]\Bigr] ds + \mathcal{O}(n^{-1})\label{Eq:PnEY}
    \end{align}
    and by the same lemma, 
    \begin{align}
        \mbE\Bigl[ (\Pn_t)^kY_T \Bigr]&=\int_0^t \sigma_p^2\frac{k(k-1)}2\mbE\Bigl[\bigl(\Pn_s\bigr)^{k-2} e^{2 \Vn_{s^-}} Y_T\Bigr] ds \label{Eq:PnY1}
    \\
    &\qquad +\int_0^t \int_{\mbR^2} k \mbE\Bigl[\bigl(\Pn_{s^-}\bigr)^{k-1} e^{ \Vn_{s^-}}\sqrt{n}h_{s^-}(v)\Bigr] u \mcP(du,dv)ds \label{Eq:PnY2}
    \\
    &\qquad+\frac{1}{\sqrt{n}}\int_0^t \int_{\mbR^2} \frac{k(k-1)}2 \mbE\Bigl[\bigl(\Pn_{s^-}\bigr)^{k-2} e^{2 \Vn_{s^-}}\sqrt{n}h_{s^-}(v)\Bigr] u^2 \mcP(du,dv)ds
    \label{Eq:PnY3}
    \\
    &\qquad+ \mathcal{O}(n^{-1}). \nonumber
    \end{align}
    
    The product rule, as well as Lemma \ref{Lem:fV_Product} applied to $$X_T^{(1)}=\int_0^T \int_{\mbR^2} k (\Pn_{s^-}\bigr)^{k-1} e^{ \Vn_{s^-}}u\frac1{\sqrt{n}} \tilde N_n(ds,du,dv)$$ imply that
    \begin{align*}
    \eqref{Eq:PnY2}&=
        \mbE\Bigl[ X_T F(\Vn_{\mathbf{t}},\mathbf{t})\Bigr]
        \\
        &=\rho\sigma_p\sigma_v \int_0^t  k\mbE\Bigl[\bigl(\Pn_{s^-}\bigr)^{k-1} e^{ \Vn_{s^-}}F(\Vn_{\mathbf{t}},\mathbf{t}) \sum_{j=1}^k c_j\phi_n(t_j-s)\Bigr] ds+\mathcal{O}\bigl(\star(n)\bigr).
    \end{align*}
    Similarly, the product formula and Lemma \ref{Lem:fV_Product} applied to $$X_T^{(2)}=\int_0^T \int_{\mbR^2} \frac{k(k-1)}2 (\Pn_{s^-}\bigr)^{k-2} e^{2 \Vn_{s^-}}u^2\frac1{{n}} \tilde N_n(ds,du,dv)$$ imply that $$\bigl|\eqref{Eq:PnY3}\bigr| =\Bigl|\mbE\Bigl[ \hat X_T^{(2)} F(\Vn_{\mathbf{t}},\mathbf{t})\Bigr]\Bigr| \lesssim \star(n).$$
    Adding the r.h.s.\ of Equation \eqref{Eq:PnEY} and \eqref{Eq:PnY1} together with the estimates for \eqref{Eq:PnY2} and \eqref{Eq:PnY3} 
    yields the desired result.
    
\end{proof}

\subsubsection{Error Bounds}
\label{Sec:MomBounds}

We now establish the Estimate \eqref{Eq:moment_est}. Theorem \ref{Thm:Logprice_Diff} compares the $N$-th moment of $\Pn$ with the $N$-th moment of $\Pncont$ and Theorem \ref{Thm:Cont_Diff} compares the $N$-th moments of $\Pncont$ and $\Pcont$. %

\begin{theorem}\label{Thm:Logprice_Diff}
Let $\Pn$ be as in Definition \ref{Def:Pn_Vn} and $\Pncont$ be as in Definition \ref{Def:Pncont_Vncont}. Then,
\begin{align*}
    \biggl|\mbE\Bigl[ \bigl(\Pn_t\bigr)^N\Bigr]-\mbE\Bigl[ \bigl(\Pncont_t\bigr)^N\Bigr]\biggr|
    \lesssim \star(n).
\end{align*}
\end{theorem}
\begin{proof}
From the moment formulas \eqref{Eq:MomRep} and \eqref{Eq:Thm_MomRep_BMn} in Theorem \ref{Thm:MomRep} it follows that
\begin{align*}
    \mbE\Bigl[ \bigl(\Pncont_t\bigr)^N\Bigr]-\mbE\Bigl[ \bigl(\Pn_t\bigr)^N\Bigr]
    = & ~ 
    \sum_{\substack{w\in \cW \\ \ell(w)=N}} 
        \sum_{\mathbf{l}\in \cL^w} C_w\hspace{-5pt}\int\limits_{\mathbf{t}\in \Delta^\circ_{|w|}}
        \biggl(\mbE\bigl[
        \psi_{\mathbf{l}}\big( \sigma_vc_H\fBMn_{t_1},\dots,\sigma_vc_H\fBMn_{t_{|w|}}\big)
        \bigr] 
        \\
        & \qquad -
        \mbE\bigl[
        \psi_{\mathbf{l}}\big( \Vn_{t_1},\dots,\Vn_{t_{|w|}}\big) 
        \bigr] \biggr)
        \prod_{i=2}^{|w|} \phi_n(t_{\alpha_{\mathbf{l}}(i)},t_i) 
        d\mathbf{t} + \mathcal{O}(\star(n)).
\end{align*}
Since $\psi_{\mathbf{l}}$ is of exponential form, Lemma \ref{Lem:ClarkOconeEstim} implies that
\begin{align*}
    \biggl|\mbE\bigl[
        \psi_{\mathbf{l}}\big( \sigma_vc_H\fBMn_{t_1},\dots,\sigma_vc_H\fBMn_{t_{|w|}}\big) 
        \bigr] -
        \mbE\bigl[
        \psi_{\mathbf{l}}\big( \Vn_{t_1},\dots,\Vn_{t_{|w|}}\big) 
        \bigr]\biggr| \lesssim  \star(n).
\end{align*}

The monotonicity condition in Assumption \ref{Ass:Kernel} yields $\phi_n(t_{\alpha_{\mathbf{l}}(i)},t_i)\le \phi_n(t_{i+1},t_i)$.  
 Condition \ref{Cond:Bound1} in Assumption \ref{Ass:Kernel} yields a constant $C>0$ such that for all $n\in \mathbb{N}$ and $i \in \{1\dots,|w|\}$,
 $$\int_0^T |\phi_n(t_{\alpha_{\mathbf{l}}(i)},t_i)|\,dt \le C.$$
Integrating out all terms iteratively and taking all finite sums, the claim follows.
\end{proof}

\begin{lemma}\label{Lem:DiffBMs}
Let $m \in \mathbb{N}$, $f$ of exponential in the sense of Definition \ref{Def:Exp_Form} and  $\phi_\infty(t)=t^{H-1/2}, t\ge 0$.
Then, %
\begin{align*}
    & \int_{\Delta_m^\circ} \Bigl| \mbE\Bigl[ f\bigl(W_{t_1}^{H},\dots,W_{t_m}^{H}\bigr)\Bigr] - \mbE\Bigl[ f\bigl(W_{t_1}^{H,n},\dots,W_{t_m}^{H,n}\bigr)\Bigr]
    \Bigr| \prod_{j=1}^{m-1}\phi_\infty(t_{j}-t_{j+1})\,dt_m\dots dt_1
    \\
    \lesssim ~& \square(n)+\diamond(n).
\end{align*}
\end{lemma}
\begin{proof}
We denote the l.h.s.\ by $\mathcal{E}_n$, use the notation $t_0=T$ and the convention that products over the empty set are equal to one. 

Let $\widetilde C_n(s,t):=\mathrm{Cov}\bigl(W_s^{n,H},W_t^{n,H}\bigr)=c_H^{-2} C_n(s,t)$ and similarly $\widetilde C(s,t):=\mathrm{Cov}\bigl(W_s^{H},W_t^{H}\bigr)=c_H^{-2}C_\infty(t,s)$. 
Let $\Sigma:[0,1]\rightarrow\mathbb{R}^{m\times m}$ be given by  
\[
    \Sigma(\lambda)_{i,j}=\widetilde C(t_i,t_j)+\lambda \Bigl(\widetilde C_n(t_i,t_j)-\widetilde C(t_i,t_j)\Bigr).
\]    
Then, by Lemma \ref{Lem:DerivSigma2}
\begin{align*}
    \Bigl| \mbE\Bigl[ f\bigl(W_{t_1}^{H},\dots,W_{t_m}^{H}\bigr)\Bigr] &- \mbE\Bigl[ f\bigl(W_{t_1}^{H,n},\dots,W_{t_m}^{H,n}\bigr)\Bigr]
    \Bigr| 
    \\
    &= \int_0^1 \sum_{k,l=1}^m \frac 12 \Bigl( \widetilde C_n(t_k,t_l)-\widetilde C(t_k,t_l)\Bigr)\mbE \bigl[\partial_k \partial_l f\bigl(\Sigma(\lambda)W_1\bigr)\bigr] \,d\lambda
     \\
    &\lesssim  \sum_{k,l=1}^m\Bigl| \widetilde C_n(t_k,t_l)-\widetilde C(t_k,t_l)\Bigr|.
\end{align*}
As a result, 
\begin{align*}
    \mathcal{E}_n \lesssim  \sum_{k,l=1}^m
    \int_{\Delta_m^\circ} \Bigl| \widetilde C_n(t_k,t_l)-\widetilde C(t_k,t_l)\Bigr| \prod_{j=1}^{m-1}\phi_\infty(t_{j}-t_{j+1})\,dt_m\dots dt_1\eqqcolon \sum_{k,l=1}^m\mathcal{E}_n^{k,l}.
\end{align*}
We now distinguish the following two cases.

\paragraph{$|k-l|\le 1$.} By symmetry we can assume that $l\le k$ such that $t_l\ge t_k$.
Integrating over $t_m\dots t_{k+1}$ we arrive at the estimate
 \begin{align*}
    \mathcal{E}_n^{k,l} \lesssim \int_{0}^T \dots \int_0^{t_{k-2}} \prod_{j=1}^{l-1} \phi_\infty(t_{j}-t_{j+1}) \int_0^{t_{l}} \phi_\infty(t_{l}-t_k) \Bigl| \widetilde C_n(t_k,t_l)-\widetilde C(t_k,t_l)\Bigr| dt_k\dots dt_1.
\end{align*}
Recalling Equation \eqref{Eq:Cn},
the inner integral is bounded by $\diamond(n)$ if $k \not = l$ and $\square(n)$ otherwise. All the remaining terms are integrable.

\paragraph{$|k-l|>1$.} By symmetry we again assume that $l<k$ such that $t_l>t_k$. As above we integrate over $t_m \dots t_{k+1}$. 
In the following calculation, we split up the iterated integrals and then use Fubini's theorem and the fact that $\sup\limits_{r,t\in [0,T]} \int_r^t \phi_\infty(t-s)\phi_\infty(s-r)\,ds<\infty$:
\begin{alignat*}{2}
     {\int_{\Delta_k^\circ}\Bigl| \widetilde C_n(t_k,t_l)-}&\mathrlap{\widetilde C(t_k,t_l)\Bigr| \prod_{j=1}^{k-1}\phi_\infty(t_{j}-t_{j+1})\,dt_k\dots dt_1}
     \\
     &=\int_{0}^T \dots &&\int_0^{t_{k-3}}   \prod_{j=1}^{k-3}\phi_\infty(t_{j}-t_{j+1})\int_0^{t_{k-2}} \int_0^{t_{k-1}} \phi_\infty(t_{k-2}-t_{k-1})
     \\
     &&&\cdot\phi_\infty(t_{k-1}-t_{k})\Bigl| \widetilde C_n(t_k,t_l)-\widetilde C(t_k,t_l)\Bigr|\,dt_k\dots dt_1
      \\
     &= \mathrlap{\int_{0}^T \dots \int_0^{t_{k-3}}  \prod_{j=1}^{k-3}\phi_\infty(t_{j}-t_{j+1})\int_0^{t_{k-2}} \int_{t_{k-1}}^{t_{k-2}} \phi_\infty(t_{k-2}-t_{k-1}) }  
     \\
     &&&
     \cdot\phi_\infty(t_{k-1}-t_{k})dt_{k-1}\Bigl| \widetilde C_n(t_k,t_l)-\widetilde C(t_k,t_l)\Bigr|\,dt_kdt_{k-2}\dots dt_1
     \\
     &\mathrlap{\lesssim \int_{0}^T \dots \int_0^{t_{k-3}}  \prod_{j=1}^{k-3}\phi_\infty(t_{j}-t_{j+1})\int_0^{t_{k-2}}\Bigl| \widetilde C_n(t_k,t_l)-\widetilde C(t_k,t_l)\Bigr|\,dt_kdt_{k-2}\dots dt_1.}
\end{alignat*}
The claim then follows from the the fact that 
\begin{align*}
    \int_0^{t_{k-2}}\Bigl| \widetilde C_n(t_k,t_l)-\widetilde C(t_k,t_l)\Bigr|\,dt_k\lesssim \int_0^{t_{l}} \phi_\infty(t_l-t_k)\Bigl| \widetilde C_n(t_k,t_l)-\widetilde C(t_k,t_l)\Bigr|\,dt_k \lesssim \diamond(n),
\end{align*}
and again from integrability of $\phi_\infty$.
\end{proof}

\begin{theorem}\label{Thm:Cont_Diff}
Let $N \in \mathbb{N}$ and $\Pncont$ and $\Pcont$ be as in Definition \ref{Def:Pncont_Vncont} and Theorem \ref{Thm:WeakConv}. %
Then, with the notation from Equation \eqref{Eq:Rate_Notation}
\begin{align*}
    \biggl|\mbE\Bigl[ \bigl(\Pncont\bigr)^N-\mbE\Bigl[& \bigl(\Pcont\bigr)^N\Bigr]\biggr|
    \lesssim \diamond(n)+\square(n)+\triangle (n).
\end{align*}
\end{theorem}
\begin{proof}
We use the representation of  \eqref{Eq:Thm_MomRep_BMn} and \eqref{Eq:Thm_MomRep_BM} in 
Theorem \ref{Thm:MomRep} to see that
 \begin{align*}
    & \mbE\Bigl[ \bigl(\Pncont_T\bigr)^N-\mbE\Bigl[ \bigl(\Pcont_T\bigr)^N\Bigr] \\
   & ~ =
    \sum_{\substack{w\in \cW \\ \ell(w)=N}} 
        \sum_{\mathbf{l}\in \cL^w} C_w\hspace{-5pt}\int\limits_{\mathbf{t}\in \Delta^\circ_{|w|}}
       \mbE\bigl[
        \psi_{\mathbf{l}}\big( \sigma_vc_H\fBMn_{t_1},\dots,\sigma_vc_H\fBMn_{t_{|w|}}\big) 
        \bigr]
        \prod_{i=2}^{|w|} \phi_n(t_{\alpha_{\mathbf{l}}(i)},t_i) 
        \\
        &\qquad -
         \mbE\bigl[
        \psi_{\mathbf{l}}\big( \sigma_vc_H\fBM_{t_1},\dots,\sigma_vc_H\fBM_{t_{|w|}}\big) 
        \bigr] 
        \prod_{i=2}^{|w|} \phi_\infty(t_{\alpha_{\mathbf{l}}(i)},t_i) 
        d\mathbf{t}.
\end{align*}
For fixed $w\in \mathcal{W}$ and $\mathbf{l}\in \cL^{w}$ we denote 
\begin{align*}
    g_{n,\mathbf{l}}^w(t_1,\dots,t_m)=\phi_n(t_{\alpha_{\mathbf{l}}(2)}-t_2) \cdots \phi_n(t_{\alpha_{\mathbf{l}}(m)}-t_m)
\end{align*}
and use triangle inequality to decompose the integrand as
\begin{align*}
    &C_w\sum_{\mathbf{l} \in \cL^w} 
    \Bigl| \mbE\Bigl[ \psi_{\mathbf{l}}\bigl( \sigma_vc_H\fBM_{t_1},\dots,\sigma_vc_H\fBM_{t_{|w|}}\bigr) \Bigr]
     -
     \mbE\Bigl[ \psi_{\mathbf{l}}\bigl( \sigma_vc_H\fBMn_{t_1},\dots,\sigma_vc_H\fBMn_{t_{|w|}}\bigr) \Bigr]\Bigr|
     g_{n,\mathbf{l}}^w(t_1,\dots,t_m)
    \\
    &\qquad+ C_w \sum_{\mathbf{l} \in \cL^w} \mbE\Bigl[ \psi_{\mathbf{l}}\bigl( \sigma_vc_H\fBM_{t_1},\dots,\sigma_vc_H\fBM_{t_{|w|}}\bigr) \Bigr]\Bigl|
     g_{n,\mathbf{l}}^w(t_1,\dots,t_m)-g_{0,\mathbf{l}}^w(t_1,\dots,t_m)\Bigr|
     \\
     &\eqqcolon \mathcal{E}_n^1+\mathcal{E}_n^2.
\end{align*}

By the monotonicity of $\phi_n$ it follows that 
$
    g_{n,\mathbf{l}}^w(t_1,\dots,t_m)\le \prod\limits_{k=1}^{m-1} \phi_n(t_{k+1}-t_k)\le \prod\limits_{k=1}^{m-1} \phi_\infty(t_{k+1}-t_k)
$,
and therefore $\mathcal{E}_n^1$ can be bounded from above by
\begin{align*}
    C_W \sum_{\mathbf{l} \in \cL^w} 
    \Bigl| \mbE\Bigl[ \psi_{\mathbf{l}}\bigl( \sigma_vc_H\fBM_{t_1},\dots,\sigma_vc_H\fBM_{t_{|w|}}\bigr) \Bigr]
     -
     \mbE\Bigl[ \psi_{\mathbf{l}}\bigl( \sigma_vc_H\fBMn_{t_1},\dots,\sigma_vc_H\fBMn_{t_{|w|}}\bigr) \Bigr]\Bigr|
     \prod_{k=1}^{m-1} \phi_\infty(t_{k+1}-t_k).
\end{align*}

Integrating this bound over $\Delta_m^\circ$ and using Lemma \ref{Lem:DiffBMs}, which we can be used because $\psi_{\mathbf{l}}$ is of exponential form,  it follows that
\begin{align*}
    \sum_{\substack{w\in \mathcal{W}\\\ell(w)=N}} \int_0^T \int_0^{t_1}\dots\int_0^{t_{|w|-1}} \mathcal{E}_n^1 \,dt_{|w|}\dots dt_1\lesssim \square(n)+\diamond(n).
\end{align*}

Since $\mathcal{W}$ and $\mathcal{L}^w$ are finite sets, by exponential integrability of Gaussian random variables,
\begin{align*}
    \sup_{w\in \mathcal{W}} \sup_{\mathbf{l}\in \mathcal{L}^w}\sup_{(t_1,\dots,t_m)\in[0,T]^m}  \mbE\Bigl[ \psi_{\mathbf{l}}\bigl( \sigma_vc_H\fBM_{t_1},\dots,\sigma_vc_H\fBM_{t_{|w|}}\bigr) \Bigr]<\infty.
\end{align*}
Hence to study $\mathcal{E}_n^2$, it suffices to bound the integral of $g_{n,\mathbf{l}}^w-g_{0,\mathbf{l}}^w$  over the simplex. After multiplications of the triangle inequality and using Assumption \ref{Ass:Kernel}, we observe that
\begin{align*}
    \bigl| g_{n,\mathbf{l}}^w-g_{0,\mathbf{l}}^w\bigr| &\le \sum_{k=2}^m \biggl(\prod_{j=2}^{k-1}\phi_n(t_{\alpha_{\mathbf{l}}(j)}-t_j)\biggr)\Bigl| \phi_n(t_{\alpha_{\mathbf{l}}(k)}-t_k)-\phi_\infty(t_{\alpha_{\mathbf{l}}(k)}-t_k)\Bigr|  \biggl(\prod_{j=k}^{m}\phi_n(t_{\alpha_{\mathbf{l}}(j)}-t_j)\biggr)
    \\
    &\lesssim \sum_{k=2}^m \biggl(\prod_{\substack{j=2 \\ j\not = k}}^{m}\phi_\infty(t_{\alpha_{\mathbf{l}}(j)}-t_j)\biggr)\Bigl| \phi_n(t_{\alpha_{\mathbf{l}}(k)}-t_k)-\phi_\infty(t_{\alpha_{\mathbf{l}}(k)}-t_k)\Bigr|. 
\end{align*}
After integrating iteratively, by integrability of $\phi_\infty$ it follows that
\begin{align*}
    \sum_{\substack{w\in \mathcal{W}\\\ell(w)=N}} \int_0^T \int_0^{t_1}\dots\int_0^{t_{|w|-1}} \mathcal{E}_n^2 \,dt_{|w|}\dots dt_1\lesssim \triangle(n).
\end{align*}
\end{proof}

\subsection{Kernel Estimates}\label{Sec:Kernel_Estimates}

So far, we obtained a series of moment estimates in terms of (combinations of) $\star, \diamond, \square$ and $\triangle$. Recalling that $\phi_\infty(t)=t^{H-1/2}$, in this section, prove that the family of kernels 
 \begin{align}\label{Def:phi_n}
        \widehat\phi_n(t)
&= \sqrt{
\phi_\infty^2(t+n^{-\beta}) + c_1 \Bigl( \phi_\infty^2(t+n^{-\beta})-\phi_\infty^2(2n^{-\beta}) \Bigr)_+
}
\end{align}
for some $\beta>0$ minimizes the sum of these functionals. 

The term $\Delta(n)$ penalizes the ($L^1$-) distance between $\phi_n$ and the Volterra kernel $\phi_\infty(t)=t^{H-1/2}$. Hence we want to be close, at least outside a neighborhood of zero. The term $\star(n)$ penalizes the $L^4$ norm, which is dominated by the blow-up at the origin. Therefore we included the shift by $n^{-\beta}$. Finally, $\square(n)$ penalizes the distance of the covariances. 
To keep this distance small, we added the $c_1$-term, which adds mass at a neighborhood at $0$ while still allowing the $L^4$-norm to be controlled. 
    We choose the constant $c_1$  in such a way that for $t\ge n^{-\beta}$ the kernels satisfy
    \begin{align*}
        \int_0^t \widehat\phi_n(t-s)^2\,ds=\frac1{2H}(t+n^{-\beta})^{2H}.
    \end{align*}
    Direct calculation shows that for $t\ge n^{-\beta}$ we have
    \begin{align*}
        \int_0^t \widehat\phi_n(t-s)^2\,ds&=\frac1{2H}(t+n^{-\beta})^{2H} +\frac{n^{-2H\beta }}{2H}\Big(-1+c_1\bigl(2^{2H}-1-2^{2H-1}\bigr)\Bigr) .
    \end{align*}
    Therefore,  we set
    \begin{align}\label{Eq:C1}
        c_1=(2^{2H}-1-2^{2H}H)^{-1},
    \end{align}
     which is positive for $H \le 1$. 
     The following lemma confirms that our kernels $\{\widehat\phi_n\}_{n\in \mbN}$ are admissible.

\begin{lemma}
   Let $\widehat\phi_n$ be defined in \eqref{Def:phi_n}. Then Assumption \ref{Ass:Kernel} is satisfied, i.e.
\begin{enumerate}
    \item $0\le \widehat\phi_n(t) \le C\left(\frac{1}{n}+ t\right)^{H-\frac12}$
   \item $0 \le -\frac{d}{dt}\widehat\phi_n(t) \le C \left(\frac{1}{n}+ t\right)^{H-\frac32}$

   \item For $\theta>2$: $\sup\limits_{h \leq n^{-\theta}} \int_0^T |\widehat \phi'_n(h+s) - \widehat \phi'_n(s)|ds \to 0$

   \item $\| \widehat \phi_n - \phi_\infty\|_{L^2} \rightarrow 0. $ 
\end{enumerate}
\pagebreak[3]
\end{lemma}
\begin{proof}
   The first statement is satisfied by construction with constant $C=(1+c_1)$. 
   From construction it also follows that $\widehat\phi_n$ is decreasing. We can furthermore calculate the derivative for all $t \not = n^{-\beta}$ as
   \begin{align*}
       \widehat\phi_n'(t)=
       \begin{cases}
           (2H-1) \frac1{2 \widehat\phi_n(t)}(1+c_1) (t+n^{-\beta})^{2H-2}, & \text{for } t< n^{-\beta} \\
           (2H-1) \frac1{2 \widehat\phi_n(t)}(t+n^{-\beta})^{2H-2} & \text{for } t> n^{-\beta}
       \end{cases}
   \end{align*}
   The fact that also $\widehat\phi_n(t) \ge (\frac1n+t)^{H-1/2}$ confirms the second property.

   For property 3.\ we split the integral into three regions and recall that $\beta<1$, hence $h<n^{-\beta}$. First, it follows from the previous two statements, that
   \begin{align*}
       \int_{n^{-\beta}-h}^{n^{-\beta}} |\phi'_n(h+s) - \phi'_n(s)| \,ds \lesssim \int_0^{h} \Bigl(\frac1n+t\Bigr)^{H-3/2}\,ds \lesssim h n^{1/2-H}.
   \end{align*}
   By convexity of $\widehat\phi_n$ on $(0,n^{-\beta}]$, see Lemma \ref{Lem:phin_conv}, we know the monotonicity of the first derivative and hence
   \begin{align*}
       \int_0^{n^{-\beta}-h} |\phi'_n(h+s) - \phi'_n(s)| \,ds &= 
       \int_0^{n^{-\beta}-h} \phi'_n(h+s) - \phi'_n(s)\,ds
       \\
       &\lesssim  
       \int_0^{h} 
       \Bigl(\frac1n+t\Bigr)^{H-3/2}
       \,ds \lesssim h n^{1/2-H}.
   \end{align*}
   Similarly, Lemma \ref{Lem:phin_conv} implies 
   \begin{align*}
       \int_{n^{-\beta}}^T |\phi'_n(h+s) - \phi'_n(s)| \,ds &= 
       \int_{n^{-\beta}}^T
       \phi'_n(h+s) - \phi'_n(s)\,ds
       \\
       &\lesssim  
       \int_0^{h} 
       \Bigl(\frac1n+ t\Bigr)^{H-3/2}
       \,ds \lesssim h n^{1/2-H}.
   \end{align*}
   As $h \lesssim n^{-2}$, the sum of these intragls converges to $0$ as $n \rightarrow \infty$.
   The $L^2$ convergence follows from pointwise convergence to $\phi_\infty$ and dominated convergence due to the bound from the first condition.
\end{proof}

\begin{theorem}\label{Thm:phi_n_estim}
    Recall $\star, \diamond, \square$ and $\triangle$ from Equation \ref{Eq:Rate_Notation}. Let $\phi_\infty(t)=t^{H-1/2}$ and $\{\widehat\phi_n\}_{n \in \mbN}$ be the kernels defined in Equation \eqref{Def:phi_n}.
    Then  for $\beta \in [0,1]$ %
    \begin{align*}
        \star(n) &\lesssim n^{-1+\beta-4H\beta},
    \\
    \diamond (n)+  \square(n)+ \triangle(n) &\lesssim n^{-\beta(H+1/2)}.
    \end{align*}
    In particular, for $\beta=\frac{2}{3-6H}$ it follows that
    \begin{align*}
        \star(n)+\diamond(n)+\square(n)+\triangle(n)\lesssim
        \begin{cases}n^{-\frac13-\frac{4H}{3-6H}} & \text{if } H \in (0,1/4),\\
        n^{-1} \log(n) & \text{if } H =1/4,\\
n^{-1} &\text{if } H \in (1/4,1/2).
\end{cases}
    \end{align*}
\end{theorem}
\begin{proof}
    We estimate the different terms individually. 
    \paragraph{Estimating $\star(n)$.} 
    By construction, 
    $$\widehat\phi_n(t) \le (1+c_1)(t+n^{-\beta})^{H-1/2},$$ and so
    \begin{align*}
        \frac1n \int_0^T  \widehat\phi_n(T-s)^4 \,ds &\lesssim 
            \frac1n \int_0^T  (T-s+n^{-\beta})^{4H-2} \,ds \lesssim \begin{cases} n^{-1+\beta-4H\beta} & \text{if } H \in (0,1/4),\\
        n^{-1} \log(n) & \text{if } H =1/4,\\
n^{-1} &\text{if } H \in (1/4,1/2).
        \end{cases}
    \end{align*}
    
    Furthermore, in the case $H \not = 1/4$,  $\int_1^\infty \widehat\phi_\infty(t)^4\,dt<\infty$ implies
    \begin{align*}
        \frac1n \int_{-\infty}^0 \Bigl(\widehat\phi_n(T-s) -\widehat\phi_n(-s)\Bigr)^4\,ds &\lesssim \frac1n \int_{-1}^0 \widehat\phi_n(-s)^4\,ds + \frac1n \int_{-\infty}^{-1} \Bigl(\widehat\phi_n(T-s) -\widehat\phi_n(-s)\Bigr)^4\,ds 
        \\
        &\lesssim  n^{-1+\beta-4H\beta}+  \frac1n \int_{-\infty}^{-1} \phi_\infty(-s)^4\,ds
        \\ &\lesssim 
        n^{-1+\beta-4H\beta}\vee n^{-1}.
    \end{align*}
    If $H=1/4$ the same calculation reveals a bound of order $n^{-1}\log(n)$.
     
     \paragraph{Estimating $\diamond(n)$.} Combining Lemma \ref{Lem:Diamond1} and Lemma \ref{Lem:Diamond2} in the Appendix below we obtain
    \begin{align*}
        \Bigl| C_n(t,s)-C_\infty(t,s)\Bigr| \lesssim \Bigl((t-s)^{-1/2}+\phi_\infty(s)\Bigr)n^{-\beta(H+1/2)}.
    \end{align*}
    By integrating and the boundedness of the beta-function it follows that
    \begin{align*}
        \diamond(n)\lesssim n^{-\beta(H+1/2)} \sup_{t\in [0,T]} \int_0^t \phi_\infty(t-s)\Bigl( (t-s)^{-1/2}+\phi_\infty(s)\Bigr)\,ds\lesssim n^{-\beta(H+1/2)}.
    \end{align*}
    
    \paragraph{Estimating $\square(n)$.} Let $s\in [n^{-\beta},T]$. Direct integration and using the definition of $c_1$ in \eqref{Eq:C1} shows
    \begin{align*}
        \biggl|\int_0^s \widehat\phi_n^2(s-r)-\phi_\infty^2(s-r)\,dr\biggr|=\frac1{2H}\Bigl|\bigl(s+n^{-\beta})^{2H}-s^{2H} \Bigr|\lesssim n^{-\beta} s^{2H-1}.
    \end{align*}
    Since $\int_0^{s} \phi_\infty^2(s-r)\,dr \lesssim n^{-2H\beta}$ for $s\le n^{-\beta}$ it follows from the boundedness of the Beta function that
    \begin{align*}
        & \int_0^t \phi_\infty(t-s) \Bigl| 
    \int_0^s \widehat\phi_n^2(s-r) -\phi_\infty^2(s-r)\,dr\Bigr| \,ds  
    \\
    & ~\int_{n^{-\beta}}^t \phi_\infty(t-s) \Bigl| 
    \int_0^s \widehat\phi_n^2(s-r)-\phi_\infty^2(s-r)\,dr\Bigr| \,ds
    \\
    &\qquad +  \lesssim \int_0^{n^{-\beta}} \phi_\infty(t-s) \Bigl| 
    \int_0^s \widehat\phi_n^2(s-r)-\phi_\infty^2(s-r)\,dr\Bigr| \,ds
    \\
    &\lesssim \int_{n^{-\beta}}^t \phi_\infty(t-s) n^{-\beta} s^{2H-1} \,ds +  \int_0^{n^{-\beta}} \phi_\infty(t-s) n^{-2H\beta} \,ds
    \\
    &\lesssim n^{-\beta(3H+1/2)}+n^{-\beta}.
    \end{align*}
    Lemma \ref{Lem:Diamond2} implies that
    \begin{align*}
        \int_0^t \phi_\infty(t-s) \Bigl| 
    \int_{-\infty}^0 &\Bigl(\widehat\phi_n(s-r)-\widehat\phi_n(-r)\Bigr)^2-\Bigl(\phi_\infty(s-r)-\phi_\infty(-r)\Bigr)^2dr\Bigr| \,ds 
    \\
    &\lesssim \int_0^{n^{-\beta}} \phi_\infty(t-s)  \,ds
    + \int_{n^{-\beta}}^t\phi_\infty(t-s) n^{-\beta(H+1/2)} \phi_\infty(s) \,ds
    \\
    &\lesssim n^{-\beta(H+1/2)}.
    \end{align*}
    Combining these estimates we observe that $\square(n)\lesssim n^{-\beta(H+1/2)}$.

    \paragraph{Estimating $\triangle(n)$.} Recall again that  $\widehat\phi_n(t)=(t+n^{-\beta})^{H-1/2}$ for $t \ge n^{-\beta}$. Therefore
    \begin{align*}
        \int_0^T \bigl|\widehat\phi_n(t)-\phi_\infty(t)\bigr|\,dt &\lesssim \int_0^{n^{-\beta}} \phi_\infty(t)\,dt 
        +
        \int_{n^{-\beta}}^T \bigl|\widehat\phi_n(t)-\phi_\infty(t)\bigr|\,dt 
        \\
        &\lesssim n^{-\beta(H+1/2)} + n^{-\beta} \int_{n^{-\beta}}^T t^{H-3/2} \,dt
        \\
        &\lesssim n^{-\beta(H+1/2)} .
    \end{align*}
\end{proof}

Armed with the preceding theorem, we can finish the proof of Theorem  \ref{Thm:Main_Thm}.

\medskip

\begin{proof}[Proof of Theorem \ref{Thm:Main_Thm}]
    Fix the Hurst parameter $H\in [0,1/2]$, some $t\in [0,T]$ and $N\in \mathbb{N}$. Take the family $\{\widehat\phi_n\}_{n \in \mathbb{N}}$ from Equation \eqref{Def:phi_n}.
    By triangle inequality, Theorem \ref{Thm:Cont_Diff} as well as Theorem \ref{Thm:Logprice_Diff} we observe that
    \begin{align*}
        \biggl|\mbE\Bigl[ \bigl(\Pn_t\bigr)^N\Bigr]-\mbE\Bigl[ \bigl(\Pcont_t\bigr)^N\Bigr]\biggr|\lesssim  \star(n)+\diamond(n)+\square(n)+\triangle(n)
    \end{align*}
    The statement then follows from Theorem \ref{Thm:phi_n_estim}.
\end{proof}

The following lemma proves the optimality of the above kernels, in the sense that the estimate in Equation \eqref{Eq:rhs_est} is sharp. %
\begin{lemma}\label{Lem:LowerBd}
    Let $\{\phi_n\}_{n \in \mbN}$ be any sequence of kernels that satisfy Assumption \ref{Ass:Kernel}. Then, 
    \begin{align*}
        \limsup_{n \rightarrow \infty} n^{\frac13+\frac{4H}{3-6H}}\Bigl(\star(n)+ \diamond(n)\Bigr)>0.
    \end{align*}
\end{lemma}
\begin{proof}
    To verify this proposition, we fix a parameter $\gamma>0$ to be specified later and study the quantity
    \begin{equation*}
        \mathfrak{C}=\limsup_{n \rightarrow \infty} n^{\gamma} \lambda\Bigl(\Bigl\{t:\phi_n(t)\ge \frac12\phi_\infty(t)\Bigr\}\cap[0,n^{-\gamma}]\Bigr),
    \end{equation*}
    where $\lambda$ denotes the usual Lebesgue measure. By construction $0 \le \mathfrak{C}\le 1$. 

    \noindent\textbf{Case 1:} $\mathfrak{C}=0$
    In this case, it follows that
    \begin{align*}
        \|\phi_n-\phi_\infty\|_{L^1} \gtrsim  \int_0^{n^{-\gamma}}
        \Bigl|\phi_n(t)-\phi_\infty(t)\Bigr|\,dt \gtrsim
        \int_0^{n^{-\gamma}} \phi_\infty(t)\,dt \gtrsim \bigl(n^{-\gamma}\bigr)^{H+1/2}.
    \end{align*}
    Therefore $\Delta(n)\gtrsim n^{-\gamma(H+1/2)}$.

    \noindent
    \textbf{Case 2:} $\mathfrak{C}>0$
    In this case, there is a subsequence $\{n_k\}_{k\in \mathbb{N}}$, such that 
    \begin{align*}
        \star(n_k) \ge C n_k^{-1} \int_{(1-\mathfrak{C})n_k^{-\gamma}}^{n_k^{-\gamma}} \phi_\infty^4(t)\,dt \gtrsim n_k^{-1} n_k^{-\gamma}\phi_\infty^4(n_k^{-\gamma})\approx n_k^{-1}n_k^{-\gamma(4H-1)}\approx n_k^{-1+\gamma(1-4H)}.  
    \end{align*}
    Optimizing over $\gamma$ in both cases yields the lower bound of order ${-\frac13}- \frac{4H}{3-6H}$, attained at $$\gamma= \frac2{3-6H}.$$
\end{proof}

\bibliographystyle{abbrvurl}
\bibliography{Reference}

\appendix
\section{Appendix}

\renewcommand{\theequation}{\thesection.\arabic{equation}}

\subsection{Deriving the Moment Formula}\label{Sec:AppMoments}

In this section we give a brief overview of the method from \cite{Friz_Salkeld_Wagenhofer_2025}. Recalling the recursion formula from Lemma \ref{Lem:Mom_Recursion} we want to combine all terms not being a power of $\Pn$ into one singular function. This motivates the following definition.

\begin{definition}
    \label{Def:IJ_Operators}
    Let $m,N\in \mbN$ and $\Delta_m^{\circ}$ be the open simplex as in Equation \eqref{Eq:Simplex}. Let $\mathcal{Q}^{m}$ be the set of functions $F:\mathbb{R}^{m}\times \Delta_{m}^\circ \rightarrow \mathbb{R}$ that are of the form $$F(x_1,\dots,x_m,t_1,\dots,t_m)=f_1(x_1,\dots,x_m)f_2(t_1,\dots,t_m),$$ with $f_1$ of exponential form, see Definition \ref{Def:Exp_Form}. Furthermore set $\mathcal{Q}=\bigcup_{m \in \mathbb{N}} \mathcal{Q}^m$. For $s<t_m$, $y \in \mathbb{R}$ and $n \in \mathbb{N}\cup\{0\}$, we define $\cI_n^N,\cJ_n^N:\mathcal{Q}^{m}\rightarrow \mathcal{Q}$  as
    \begin{align*}
        (\cI_n^N F)(x_1,\dots,x_m,y,t_1,\dots,t_m,s)&=\rho\sigma_p\sigma_v N  e^{ y} \sum_{j=1}^m \partial_{x_j} F(x_1,\dots,x_m,t_1,\dots,t_m) \phi_n(t_j-s),
         \intertext{and}
        (\cJ_n^N F)(x_1,\dots,x_m,y,t_1,\dots,t_m,s)&=\sigma_p^2\frac{N(N-1)}2 e^{2 y} F(x_1,\dots,x_m,t_1,\dots,t_m).
    \end{align*}
\end{definition}
In terms the notation of above definition the statement of Lemma \ref{Lem:Mom_Recursion} takes the following form.

 \begin{align*}
        \mbE\Bigl[ (\Pn_t)^k F\bigl(\Vn_{\mathbf{t}},\mathbf{t}\bigr) \Bigr]
        &=
        \int_0^t \mbE\Bigl[ (\Pn_s)^{k-1} \bigl(\cI_n^{k}(F)\bigr)\bigl(\Vn_{\mathbf{t}},\Vn_s,\mathbf{t},s\bigr) \Bigr] \,ds
        \\
        &\quad+ 
        \int_0^t \mbE\Bigl[ (\Pn_s)^{k-2}  \bigl(\cJ_n^{k} (F)\bigr)\bigl(\Vn_{\mathbf{t}},\Vn_s,\mathbf{t},s\bigr)\Bigr] \,ds
        \\
        &\quad +\mathcal{O}\bigl(\star(n)\bigr).
\end{align*}
Considering the r.h.s.\ we can apply Lemma \ref{Lem:Mom_Recursion} once more, to obtain the following recursion: 
 \begin{align*}
        \mbE\Bigl[ (\Pn_t)^k F\bigl(\Vn_{\mathbf{t}},\mathbf{t}\bigr) \Bigr]
        &=
        \int_0^t \int_0^s \mbE\Bigl[ (\Pn_r)^{k-2} \bigl(\cI_n^{k-1}\cI_n^{k}(F)\bigr)\bigl(\Vn_{\mathbf{t}},\Vn_s,\Vn_r,\mathbf{t},s,r\bigr) \Bigr] \,dr\,ds
        \\
        &\qquad + \int_0^t \int_0^s \mbE\Bigl[ (\Pn_s)^{k-3} \bigl(\cJ_n^{k-1}\cI_n^{k}(F)\bigr)\bigl(\Vn_{\mathbf{t}},\Vn_s,\Vn_r,\mathbf{t},s,r\bigr) \Bigr] \,dr \,ds
        \\
        &\quad+ 
        \int_0^t \int_0^s \mbE\Bigl[ (\Pn_s)^{k-3}  \bigl(\cI_n^{k-1}\cJ_n^{k} (F)\bigr)\bigl(\Vn_{\mathbf{t}},\Vn_s,\Vn_r,\mathbf{t},s,r\bigr) \Bigr] \,dr \,ds
        \\
         &\qquad+ 
        \int_0^t \int_0^s \mbE\Bigl[ (\Pn_s)^{k-4}  \bigl(\cJ_n^{k-2}\cJ_n^{k} (F)\bigr)\bigl(\Vn_{\mathbf{t}},\Vn_s,\Vn_r,\mathbf{t},s,r\bigr) \Bigr] \,dr \,ds
         \\
        &\quad +\mathcal{O}\bigl(\star(n)\bigr).
\end{align*}

We observe that each application of $\cI$ corresponds to a reduction of the power of $\Pn$ by one, and each application $\cJ$ reduces the power by two. Furthermore, when iterating this procedure we obtain (iterated) integrals of powers of $\Pn$ and iterated applications of $\cI$ and $\cJ$. To also keep track of the constants we use the following structure.

\begin{definition}\label{Def:Wordset}
    Let $\cW$ be the set of all finite words with letters in $\{\wI,\wJ\}$. We denote by $|.|$ the number of letters, i.e.\ for $w=w_1\dots w_m$ we have $|w|=m$.
    
    We define an inhomogeneous length $\ell: \cW\to \mbN$ via  $\ell(\wI)=1$, $\ell(\wJ)=2$ and 
    \begin{equation*}
        \ell(w)=\sum_{j=1}^{|w|} \ell(w_j) 
    \end{equation*}
    and an embedding $\iota_n: \cW \to \mathcal{Q}$ as follows: If $w$ has a single letter $w=\wI$ or $w=\wJ$ then 
    \begin{equation*}
        \iota_n(\wI)=\cI_n^1
        \quad\mbox{or}\quad
        \iota_n(\wJ)=\cJ_n^2.
    \end{equation*}
    For $w=w_1\dots w_m$ we define
    \begin{equation*}
        \iota_n(w)=
        \begin{cases}
            \iota_n(w_1\dots w_{m-1}) \circ \wI^{\ell(w)}_n \quad &\quad \mbox{if} \quad w_m = I
            \\
            \iota_n(w_1\dots w_{m-1}) \circ \cJ^{\ell(w)}_n \quad &\quad \mbox{if} \quad w_m = J. 
        \end{cases}
    \end{equation*}
\end{definition}

\begin{remark}
    Let $w \in \cW$ be a word with last letter $w_{|w|}=I$. Denote by $1$ the constant one function. Then, by definition of the function $\iota_n$ it holds that $\iota_n(w)(1)=\iota_n(w_1\dots w_{|w|-1})(\cI_n^{\ell(w)} 1)$. From Definition \ref{Def:IJ_Operators} it follows that $\cI_n^{\ell(w)} 1\equiv0$.
\end{remark}
Having fixed the notation we can state the first moment formula, derived in \cite[Theorem 3.13]{Friz_Salkeld_Wagenhofer_2025} and stating that
\begin{align*}
        \mbE\Bigl[ (\Pn_T)^N \Bigr]
            &=
        \sum_{\substack{w\in \cW \\ \ell(w)=N}} 
        ~ \int\limits_{\mathbf{t}\in \Delta^\circ_{|w|}}
        \mbE\bigl[
        \bigl( \iota_n(w)1\bigr)\bigl(\Vn_{t_1},\dots,\Vn_{t_{|w|}},t_1,\dots,t_{|w|}\bigr) 
        \bigr] 
        d\mathbf{t}
       +\mathcal{O}\bigl(\star(n)\bigr).
\end{align*}
Unfortunately, it is not immediately clear, how the function $\bigl( \iota_n(w)1\bigr)$ is defined. Definition \ref{Def:IJ_Operators} suggests that the function is of product form, depending on the space and time arguments. 

Even more, there is an explicit structure on how the time and space component looks like, as established in Proposition \ref{Prop:wrep} below. 
To derive this representation, we must first identify the positions of letters  $I$ and $J$ within a word $w$, motivating the following definition. 
\begin{definition}
    \label{Def:Afcts}
    For $w \in \cW$ such that $|w|=m$, let 
    \begin{equation*}
        N_J^w=\{j: w_j = \wJ\} \quad \mbox{and} \quad N_I^w=\{1,\dots,m\}\setminus N_J^w. 
    \end{equation*}
    Assume that $k\coloneqq |N_I^w|\ge 1$ and use an enumeration $N_I^w=\{j_1,\dots,j_k\}$ such that $j_1<\dots<j_k$. 
    
    We define 
    \begin{equation*}
        \cL^w= \Big\{ \mathbf{l}=(l_1,\dots,l_k) \in \mbN^{\times k}: l_1\le m-j_k,\dots,l_k\le m-j_1 \Big\}.
    \end{equation*}
    For $\mathbf{l}\in \cL^w$, we define $\alpha_\mathbf{l}: \{1,\dots,m\} \rightarrow \{0,1,\dots,m\}$ such that 
    \begin{align*}
        &\alpha_\mathbf{l}(m-j_{k-i}+1)=l_{i} 
        \quad \text{for } i=1,\dots,k-1,
        \\
        &\alpha_\mathbf{l}(j)=0 \quad \forall j \in \{m-j+1:j \in N_J^w\}. 
    \end{align*}
    If $k=0$ we define $\cL^w=\{1\}$ and set $\alpha_1\equiv 0$.
\end{definition}

\begin{proposition}
    \label{Prop:wrep}
    Fix $m\in \mbN$ and $w\in \cW$ such that $|w|=m$. For $\mathbf{l} \in \cL^w$, we define
    $\psi_{\mathbf{l}}: \mbR^{\times m} \to \mbR$ to be 
    \begin{equation*}
        \psi_{\mathbf{l}}(x_1,\dots,x_m)= 
         \prod_{\substack{l \in N_J^w}}e^{2 x_{m-l+1}} \prod_{\substack{l \in N_I^w}}e^{ x_{m-l+1}} .
    \end{equation*}
    Then there exist a constant $C_w$ only depending on $w$ and $\rho$ such that
    \begin{align*}
        (\iota_n(w)1)&(x_1,\dots,x_m,t_1,\dots,t_m)
        \\
        =&C_w \sum_{\mathbf{l} \in \cL^w} \psi_{\mathbf{l}}(x_1,\dots,x_m)\phi_n(t_{\alpha_{\mathbf{l}}(2)}-t_2) \dots \phi_n(t_{\alpha_{\mathbf{l}}(m)}-t_m).
    \end{align*}
\end{proposition}
\begin{proof}
    The proof is already done in \cite[Proposition 3.13]{Friz_Salkeld_Wagenhofer_2025} with $\phi_n(t-s)=K(t,s)$ and $f(x)=e^{ x}$. It only remains to show that
    \begin{align*}
        \partial_{x_{l_1}}\dots \partial_{x_{l_{k}}} &
        \prod_{\substack{l \in N_I^w}}f(x_{m-l+1}) \prod_{\substack{l \in N_J^w}}f^2(x_{m-l+1})
        \\
        &= \prod_{\substack{l \in N_J^w}}f^2(x_{m-l+1}) \prod_{\substack{l \in N_I^w}}f(x_{m-l+1}).
    \end{align*}
    But this follows immediately from the special form of $f$ and the fact that $\ell(w_{m-l+1})=2$ if $l \in N_J^w$ and $\ell(w_{m-l+1})=1$ if $l \in N_I^w$.
\end{proof}

Lastly, for estimating the difference of $\Pcont$ and $\Pncont$ we need to compare certain expectation of functions of Gaussian random variables. For this we use the following lemma.
\begin{lemma}[\cite{Friz_Salkeld_Wagenhofer_2025}, Lemma 2.2]
    \label{Lem:DerivSigma2}
    Let $d\in \mathbb{N}$, $\Sigma: [0,T]\rightarrow \mathbb{R}^{d \times d}$ a continuously differentiable, matrix valued map, such that $\Sigma(t)$ is positive semi-definite for all $t \in [0,T]$. Let $W(t)$ be a $d$-dimensional centered Gaussian random variable with covariance $\Sigma(t)$, and let  $g \in C^2$ be at most exponential growth. Define $\varphi(t)=\mbE\bigl[ g(W(t))\bigr]$.
    Then
    \begin{gather*}
        \partial_t \varphi(t) = \sum_{k,l=1}^d \frac 12 \partial_t \Sigma(t)_{k,l} \mbE \bigl[\partial_k \partial_l g(W(t))\bigr]
    \end{gather*}
\end{lemma}

\begin{proof}[Proof of Theorem \ref{Thm:MomRep}]
    By Lemma \ref{Lem:Mom_Recursion} 
     \begin{align*}
        \mbE\Bigl[ (\Pn_t)^k F\bigl(\Vn_{\mathbf{t}},\mathbf{t}\bigr) \Bigr]
        &=
        \int_0^t \mbE\Bigl[ (\Pn_s)^{k-1} \bigl(\cI_n^{k}(F)\bigr)\bigl(\Vn_{\mathbf{t}},\Vn_s,\mathbf{t},s\bigr) \Bigr] \,ds
        \\
        &\quad+ 
        \int_0^t \mbE\Bigl[ (\Pn_s)^{k-2}  \bigl(\cJ_n^{k} (F)\bigr)\bigl(\Vn_{\mathbf{t}},\Vn_s,\mathbf{t},s\bigr)\Bigr] \,ds
        \\
        &\quad +\mathcal{O}\bigl(\star(n)\bigr)
    \end{align*}
    which corresponds to \cite[Lemma 3.4]{Friz_Salkeld_Wagenhofer_2025} with the additional $\mathcal{O}(\star(n))$ term. Redoing the proof of \cite[Theorem 3.13]{Friz_Salkeld_Wagenhofer_2025} implies that
    \begin{align*}
        \mbE\Bigl[ (\Pn_T)^N \Bigr]
            &=
        \sum_{\substack{w\in \cW \\ \ell(w)=N}} 
        ~ \int\limits_{\mathbf{t}\in \Delta^\circ_{|w|}}
        \mbE\bigl[
        \bigl( \iota_n(w)1\bigr)\bigl(\Vn_{t_1},\dots,\Vn_{t_{|w|}},t_1,\dots,t_{|w|}\bigr) 
        \bigr] 
        d\mathbf{t}
       +\mathcal{O}\bigl(\star(n)\bigr).
    \end{align*}
    The decomposition of $ \bigl( \iota_n(w)1\bigr)$ then follows from Proposition \ref{Prop:wrep}.

    The second and third representations follow immediately from  \cite[Theorem 3.9]{Friz_Salkeld_Wagenhofer_2025} and Proposition~\ref{Prop:wrep}.
\end{proof}

\subsection{Kernel Calculations}

This appendix summarizes three auxiliary results on the kernels $\widehat \phi_n$. We start with the following convexity result.  

 \begin{lemma}\label{Lem:phin_conv}
    For all $n \in \mathbb{N}$ the kernel $\widehat\phi_n$ is convex on $[0,n^{-\beta}]$ as well as on $[n^{-\beta},T]$.
\end{lemma}
\begin{proof}
    We start with the second statement. Recall that $\widehat\phi_n(t)=(n^{-\beta}+t)^{H-1/2}$ for $t \ge n^{-\beta}$. Convexity follows immediately.

    For $t \in (0,n^{-\beta}]$ we write $\widehat\phi_n(t)=\sqrt{f_n(t)+g_n(t)}$ for functions $f_n(t)=(1+c_1)(n^{-\beta}+t)^{2H-1}$ and the constant function $g_n(t)=-c_1 2^{2H-1}n^{-\beta(2H-1)}$.
    Then it follows that
    \begin{align*}
        \partial_t \widehat\phi_n=\frac1{2\widehat\phi_n}\Bigl(\partial_t f_n + \partial_t g_n\Bigr)=\frac1{2\widehat\phi_n}\partial_t f_n
    \end{align*}
    as well as
    \begin{align*}
        \partial_t^2 \widehat\phi_n= -\frac1{4\widehat\phi_n^3} \bigl(\partial_t f_n\bigr)^2 + \frac1{2\widehat\phi_n} \partial_t^2 f_n.
    \end{align*}
    For convexity we only need to verify that $ \partial_t^2 \widehat\phi_n$ is non-negative. 
    By construction, 
    \[
        \partial_tf_n=(1+c_1)(2H-1)(n^{-\beta}+t)^{2H-2} \quad \mbox{and}  \quad \partial_t^2 f_n=(1+c_1)(2H-1)(2H-2)(n^{-\beta}+t)^{2H-3}. 
    \]    
    Using the bound $\widehat\phi_n^2 \ge (1+c_1)(t+n^{-\beta})^{2H-1}$ and omitting the positive factor $1/(2\widehat\phi_n)$ it suffices to show that
    \begin{align*}
        -\frac{(1+c_1)^2(2H-1)^2(n^{-\beta}+t)^{4H-4}}{2(1+c_1)(n^{-\beta}+t)^{2H-1}}+(1+c_1)(2H-1)(2H-2)(n^{-\beta}+t)^{2H-3} \ge 0.
    \end{align*}
    This is indeed the case, as this expression is reduced to showing
    \begin{align*}
        \frac12(1-2H) \le (2-2H),
    \end{align*}
    which is obviously satisfied if $H\le 3/2$.
\end{proof}

The next two results establish estimates for the covariance function 
\begin{align*}
    C_n(t,s)&=\int_0^{s\wedge t} \widehat\phi_n(t-r)\widehat\phi_n(s-r)\,dr  + \int_{-\infty}^0 \Bigl(\widehat\phi_n(t-r) -\widehat\phi_n(-r)\Bigr) \Bigl(\widehat\phi_n(s-r) -\widehat\phi_n(-r)\Bigr)\,dr
    \\
    &\eqqcolon C_n^1(t,s)+C_n^2(t,s).
\end{align*}

\begin{lemma}\label{Lem:Diamond1}
    Let $t>s>n^{-\beta}$ and $H\not = 1/4$, then
    \begin{align*}
        \Bigl| C_n^1(t,s)-C_\infty^1(t,s)\Bigr|\lesssim (t-s)^{-1/2}\Bigl(n^{-\beta(2H+1/2)} \vee n^{-1} \Bigr).
    \end{align*}
    If $H=1/4$ it holds that
    \begin{align*}
        \Bigl| C_n^1(t,s)-C_\infty^1(t,s)\Bigr|\lesssim (t-s)^{-1/2}n^{-\beta}\log(n).
    \end{align*}
\end{lemma}
\begin{proof}
Using that $\widehat\phi_n \lesssim \phi_\infty$, uniformly in $n$ and exploiting monotonicity we see
    \begin{align*}
        \biggl|\int_{s-n^{-\beta}}^s \widehat\phi_n(t-r)\widehat\phi_n(s-r)-&\phi_\infty(t-r)\phi_\infty(s-r)\,dr\biggr|\lesssim \int_{s-n^{-\beta}} ^s\phi_\infty(t-r)\phi_\infty(s-r)\,dr
        \\
        &\lesssim (t-s)^{-1/2} \int_{s-n^{-\beta}}^{s}(s-r)^{2H-1/2}\,dr \lesssim (t-s)^{-1/2} n^{-\beta(2H+1/2)}.
    \end{align*}
    By triangle inequality  and the fact that $\widehat\phi_n\lesssim \phi_\infty$ we get that
    \begin{align*}
        \biggl|\int_0^{s-n^{-\beta}} \widehat\phi_n(t-r)\widehat\phi_n(s-r)-&\phi_\infty(t-r)\phi_\infty(s-r)\,dr\biggr| \\
        &\lesssim \int_0^{s-n^{-\beta}} \bigl|\widehat\phi_n(t-r)-\phi_\infty(t-r)\bigr|\phi_\infty(s-r)\,dr
        \\
        &\qquad +\int_0^{s-n^{-\beta}} \phi_\infty(t-r)\bigl|\widehat\phi_n(s-r)-\phi_\infty(s-r)\bigr|\,dr
    \end{align*}
    Note that $\widehat\phi_n(t)=(t+n^{-\beta})^{H-1/2}$ for $t > n^{-\beta}$. 
    Assume that $H\not=1/4$, then it follows that
    \begin{align*}
        \int_0^{s-n^{-\beta}} \bigl|\widehat\phi_n(t-r)-\phi_\infty(t-r)\bigr|\phi_\infty(s-r)\,dr
        &=
        \int_0^{s-n^{-\beta}} \bigl|(t-r+n^{-\beta})^{H-1/2}-(t-r)^{H-1/2}\bigr|\phi_\infty(s-r)\,dr
        \\
        &\lesssim n^{-\beta} \int_0^{s-n^{-\beta}} (t-r)^{H-3/2} \phi_\infty(s-r)\,dr
        \\
        &\lesssim n^{-\beta} (t-s)^{-1/2}\int_0^{s-n^{-\beta}}  (s-r)^{H-1} \phi_\infty(s-r)\,dr
        \\
        &\lesssim n^{-\beta} (t-s)^{-1/2} \Bigl(1\vee n^{-\beta(2H-1/2)}\Bigr)
        \\
        &\lesssim n^{-\beta\bigl((2H+1/2)\wedge1\bigr)} (t-s)^{-1/2}.
    \end{align*}
    Here we used again monotonicity and the explicit form of $\phi_\infty$. Similarly it follows that
    \begin{align*}
        \int_0^{s-n^{-\beta}} \phi_\infty(t-r)\bigl|\widehat\phi_n(s-r)-\phi_\infty(s-r)\bigr|\,dr&\lesssim n^{-\beta}\int_0^{s-n^{-\beta}} \phi_\infty(t-r)(s-r)^{H-3/2}\,dr
        \\
        &\lesssim n^{-\beta}(t-s)^{-1/2}\int_0^{s-n^{-\beta}} (s-r)^{2H-3/2}\,dr
        \\
        &\lesssim n^{-\beta\bigl((2H+1/2)\wedge1\bigr)} (t-s)^{-1/2}.
    \end{align*}
    Combining these estimates yields the statement for $H\not=1/4$.
    If $H=1/4$ then the same calculations reveal estimates of the form $n^{-\beta} \log(n) (t-s)^{-1/2}$.

\end{proof}

\begin{lemma}\label{Lem:Diamond2}
    Let  $t>s>n^{-\beta}$, then
    \begin{align*}
        \Bigl| C_n^2(t,s)-C_\infty^2(t,s)\Bigr|\lesssim \phi_\infty(t)n^{-\beta(H+1/2)}.
    \end{align*}
\end{lemma}
\begin{proof}

    It follows by triangle inequality and the bound $\widehat\phi_n \lesssim \phi_\infty$ for $u\in [0,t]$ that
    \begin{equation}\label{Eq:CnEstim}
    \begin{aligned}
        \biggl|\int_0^{1}\widehat\phi_n(t+r)\widehat\phi_n(u+r)-\phi_\infty(t+r)\phi_\infty(u+r) \,dr\biggr|&\lesssim \int_{n^{-\beta}}^1 |\widehat\phi_n(t+r)-\phi_\infty(t+r)|\phi_\infty(u+r)\,dr
        \\
        &+\int_{n^{-\beta}}^1 |\widehat\phi_n(u+r)-\phi_\infty(u+r)|\phi_\infty(t+r)\,dr
        \\
        &+\int_0^{n^{-\beta}} \phi_\infty(t+r)\phi_\infty(u+r)\,dr.
        \end{aligned}
    \end{equation}
    The last integral is bounded by $\phi_\infty(t)n^{-\beta(H+1/2)}$ by direct integration. 
    Recall that $\widehat\phi_n(s)=\phi_\infty(n^{-\beta}+s)$ for $s \ge n^{-\beta}$. Taking derivatives tells us that 
    \begin{align*}
        \int_{n^{-\beta}}^1 |\widehat\phi_n(u+r)-\phi_\infty(u+r)|\phi_\infty(t+r)\,dr\lesssim n^{-\beta} \phi_\infty(t) \int_{n^{-\beta}}^1 (u+r)^{H-3/2}\,dr \lesssim n^{-\beta(H+1/2)}\phi_\infty(t).
    \end{align*}
    Furthermore
    \begin{align*}
        \int_{n^{-\beta}}^1 |\widehat\phi_n(t+r)-\phi_\infty(t+r)|\phi_\infty(u+r)\,dr&\lesssim n^{-\beta}  \int_{n^{-\beta}}^1 (t+r)^{H-3/2} (u+r)^{H-1/2}\,dr
        \\
        &\lesssim  n^{-\beta}  \phi_\infty(t)\int_{n^{-\beta}}^1 r^{H-3/2}\,dr   \lesssim n^{-\beta(H+1/2)}\phi_\infty(t).
    \end{align*}
    Extending 
    \begin{align*}
         \int_0^1\bigl(\widehat\phi_n(t+r) -\widehat\phi_n(r)\bigr) &\bigl(\widehat\phi_n(s+r) -\widehat\phi_n(r)\bigr)\,dr\\
         & = \int_{0}^{1} \widehat\phi_n(t+r)\widehat\phi_n(s+r)-\widehat\phi_n(t+r)\widehat\phi_n(r)-\widehat\phi_n(s+r)\widehat\phi_n(r)+\widehat\phi_n^2(r)\,dr,
    \end{align*}
   we see that all differences are of the form as in Equation \ref{Eq:CnEstim}.
        If $r > 1$  it follows from Taylor's formula that 
    \begin{align*}
        \Bigr|\Bigl(\widehat\phi_n(t+r) -\widehat\phi_n(r)\Bigr) \
        -
        \Bigl(\phi_\infty(t+r) -\phi_\infty(r)\Bigr) \Bigl|\lesssim n^{-\beta}  r^{H-3/2}.
    \end{align*}
    By boundedness of $\bigl(\widehat\phi_n(t+r) -\widehat\phi_n(r)\bigr)$ uniformly for $t \in [0,T]$ and $r\ge 1$ it follows that
    \begin{align*}
        &\int_1^\infty \biggl(\bigl(\widehat\phi_n(t+r) -\widehat\phi_n(r)\bigr) \bigl(\widehat\phi_n(s+r) -\widehat\phi_n(r)\bigr) - \bigl(\widehat\phi_\infty(t+r) -\widehat\phi_\infty(r)\bigr) \bigl(\widehat\phi_\infty(s+r) -\widehat\phi_\infty(r)\bigr)\biggr)\,dr
        \\
        &\lesssim \int_1^\infty \ n^{-\beta}  r^{H-3/2} \,dr \lesssim n^{-\beta} \lesssim n^{-\beta(H+1/2)}\phi_\infty(t).
    \end{align*}
    Recalling the definition of $C_2(s,t)$ the claim follows.
\end{proof}

\end{document}